\documentclass[sigconf, authorversion=true, nonacm=true]{acmart}
\AtBeginDocument{%
  }

\setcopyright{acmlicensed}
\copyrightyear{2018}
\acmYear{2018}
\acmDOI{XXXXXXX.XXXXXXX}
\acmConference[Conference acronym 'XX]{Make sure to enter the correct
  conference title from your rights confirmation email}{June 03--05,
  2018}{Woodstock, NY}
\acmISBN{978-1-4503-XXXX-X/2018/06}

\usepackage[table]{xcolor}
\usepackage{multirow}
\usepackage{thmtools}
\usepackage{caption}
\usepackage{outlines}
\usepackage{subcaption}
\usepackage{tabularx}
\usepackage{array}
\usepackage[figuresright]{rotating}
\newcolumntype{Y}{>{\centering\arraybackslash}X}
\usepackage{booktabs} 
\usepackage{hvfloat} 
\usepackage{stmaryrd} 

\usepackage[edges]{forest}
\usetikzlibrary{arrows.meta,shadows.blur, fit}
\usetikzlibrary{decorations.pathmorphing, calligraphy}
\usepackage{graphicx}      
\usepackage[normalem]{ulem}
\forestset{%
  colour me out/.style={
    draw=darkgray,
    rounded corners
  },
  rect/.append style={},
  dir tree switch/.style args={at #1}{%
    for tree={
      edge=-Latex,
      font=\ttfamily,
      fit=rectangle,
    },
    where level=#1{
      for tree={
        folder,
        grow'=0,
        s sep=1pt,  
        l sep=12pt,  
      },
      delay={child anchor=north},
    }{},
    before typesetting nodes={
      for tree={
        content/.wrap value={\strut ##1},
      },
      if={isodd(n_children("!r"))}{
        for nodewalk/.wrap pgfmath arg={{fake=r,n=##1}{calign with current edge}}{int((n_children("!r")+1)/2)},
      }{},
    },
  },
}

\usepackage{mathtools}
\usepackage[T1]{fontenc}
\usepackage{listings}
\usepackage[table]{xcolor}
\usepackage{tableaucolors}
\usepackage{amsmath}
\usepackage{amsfonts}
\usepackage{comment}
\usepackage{tabularx}
\usepackage{multirow}
\usepackage{soul}
\usepackage[normalem]{ulem}
\usepackage[colorinlistoftodos,prependcaption]{todonotes}

\newcommand{\mycaption}[2]{\caption{\normalfont\textbf{#1}~#2}}

\usepackage{amssymb}
\usepackage{enumitem}
\usepackage{multirow}

\usepackage{graphicx}
\theoremstyle{definition}
\usepackage{amsthm}

\makeatletter
\def\thm@space@setup{\thm@preskip=0pt
\thm@postskip=0pt}
\makeatother
\newtheoremstyle{newstyle}      
{} 
{} 
{\mdseries} 
{} 
{\bfseries} 
{.} 
{ } 
{} 

\theoremstyle{newstyle}

\newtheorem{prop}{Proposition}
\newtheorem{definition}{Definition}

\makeatletter
\newenvironment{pf}[1][\proofname]{\par
  \pushQED{\qed}%
  \normalfont \topsep0\p@\relax
  \trivlist
  \item[\hskip\labelsep\itshape
  #1\@addpunct{.}]\ignorespaces
}{%
  \popQED\endtrivlist\@endpefalse
}
\makeatother

\newcommand{\revwayne}[1]{\textcolor{black}{#1}}

\usepackage{hyperref}
\usepackage[most]{tcolorbox}
\newtcolorbox[auto counter, number within=section]{mybox}[2][]{colback=C7!10, colframe=C7, title={\textbf{#2}}, fonttitle=\small, title style={right}, top=0.5pt, bottom=0.5pt,boxsep=0.5pt, left=1pt, right=1pt, before=,after=, #1}
\newtcolorbox[auto counter, number within=section]{myboxNoTitle}[1][]{%
colback=C7!10,
colframe=C7,
top=0.5pt,
bottom=0.5pt,
boxsep=0.5pt,
left=1pt,
right=1pt,
before=,
after=,
#1
}
\lstdefinestyle{mystyle}{
    basicstyle=\small\ttfamily,
    stepnumber=1,
    numbersep=10pt,
    showstringspaces=false,
    breaklines=true,
    frame=single,
    rulecolor=\color{black},
    keywordstyle=\color{C0},
    commentstyle=\color{C7},
    stringstyle=\color{C3},
morekeywords={@prefix, a, sh:NodeShape, sh:targetClass, sh:property, sh:path, sh:qualifiedValueShape, sh:qualifiedMinCount, sh:class},
    literate={sh:}{{\textcolor{C2}{sh:}}}3
             {brick:}{{\textcolor{C1}{brick:}}}5
             {ex:}{{\textcolor{C5}{ex:}}}3, 
    lineskip=-0.5ex, 
    aboveskip=1pt, 
    belowskip=1pt,
}

\begin{document}

\title{A Formal Diagnostic Framework for Graph Repair Systems}

\author{Tung-Wei Lin}
\affiliation{%
  \institution{UC Berkeley}
  \city{Berkeley}
  \country{USA}
}

\author{Gabe Fierro}
\affiliation{%
  \institution{Colorado School of Mines}
  \city{Golden}
  \country{USA}
}

\author{Han Li}
\author{Tianzhen Hong}
\affiliation{%
  \institution{Lawrence Berkeley National Laboratory}
  \city{Berkeley}
  \country{USA}
}

\author{Pierluigi Nuzzo}
\author{Alberto Sangiovanni-Vincentelli}
\affiliation{%
  \institution{UC Berkeley}
  \city{Berkeley}
  \country{USA}
}

\begin{abstract}
Graph-based data representations excel at capturing relations among entities, e.g., for semantic web, artificial intelligence, or database applications, and their utility depends critically on their quality, usually enforced via \emph{graph schemas}. These schemas are often violated, e.g., during data ingestion, calling for repair actions to achieve compliance. This paper presents a formal diagnostic framework for graph repair methods. Current methods often lack rigor and generality, since they rely on \textit{ad hoc} test datasets. Our method starts, instead, from a compliant graph and uses an abstract rewriting system to systematically introduce schema violations while ensuring provable coverage of the constraints in the schema. We apply the framework to several repair methods, including those based on large language models, and assess their performance across multiple metrics. Results show that our framework effectively differentiates the performance of different methods across various constraints, highlighting the importance of fine‑grained, systematic tools for graph repair.
\end{abstract}

\maketitle

\begin{figure}[t]
    \captionsetup{aboveskip=-0pt, belowskip=-0pt}
    \setlength{\belowcaptionskip}{-1pt}
    \centering
    \begin{tikzpicture}[scale=1, >=Latex, line cap=round]
\colorlet{beforecolor}{cyan}
\colorlet{startcolor}{teal}     
\colorlet{endcolor}{magenta}
\tikzset{
  beforeedge/.style={draw=beforecolor,fill=beforecolor,thick,->},
  startedge/.style={draw=startcolor,fill=startcolor,thick,->},
  endedge/.style={draw=endcolor,fill=endcolor,thick,->},
}

\begin{scope}[every node/.style={ellipse,draw,font=\ttfamily\small,minimum width=0.3cm, minimum height=0.4cm,inner sep=0.1pt},shift={(0,0)}]

    \node[] (b) at (-1,0) {\shortstack{submission\\phase1}};
    \node[] (b2) at (-3.3,0) {Jan 10};
    \node[] (d1) at (1.5,0) {\shortstack{reviewing\\phase1}};
    \node[] (d11) at (3.8,0.2) {Jan 12};
    \node[] (d12) at (3.8,-0.3) {Feb 10};

    \node[] (D) at (-1,-1) {\shortstack{submission\\phase2}};
    \node[fill=C6!20] (D2) at (-3.3,-1) {Oct 10};
    \node[] (d3) at (1.5,-1) {\shortstack{reviewing\\phase2}};
    \node[fill=C6!20] (d31) at (3.8,-0.8) {Oct 8};
    \node[] (d32) at (3.8,-1.3) {Nov 10};
    
\end{scope}

    
\path [beforeedge] (b) edge (d1);
\path [beforeedge] (D) edge (d3);


\path [startedge] (d1) edge (d11);
\path [startedge] (d3) edge (d31);

\path [endedge] (b) edge (b2);
\path [endedge] (D) edge (D2);

\path [endedge] (d1) edge (d12);
\path [endedge] (d3) edge (d32);

\begin{scope}[every node/.style={font=\ttfamily\footnotesize},shift={(0.5,-1.9)}]
  \node[font=\ttfamily] (lbl2) at (-2.6,0.2) {before};
  \draw[beforeedge,->] (-3.5,0) -- (-1.8,0);
  \node[font=\ttfamily] (lbl3) at (-0.1,0.2) {start};
  \draw[startedge,->] (-1,0) -- (0.7,0);
  \node[font=\ttfamily] (lbl2) at (2.3,0.2) {end};
  \draw[endedge,->] (1.5,0) -- (3.2,0);
\end{scope}

\node[anchor=north, align=left, draw, rounded corners, fill=gray!10, text width=7.8cm, font=\small]
      at (0.4,-2) {%
      \textbf{Schema $\mathcal{S}$:} for any node $u$,
      
      $u\models\underbrace{\geq_1\texttt{before}.\top}_{sel}\implies u\models \underbrace{\texttt{lessThan}(\texttt{end}, \texttt{before}\cdot\texttt{start})}_\varphi$

      };
\end{tikzpicture}
     \mycaption{Example of a schema violation.}{The schema requires that a paper submission phase end before reviewing phase start. \texttt{submission phase2} violates this constraint. The schema syntax is explained in Sec.~\ref{sec:prelim}.}
    \label{fig:mot_ex}
\end{figure}

\section{Introduction}
\label{sec:intro}

Graphs are structured representations of entities and their relationships, widely used to encode domain knowledge in applications such as smart building management~\cite{balaji2016brick}, genome clustering~\cite{yuDelineationMetabolicGene2016}, and large-scale knowledge bases such as Wikipedia~\cite{vrandecicWikidataFreeCollaborative2014}. They also support emerging use cases in automatic software configuration and remote data access~\cite{fierro2020mortar,rodriguez-muroOntologyBasedDataAccess2013}.

Graphs are populated through manual curation~\cite{bollacker2008freebase}, automated conversion from existing databases, and extraction from unstructured text~\cite{bi2024codekgc}. Their correctness and completeness are critical, as downstream tasks such as query generation~\cite{farzana2023knowledge} and question answering~\cite{xu2024retrieval} depend on accurate facts. Large resources like Wikidata~\cite{vrandecicWikidataFreeCollaborative2014} and YAGO~\cite{suchanekYAGO45Large2024} link real-world entities to attributes. Cyber-physical ontologies such as Brick~\cite{balaji2016brick} and RealEstateCore~\cite{hammarRealEstateCoreOntology} enable digital twins and automated control. Reference ontologies like QUDT~\cite{qudt} document physical constants and engineering units. As structured, factual, and symbolic repositories, graphs require rigorous validation.

\sloppypar 
\emph{Schema languages} such as SHACL~\cite{pareti2021review}, ShEx~\cite{shex_spec}, and PG-Schema~\cite{angles2023pg} have proved to be effective in specifying and validating constraints over graphs. When a violation is detected, the invalid graph must be repaired. This is often a non-trivial task, since multiple repair strategies may exist. For example, Fig.~\ref{fig:mot_ex}  shows an academic publication schedule. The schema enforces temporal ordering with the well-established time ontology~\cite{hobbs2004ontology}: if node $u$ has a \texttt{before} edge to node $v$, then the \texttt{end} time of $u$ must be earlier than the \texttt{start} time of $v$. \texttt{submission phase1} adheres to this constraint, where \texttt{reviewing phase1} starts after \texttt{submission phase1} ends. \texttt{submission phase2}, however, violates the schema because it ends on \texttt{Oct 10} while \texttt{reviewing phase2} starts on \texttt{Oct 8}. There are many valid repairs, including the following two groups of options to satisfy the schema:  
\begin{enumerate}[left=0pt, align=left]
    \item[(Op1)] Delay \texttt{reviewing phase2} to start any time after \texttt{Oct 10}.
    \item[(Op2)] Shorten \texttt{submission phase2} to end before \texttt{Oct 8}.
\end{enumerate}

A variety of repair methods have been proposed, including symbolic reasoning~\cite{ahmetaj2022repairing}, domain-specific heuristics with human input~\cite{wright2020schimatos,fierro2022application,fierro2020mortar,pachera2025user}, learning from the edit history~\cite{pellissier2019learning,pellissier2021neural,fan2019deducing}, and large language model (LLM) prompting~\cite{arnaout2022utilizing,terdalkar2025graph}. However, evaluating and diagnosing such systems remains difficult. Existing test datasets tend to be limited in size (e.g., SHACL test suite~\cite{w3c_data_shapes_test_suite}), scoped to just a few violation types 
\cite{pachera2025user,terdalkar2025graph}, 
or tied to specific histories (e.g., Wikidata). This makes it difficult to compare repair systems, characterize the kind of repairs that are supported by different repair systems (e.g., some repair system can not reuse existing nodes in the graph), or analyze scalability.

In this paper, we introduce a formal diagnostic framework for \revwayne{Resource Description Framework (RDF) graph~\cite{RDF} and property graph~\cite{angles2023pg} repair methods}. Starting with a schema and a valid graph, our approach generates test datasets by injecting constraint violations into the valid graph, providing a coverage guarantee that every constraint in the schema is individually tested. Because the schema encapsulates domain logic, this exhaustive coverage ensures the practical relevance of test datasets. In addition, the dynamic injection of violations given any schema and graph allows the construction of novel test datasets with private schema and graphs, reducing the risk of data contamination in evaluating LLM methods~\cite{kiela2021dynabench}. Achieving coverage-guaranteed violation injection is non-trivial, as it requires a formalized mechanism to systematically negate arbitrary constraints within a schema to provide proof for coverage. Consequently, we contribute a diagnostic framework that enables the rigorous and fine-grained analysis of graph repair systems.

%
%
We validate the proposed framework on several repair systems, including a symbolic method, a supervised learning-based method trained on edit history, and an LLM-based system. Our contributions can be summarized as follows:
\begin{itemize}[left=0pt, align=left]
    \item A formalization of \emph{violation expressions} for the  systematic and controlled generation of schema violations. 
    \item A dataset generation method enabling coverage-controlled test datasets for graph repair.  
    \item A comparative evaluation of classical and LLM-based repair methods on three real-world graphs, illustrating the framework’s diagnostic capabilities.
\end{itemize}
To the best of our knowledge, this is the first formal diagnostic framework enabling precise characterization of graph repair systems by the types of repairs they can perform and their accuracy.

The remainder of this paper is organized as follows.
Sec.~\ref{sec:related} reviews the related work. Sec.~\ref{sec:framework} outlines the evaluation framework while Sec.~\ref{sec:prelim} introduces the preliminaries. Sec.~\ref{sec:ars} formalizes the violation expressions. Sec.~\ref{sec:repairsys} describes the evaluated repair systems. Sec.~\ref{sec:metrics} defines the evaluation metrics, and Sec.~\ref{sec:results} reports the results for all the repair systems. Finally, Sec.~\ref{sec:conclusion} concludes the paper.

\section{Related Work}
\label{sec:related}
Graphs are essential for various applications, but costly to create and maintain. While automatically generating graphs from unstructured text is an active research area~\cite{bi2024codekgc}, a graph generated from text would still require post-processing and validation to ensure compliance with schemas~\cite{mihindukulasooriya2023text2kgbench}. An alternative approach to achieve compliance is knowledge graph completion, which aims to predict missing knowledge rather than repair violations~\cite{malaviya2020commonsense,wang2022simkgc,shen2022comprehensive}. This task is distinct from graph repair due to the absence of a schema.

Several languages have been proposed for expressing and validating constraints over graphs, including SHACL~\cite{pareti2021review}, ShEx~\cite{shex_spec}, and PG-Schema~\cite{angles2023pg}. SHACL and ShEx are designed for RDF graphs and PG-Schema is for property graphs. While syntactically different and designed for different data models, recent work~\cite{ahmetaj2025common} shows that these languages share a common abstract semantics. Our work leverages this foundation to define violation expressions in a language-agnostic way. For creating graph schemas, research has focused on mining methods from existing graphs~\cite{duan2024scoop,cimmino2020astrea,rabbani2023extraction,barret2024finding}. 

Finally, to repair schema violations, we categorize repair systems into four paradigms. (i) \textbf{Human-in-the-loop:} Schimatos~\cite{wright2020schimatos} proposes a graphical user interface to assist manual repair, BuildingMOTIF~\cite{fierro2022application} provides repair templates for users, and Pachera et al.~\cite{pachera2025user} assign independent violations to multiple users to parallelize repairs, Fan et al.~\cite{fan2019deducing} propagate user-confirmed repairs to deduce repairs. (ii) \textbf{Logic based:} Ahmetaj et al.~\cite{ahmetaj2022repairing} propose an automated solution by encoding the repair procedure into an answer set program. (iii) \textbf{Learning based:} Pellissier et al.~\cite{pellissier2019learning} use rule mining on repair histories and, later, employ neural networks to predict repairs~\cite{pellissier2021neural}. Both of these approaches require an edit history for training. (iv) \textbf{LLM based:} Arnaout et al.~\cite{arnaout2022utilizing} and Terdalkar et al.~\cite{terdalkar2025graph} prompt LLMs for repairs. While these works propose repair methods, we focus on the evaluation of repair systems.


\begin{figure*}[t!]
    \captionsetup{aboveskip=2pt}
    \setlength{\belowcaptionskip}{-1pt}
    \centering
    \begin{tikzpicture}[scale=1, >=Latex, line cap=round]
\begin{scope}[every node/.style={rounded corners,draw}]

    \node[] (S) at (-8,1.8) {Schema $\mathcal{S}$};
    \node[] (G) at (-8,1.2) {Valid Graph $\mathcal{G}$};
    \node[fill=C0!20] (P) at (-4,1.5) {\shortstack{Proposed formalism for \\automatic VE generation}};

    \node[fill=white] (V1) at (-0.9,1.7) {VE};
    \node[fill=white] (V2) at (-0.6,1.5) {VE};
    \node[fill=white] (V3) at (-0.3,1.3) {VE};

    \node[] (V4) at (2.4,1.5) {VE};
    \node[] (G2) at (3.3,1.5) {$\mathcal{G}$};    
    \node[draw, dashed, fit=(V4)(G2), inner sep=5pt, label=above:{Apply VE}] (box) {};

    \node[] (I) at (6,1.5) {\shortstack{Invalid Graph $\mathcal{G}'$}};    

    \node[] (S2) at (5.9,0.3) {$\mathcal{S}$};
    \node[] (Gp) at (6.6,0.3) {$\mathcal{G}'$};
    \node[draw, dashed, fit=(S2)(Gp), inner sep=5pt, label=above:{Validation}] (val) {};
    \node[] (out) at (6.3, -0.9) {\shortstack{Validation\\Outcome}};

\node[fill=white, align=center, inner sep=2pt] (sys) at (3.5,-0.2) {
    \begin{tikzpicture}[node distance=0.5mm]
        \node[draw=none] (text) {\shortstack{\underline{Repair System}\\Can include LLMs,\\logic solvers, and\\knowledge bases}};

        \node[draw, inner sep=3pt, fill=white, below=of text, xshift=-0.5cm,yshift=0.1cm] (innie1) {$\mathcal{S}$};

        \node[draw, inner sep=3pt, fill=white, right=of innie1, xshift=0.5cm] (innie2) {$\mathcal{G}'$};
    \end{tikzpicture}
};
    \node[] (repair) at (0.95,-0.2) {\shortstack{Repair\\$\mathfrak{R}$}};

    \node[] (rho2) at (-1.3,-0.2) {$\mathfrak{R}$};
    \node[] (Gp2) at (-0.63,-0.2) {$\mathcal{G}'$};    
    \node[draw, dashed, fit=(rho2)(Gp2), inner sep=5pt, label=above:{Apply Repair}] (box2) {};

    \node[] (fixed) at (-2.9,-0.2) {\shortstack{Repaired\\Graph\\$\mathcal{G}'_\mathfrak{R}$}};

    \node[fill=C0!20] (eval) at (-5.8,-0.2) {\shortstack{\underline{Evaluation of $\mathfrak{R}$}\\1. Syntactic Validity\\2. Semantic Validity\\3. Isomorphism}};
\end{scope}
\draw[->] (S.east) -- (P.west);
\draw[->] (G.east) -- (P.west);
\draw[->] (P.east) -- (V2.west);

\draw[->, thick] (V2) -- node[sloped] {\shortstack{For each\\ VE}} (box);

\draw[->] (box.east) -- (I.west);

\draw[->] (I.east) -- ++(0.2,0) |- (val.east);

\draw[->] (G.south) |- (eval.west);

\draw[->] (fixed.west) -- (eval.east);

\draw[->] (box2.west) -- (fixed.east);
\draw[->] (sys.west) -- (repair.east);
\draw[->] (repair.west) -- (box2.east);
\draw[->] (val.west) -- ++(-0.2,0) |- ([yshift=3pt]sys.east) ;
\draw[->] (val.south) -- (out.north);
\draw[->, dashed] (out.west) -- ++(-0.2,0) |- ([yshift=-3pt]sys.east);

\node[draw, rounded corners, fill=blue!20, align=center] (tag1) at (-4, 2.2) {Sec. 5.1-5.3};
\node[draw, rounded corners, fill=blue!20, align=center] (tag2) at (4.1, 2.2) {Sec. 5.4};
\node[draw, rounded corners, fill=blue!20, align=center] (tag3) at (3.5, -1.6) {Sec. 6};
\node[draw, rounded corners, fill=blue!20, align=center] (tag4) at (-5.8, -1.35) {Sec. 7-8};

\end{tikzpicture}
    \mycaption{The proposed framework.}{The shaded blocks are our main contributions; tags specify the section where they are introduced.}
    \label{fig:flow}

\end{figure*}
\section{Evaluation Framework for Graph Repair}
\label{sec:framework}
Building on the abstraction of SHACL, ShEx, and PG-Schema identified in prior work~\cite{ahmetaj2025common}, we develop a language-agnostic framework for generating schema violations.

In our framework, graphs are represented in a common model capturing both RDF and property graphs. A \textbf{schema} is a set of constraints of the form $(sel, \varphi)$, where both $sel$ and $\varphi$ are unary logical expressions over nodes, called \textbf{shapes}. The \textbf{selector} $sel$ identifies the set of nodes to which the constraint applies, and those nodes must conform to the shape $\varphi$. 
Formally, a graph $\mathcal{G}$ is \textbf{valid} with respect to a schema $\mathcal{S}$ if, for every  $(sel,\varphi)\in\mathcal{S}$ and node $u$:

\[
\mathcal{S,G},u \models sel \implies \mathcal{S,G},u \models \varphi.
\]

This abstract syntax supports predicates, quantifiers, logical connectives, and path expressions, enabling reasoning that is independent of any specific language (e.g., SHACL, ShEx, or PG-Schema). 

The objective of this paper is the inverse of schema validation: given a valid graph $\mathcal{G}$ and a schema $\mathcal{S}$, we aim to systematically construct violations. For a constraint $(sel, \varphi) \in \mathcal{S}$, we seek a sequence of operations, such as adding or removing edges, that produces a graph $\mathcal{G}'$ such that, for at least a node $u$, we have $u\models sel$ but $u\not\models\varphi$. We express such modifications through \textbf{violation expressions} (VEs), defined in the abstract model so that they apply across different schema languages (formal definition in Sec.~\ref{sec:vio}).

We use this formulation to design an automated test-case generation methodology. As shown in Fig.~\ref{fig:flow}, the framework takes $\mathcal{S}$ and a valid $\mathcal{G}$, generates VEs, and applies each to a separate copy of $\mathcal{G}$, producing an invalid graph $\mathcal{G}'$. \revwayne{Each $(\mathcal{S},\mathcal{G}')$ is a test case, enabling the attribution of a repair to a single violation expression. This paper is scoped to these single-violation test cases for diagnostic evaluations. Unlike multi-violation scenarios where repair failures may stem from interactions between violations, this isolation acts as a control, ensuring that performance drop is attributable to the specific constraint violation.} 

After generating $\mathcal{G}'$, $\mathcal{G}'$ is validated against $\mathcal{S}$ to obtain the validation outcome, which identifies the induced violations. The repair system then uses $\mathcal{G}'$, $\mathcal{S}$, and optionally, the validation outcome to produce a repair $\mathfrak{R}$, formalized as a set of graph operations. Applying $\mathfrak{R}$ to $\mathcal{G}'$ yields a repaired graph $\mathcal{G}'_\mathfrak{R}$. To assess the quality of $\mathfrak{R}$, we employ tiered metrics, ranging from basic syntactic correctness to the strictest criterion: whether $\mathfrak{R}$ restores the original graph $\mathcal{G}$. 
The definitions for these metrics are detailed in Sec.~\ref{sec:metrics}.

In the remainder of the paper we instantiate the abstract representation of the schema using the SHACL language, providing the relevant syntax and semantics for path expressions, node tests, and complex shape constructs, following the formalizations in~\cite{ahmetaj2025common, ahmetaj2022repairing,corman2018semantics,jakubowski2024shapes}. All definitions carry over unchanged to other schema languages such as ShEx and PG-Schema by replacing the instantiation of shapes and path expressions.

\section{Preliminaries}
\label{sec:prelim}
\begin{definition}[\textbf{Common Graph}]
Let $\mathcal{N}$, $\mathcal{P}$, and $\mathcal{K}$ be countable sets of nodes, predicates, and keys, respectively, and $\mathcal{V}$ be a countable set of values (e.g., the integer 42). A common graph is a pair $\mathcal{G} = (\mathcal{E}, \rho)$ where:
\begin{itemize}[left=0pt, align=left]
    \item $\mathcal{E} \subseteq \mathcal{N} \times \mathcal{P} \times \mathcal{N}$ is the finite set of edges (labeled with predicates, e.g., \texttt{before} in Fig.~\ref{fig:mot_ex}).
    \item $\rho : \mathcal{N} \times \mathcal{K} \rightharpoonup \mathcal{V}$ is a finite-domain partial function mapping node--key pairs to values.
\end{itemize}
Following the convention in~\cite{ahmetaj2025common}, we adopt the unified view where \textit{node} refers to elements in $\mathcal{N} \cup \mathcal{V}$ and \textit{edge} refers to $\mathcal{P} \cup \mathcal{K}$. This is used for concise presentation of path expressions composed of elements in $\mathcal{P} \cup \mathcal{K}$ that apply to elements in $\mathcal{N} \cup \mathcal{V}$.
\end{definition}
\begin{definition}[\textbf{Atomic Graph Operations}]\label{def:ops}
    
Let $\mathcal{G}=(\mathcal{E},\rho)$ be a common graph. The set of atomic operations is:
\begin{enumerate}[left=0pt, align=left]
    \item Edge Operations
    \begin{itemize}[left=0pt, align=left]
        \item $+e(u,p,v)$: Add the edge $(u,p,v)$ to $\mathcal{E}$.
        \item $-e(u,p,v)$: Remove the edge $(u,p,v)$ from $\mathcal{E}$.
    \end{itemize}
    \item Attribute Operations
    \begin{itemize}[left=0pt, align=left]
        \item $+a(u, k, val)$: Set $\rho(u, k) = val$.
        \item $-a(u,k)$: Remove the mapping for $(u, k)$ from $\rho$.
    \end{itemize}
    
\end{enumerate}
More complex actions can be obtained by composing these atomic operations. A replace operation, for instance, simultaneously removes an existing element and adds a new one, indicated by the $\cdot$ symbol. We define a shorthand for replacing the target of an edge or the value of an attribute: $rep(u,q,v,v')$ denotes:
\[
\begin{cases}
 -e(u,q,v) \cdot +e(u,q,v') &\text{if $q$ is a predicate}\ (q\in \mathcal{P}),\\
 -a(u,q) \cdot +a(u,q,v')& \text{if $q$ is a key}\ (q\in\mathcal{K}).
\end{cases}
\]
\end{definition}

\begin{definition}[\textbf{Path Expressions}]
SHACL uses path expressions to navigate the graph. A path expression $\pi$ is built recursively from predicates ($p$) and the following path operators, inverse ($\pi^-$), sequence ($\pi\cdot \pi^\prime$), alternative ($\pi \cup \pi^\prime$), zero-or-more repetitions ($\pi^* $), and zero-or-one ($\pi?$). The grammar is as follows:
\begin{align*}
    \pi ::= p \mid \pi^- \mid \pi\cdot \pi^\prime \mid \pi \cup \pi^\prime \mid \pi^* \mid \pi?,
\end{align*}
where $p\in\mathcal{P}\cup \mathcal{K}$. The evaluation of $\pi$ over a graph $\mathcal{G}$, denoted $\llbracket\pi\rrbracket^{\mathcal G}$, yields the set of all endpoint pairs $(u,v)$ connected by paths matching the pattern. The formal semantics are given in Table~\ref{tab:paths}. We omit the superscript $\mathcal{G}$ when it is clear from the context. For a given source node $u$, we define its set of \textbf{$\pi$-targets} as all the nodes connected from $u$ via $\pi$:
\(\llbracket \pi \rrbracket^{\mathcal{G}}(u) = \{\, v \mid (u,v) \in \llbracket \pi \rrbracket^{\mathcal{G}} \,\}.\) {For example, in Fig.~\ref{fig:mot_ex}, $\llbracket\texttt{before$\cdot$start}\rrbracket$(\texttt{submission phase1})=\{\texttt{Jan 12}\}}.

While a path expression is an abstract pattern, a concrete instance of a path expression in $\mathcal{G}$ is a trace $\tau$, represented by a sequence of alternating nodes and edges:
\[
\tau = \langle u_{0},p_{1},u_{1},\dots ,p_{k},u_{k}\rangle\quad(k\ge 0)
\]
such that each \((u_{i},p_{i+1},u_{i+1})\in\mathcal{G}\). Each triple \((u_{i},p_{i+1},u_{i+1})\) corresponds to a single edge traversal and is called a hop. For a non-zero-length trace ($k\geq1$), the final hop is the triple \((u_{k-1},p_{k},u_{k})\). We denote the start and end nodes of a trace $\tau$ as $src(\tau)=u_0$ and $tgt(\tau)=u_k$. The set of all traces in $\mathcal{G}$ is denoted $\Gamma(\mathcal G)$. A trace $\tau$ is said to \textbf{match} an expression $\pi$ (written $\tau\models\pi$) if its endpoints $(src(\tau), tgt(\tau))$ are in $\llbracket\pi\rrbracket^{\mathcal G}$. 

For example, consider the following trace in Fig.~\ref{fig:mot_ex},
\begin{equation*}
    \begin{aligned}
    \tau_1=\langle&\texttt{submission phase2$,$before$,$}\\
                  &\texttt{reviewing phase2$,$start$,$Oct 8}\rangle.  
    \end{aligned}
\end{equation*}
$\langle\texttt{reviewing phase2$,$start$,$Oct 8}\rangle$ is the final hop of $\tau_1$. Since $(src(\tau_1),tgt(\tau_1))\in\llbracket\texttt{before$\cdot$start}\rrbracket$, $\tau_1$ matches the path expression \texttt{before$\cdot$start}.

Then, we define $\Gamma(u,\pi,v)$ as the set of all traces that match $\pi$ and connect source $u$ to target $v$: 
\begin{equation*}
    \begin{aligned}
    \Gamma(u,\pi,v)=\{\tau\in\Gamma(\mathcal{G})\mid& \tau\models\pi \ \text{and}\ src(\tau)=u \ \text{and}\ tgt(\tau)=v \}.
    \end{aligned}
\end{equation*}
Finally, to filter an arbitrary set of traces $T$ for those ending at a specific node $u$, we define the target-restricted trace set: $$T|_u=\{\tau\in T\mid tgt(\tau)=u\}.$$
\end{definition}

\begin{definition}[\textbf{Node Tests}]
SHACL supports constraints over nodes, e.g., for checking values or class membership. Following the literature~\cite{ahmetaj2025common}, we abstract these constraints as a set $\Omega$ of \textbf{node tests}. For each SHACL keyword $t$ and node or value $y$, $\omega_{t,y} \in \Omega$. A node test, test$(\omega_{t,y})$, is a unary predicate. i.e., $\text{test}(\omega_{t,y}) : (\mathcal{N}\cup\mathcal{V}) \to \{\top, \bot\}$\footnote{The keywords covered are \texttt{sh:languageIn}, \texttt{sh:datatype}, \texttt{sh:nodeKind}, \texttt{sh:class}, \texttt{sh:minExclusive}, \texttt{sh:maxExclusive}, \texttt{sh:minInclusive}, \texttt{sh:maxInclusive}, \texttt{sh:minLength}, \texttt{sh:maxLength}, \texttt{sh:hasValue}, and \texttt{sh:pattern} (see Sec.~2 in~\cite{jakubowski2024shapes}).}.  
Three cases we use throughout the paper are:
\begin{enumerate}[left=0pt, align=left]
    \item \texttt{sh:hasValue}, where $\text{test}(\omega_{\texttt{sh:hasValue}, y})$ is true when the entity under evaluation equals $y$.
    \item \texttt{sh:class}, where $\text{test}(\omega_{\texttt{sh:class}, y})$ is true when the entity under evaluation is an instance of $y$ or subclasses of $y$.  
    \item \texttt{sh:minLength}, where test$(\omega_{\texttt{sh:minLength}, y})$ is true when the entity under evaluation is a string of length at least $y$.  
\end{enumerate}
We introduce shorthand: ``$=y$'' for $\text{test}(\omega_{\texttt{sh:hasValue}, y})$, ``$\in y$'' for $\text{test}(\omega_{\texttt{sh:class}, y})$, and ``$\text{len}\geq y$'' for $\text{test}(\omega_{\texttt{sh:minLength}, y})$.


For example, using our shorthand for \texttt{sh:hasValue}, the test ``$=\texttt{Oct 10}$'' checks for equality with the literal $\texttt{Oct 10}$. Thus, this test evaluates to $\top$ on the value $\texttt{Oct 10}$ itself and $\bot$ on any different node in Fig.~\ref{fig:mot_ex}.

\end{definition}

\begin{table}[t]
    \captionsetup{aboveskip=2pt}
    \setlength{\belowcaptionskip}{-1pt}
\centering
\caption{Evaluation of path expressions.}
\label{tab:paths}
\begin{tabular}{ll}
\toprule
$\pi$ & $\llbracket \pi \rrbracket^{\mathcal{G}} \subseteq (\mathcal{N}\cup\mathcal{V}) \times (\mathcal{N}\cup\mathcal{V})$ \\
\midrule
$p$ & $\{(u,v) \mid (u, p, v) \in \mathcal{G}\}$ \\
$\pi^{-}$ & $\{(u, v) \mid (v, u) \in \llbracket \pi \rrbracket^{\mathcal{G}}\}$ \\
$\pi \cdot \pi'$ & $\{(u, v) \mid \exists w : (u, w) \in \llbracket \pi \rrbracket^{\mathcal{G}} \wedge (w, v) \in \llbracket \pi' \rrbracket^{\mathcal{G}}\}$ \\
$\pi \cup \pi'$ & $\llbracket \pi \rrbracket^{\mathcal{G}} \cup \llbracket \pi' \rrbracket^{\mathcal{G}}$ \\
$\pi^*$ & $\{(u, u) \mid u \in \mathcal{N}\} \cup \llbracket \pi \rrbracket^{\mathcal{G}} \cup \llbracket \pi \cdot \pi \rrbracket^{\mathcal{G}} \cup \ldots$ \\
$\pi?$ & $\{(u, u) \mid u \in \mathcal{N}\} \cup \llbracket \pi \rrbracket^{\mathcal{G}} $ \\
\bottomrule
\end{tabular}
\end{table}

\begin{definition}[\textbf{SHACL Shapes and Subshapes}]
\label{def:shapes}
Let $\Sigma$ be a set of \textbf{shape names}, disjoint from $\mathcal{N}$, $\mathcal{V}$, $\mathcal{P}$, and $\mathcal{K}$.  
Each shape name $s \in \Sigma$ is associated with a shape formula $def(s,\mathcal{S})$ in the schema $\mathcal{S}$, following the notation in ~\cite{jakubowski2024shapes}. Shapes are defined inductively from two kinds of building blocks: 

\noindent \textbf{Atomic shapes} $\alpha$: the base cases with no proper subshapes:
\[
\begin{aligned}
\alpha ::= &\; \top 
\mid \text{test}(\omega) 
\mid \text{closed}(Q) 
\mid \text{eq}(\pi,q) 
\mid \text{disj}(\pi,q) \\
&\mid \text{lessThan}(\pi,q) 
\mid \text{lessThanEq}(\pi,q) 
\mid \text{uniqueLang}(\pi) \\
&\mid \neg \beta \quad (\beta \text{ is one of the above non-negated forms})
\end{aligned}
\]
\textbf{Compound shapes} $\varphi$: formed from atomic shapes, shape names, and logical or cardinality operators:
\[
\begin{aligned}
\varphi ::= &\; \alpha \quad (\text{atomic shape})\\ 
&\mid s \quad (s \in \Sigma) \\
&\mid \varphi \wedge \varphi' 
\mid \varphi \vee \varphi' 
\mid \neg \varphi \quad (\text{non-atomic negation}) \\
&\mid \forall \pi.\varphi 
\mid \geq_m \pi.\varphi 
\mid \leq_n \pi.\varphi,
\end{aligned}
\]
where $\omega \in \Omega$, $Q \subseteq \mathcal{P}\cup\mathcal{K}$, $\pi,q$ are path expressions, and $m \in \mathbb{N}^+, n\in\mathbb{N}$.  
The semantics of $u \models \varphi$ are defined inductively over this syntax in Table~\ref{tab:shapes}. In particular, ${\geq_n}\pi.\varphi$ and ${\leq_n}\pi.\varphi$ are qualified cardinality constraints. To formalize their semantics, we define the set of \textbf{qualified $\pi$-targets} of a node $u$ with respect to a shape $\varphi$ as the subset of its $\pi$-targets that satisfy $\varphi$:
\[
\llbracket \pi,\varphi \rrbracket^{\mathcal{G}}(u) \;=\; \{\, v \in \llbracket \pi \rrbracket^{\mathcal{G}}(u) \mid v \models \varphi \,\}.
\]
The semantics of the cardinality constraints can then be expressed using this set as in Table \ref{tab:shapes}.

Given a shape $\varphi$, its set of \textbf{subshapes}, written $\mathrm{Sub}(\varphi)$, is the set containing $\varphi$ itself and all shapes occurring inside it. Formally: $\mathrm{Sub}(\varphi)$ is defined as
\[
\begin{cases}
\{\varphi\} & \varphi \text{ atomic}, \\
\mathrm{Sub}(def(s,\mathcal{S})) & \varphi = s \in \Sigma, \\
\{\varphi\} \cup \mathrm{Sub}(\varphi_1) \cup \mathrm{Sub}(\varphi_2) & \varphi = \varphi_1 \wedge \varphi_2 \ \text{or} \ \varphi_1 \vee \varphi_2, \\
\{\varphi\} \cup \mathrm{Sub}(\psi) & \varphi = \neg \psi \ \text{and $\varphi$ non-atomic}, \\
\{\varphi\} \cup \mathrm{Sub}(\psi) & \varphi \in \{\forall \pi.\psi, \geq_n\pi.\psi,  \leq_n\pi.\psi\}.
\end{cases}
\]
To illustrate the subshape definition, consider the selector $sel=\geq_1\texttt{before}.\top$ from Fig.~\ref{fig:mot_ex}. Its set of subshapes is calculated as $\mathrm{Sub}(sel)={\{\geq_1\texttt{before}.\top}\}\cup \mathrm{Sub}(\top)=\{\geq_1\texttt{before}.\top,\top\}$.

\end{definition}

\begin{table}[t]
    \captionsetup{aboveskip=2pt}
    \setlength{\belowcaptionskip}{-1pt}
\centering
\mycaption{Semantics of a SHACL shape $\varphi$.}{$<$ and $\leq$ are the standard orderings on comparable terms (e.g., numbers, strings, dates) and $v \not\sim w$ means $v$ and $w$ are literals with different language tags.}
\label{tab:shapes}
\begin{tabular}{ll}
\toprule
$\varphi$ & $\mathcal{S,G}, u \models \varphi$ if: \\
\midrule
$\top$ & trivially satisfied \\
$s$ & $\mathcal{S,G},u\models def(s,\mathcal{S})$\\
$\mathrm{test}(\omega_{t,y})$ & $u\models \omega_{t,y}$ \\
$\mathrm{closed}(Q)$ & $\forall p \in (\mathcal{P}\cup\mathcal{K}) \setminus Q : \text{ not } \mathcal{S,G}, u \models \geq_1p.\top$ \\
$\mathrm{eq}(\pi, q)$ & $\{v \mid (u, v) \in \llbracket \pi \rrbracket^{\mathcal{G}}\} = \{v \mid (u, v) \in \llbracket q \rrbracket^{\mathcal{G}}\}$ \\
$\mathrm{disj}(\pi, q)$ & $\{v \mid (u, v) \in \llbracket \pi \rrbracket^{\mathcal{G}}\} \cap \{v \mid (u, v) \in \llbracket q \rrbracket^{\mathcal{G}}\} = \emptyset$ \\
$\varphi \wedge \varphi'$ & $\mathcal{S,G}, u \models \varphi$ and $\mathcal{S,G}, u \models \varphi'$ \\
$\varphi \vee \varphi'$ & $\mathcal{S,G}, u \models \varphi$ or $\mathcal{S,G}, u \models \varphi'$ \\
$\neg \varphi$ & not $\mathcal{S,G}, u \models \varphi$ \\
$\forall \pi.\varphi$& $\forall v ((u,v)\in\llbracket\pi\rrbracket^\mathcal{G} \to \mathcal{S,G},v\models\varphi)$\\
${\geq_n}\pi.\varphi$ & $|\llbracket \pi,\varphi \rrbracket^{\mathcal{G}}(u)| \geq n$ \\
${\leq_n}\pi.\varphi$ & $|\llbracket \pi,\varphi \rrbracket^{\mathcal{G}}(u)| \leq n$ \\
$lessThan(\pi,q)$ & $\forall (u,v)\in\llbracket \pi\rrbracket^\mathcal{G}, (u,w)\in\llbracket q\rrbracket^\mathcal{G}:v<w$\\
$lessThanEq(\pi,q)$ & $\forall (u,v)\in\llbracket \pi\rrbracket^\mathcal{G}, (u,w)\in\llbracket q\rrbracket^\mathcal{G}:v\leq w$\\
$uniqueLang(\pi)$&  $\forall (u,v), (u,w)\in\llbracket \pi\rrbracket^\mathcal{G} \text{and } v\neq w :v\not\sim w$\\
\bottomrule
\end{tabular}
\end{table}

\begin{definition}[\textbf{Non-recursive Schema}]\label{def:nonrecursive}
Let each shape name $s \in \Sigma$ be associated with a shape formula $def(s, \mathcal{S})$. We say that $s$ \emph{refers} to a shape name $s'$ if there exists an integer $k \ge 1$ and a sequence of shape names \(s_0, s_1, \dots, s_k\) such that $s_0 = s$, $s_k = s'$, and $s_i$ occurs in $def(s_{i-1}, \mathcal{S})$ for all $1 \le i \le k$.
A schema $\mathcal{S}$ is \emph{non-recursive} if no shape name refers to itself, i.e., there is no such sequence with $s_0 = s_k$.
Following~\cite{ahmetaj2025common}, we consider only non-recursive schemas.
\end{definition}
\begin{definition}[\textbf{Offending Subgraph}]
When a graph $\mathcal{G}$ is validated against a schema $\mathcal{S}$ and is invalid, there exists at least one \textbf{offended constraint} $(sel, \varphi) \in \mathcal{S}$ with the corresponding \textbf{offending focus} $u$ such that $u\models sel$ and $u\not\models\varphi$. To analyze the cause of this violation, we define the \textbf{offending subgraph} $\mathsf{OffSub}_\mathcal{G}(u,(sel,\varphi))$ as the portion of $\mathcal{G}$ \textbf{visited} during the evaluation of whether $u\models sel \implies u \models \varphi$ holds according to the SHACL specification~\cite{shacl_spec}: $\mathsf{OffSub}_\mathcal{G}(u, (sel,\varphi)) = (\mathcal{E}', \rho')$, where $\mathcal E'\subseteq\mathcal E$ and $\rho'\subseteq\rho$ are the collected edges and attribute mappings.

For example, let the offending focus $u$ be \texttt{submission phase2} in Fig.~\ref{fig:mot_ex}. Since $u$ satisfies the selector $\geq_1\texttt{before}.\top$, validating the shape $\texttt{lessThan}(\texttt{end}, \texttt{before}\cdot\texttt{start})$ involves visiting two traces from $u$: $\langle\texttt{submission phase2},\texttt{end},\texttt{Oct 10}\rangle$ and $\langle\texttt{submission phase2},\texttt{before},\texttt{reviewing phase2},\texttt{start},\allowbreak \texttt{Oct 8}\rangle$. The offending subgraph is  composed of all nodes and edges from these traces, as shown in Fig.~\ref{fig:mot_ex_subgraph}.
Intuitively, every element in $\mathsf{OffSub}$ is relevant to repairing the violation. However, $\mathsf{OffSub}$ does not necessarily contain all the relevant information -- additional parts of the graph (or external knowledge) may also be important.

\end{definition}

\begin{figure}
    \captionsetup{aboveskip=2pt}
    \setlength{\belowcaptionskip}{-1pt}
    \centering
    \begin{tikzpicture}[scale=1, >=Latex, line cap=round]
\colorlet{beforecolor}{cyan}
\colorlet{startcolor}{teal}     
\colorlet{endcolor}{magenta}
\tikzset{
  beforeedge/.style={draw=beforecolor,fill=beforecolor,thick,->},
  startedge/.style={draw=startcolor,fill=startcolor,thick,->},
  endedge/.style={draw=endcolor,fill=endcolor,thick,->},
}

\begin{scope}[every node/.style={ellipse,draw,font=\ttfamily\small,minimum width=0.3cm, minimum height=0.4cm,inner sep=0.1pt},shift={(0,0)}]

    \node[] (D) at (-1,-1) {\shortstack{submission\\phase2}};
    \node[] (D2) at (-3.3,-1.3) {Oct 10};
    \node[] (d3) at (1.5,-1) {\shortstack{reviewing\\phase2}};
    \node[] (d31) at (4,-0.8) {Oct 8};
    
\end{scope}

\path [beforeedge] (D) edge (d3);

\path [startedge] (d3) edge (d31);
\path [endedge] (D) edge (D2);

\end{tikzpicture}
    \mycaption{The offending subgraph of the violation in Fig.~\ref{fig:mot_ex} with the offending focus \texttt{submission phase2}.}{}
    \label{fig:mot_ex_subgraph}
\end{figure}

\section{Abstract Rewriting System}
\label{sec:ars}
\revwayne{To formalize constraint violation injection and derive guarantee that every constraint in the schema is covered, we introduce a formal system.} As stated in Def.~\ref{def:shapes}, a shape $\varphi$ may be (i) atomic, (ii) a reference to another shape $s$, or (iii) a compound shape built from conjunction $(\varphi \wedge \varphi')$, disjunction $(\varphi \vee \varphi')$, negation $(\neg\varphi)$, universal quantification $(\forall\pi.\varphi)$, or qualified cardinality constraints (${\geq_n}\pi.\varphi$ and ${\leq_n}\pi.\varphi$). 

To unpack complex shapes into simpler components, we use an Abstract Rewriting System (ARS) $\mathcal{A} = (\texttt{Expr}, \to)$. The system starts with a high-level task—violating a complex shape—and rewrites it into simpler tasks using \textbf{violation expressions} $\texttt{Expr}$ (Sec.~\ref{sec:vio}) and \textbf{rewrite rules} $\to$ (Sec.~\ref{sec:rules}). Rewriting starts with an initial violation expression derived from a constraint in the given schema and concludes when the expression is in \textbf{normal form}, i.e., no more rewrite rules can be applied. \revwayne{The repeated applications of rewrite rules induce a graph data structure for violation expressions. We define a traversal strategy to collect a complete set of expressions, ensuring coverage of all constraints in the schema (Sec.~\ref{sec:dedup}).} Finally, the normal-form expressions are mapped to graph operations to create violations (Sec.~\ref{sec:normal}). $\mathcal{A}$ is designed to be \textbf{strongly normalizing}, ensuring any expression reduces to normal form in finitely many steps.

\noindent\textbf{Running Example.} Fig.~\ref{fig:runKG_a} illustrates our running example, a graph modeling academic publications. The graph contains individuals and papers that are \texttt{reviewedBy} these individuals. The schema $\mathcal{S}$ specifies that any \texttt{Paper} must be \texttt{reviewedBy} at least one but at most three \texttt{\textbf{Reviewers}}, where a \texttt{\textbf{Reviewer}} is defined as an individual who has a \texttt{workPhone} or \texttt{cellPhone} string at least of length $3$. The graph currently satisfies the schema.

\begin{figure*}[!t]
    \captionsetup{aboveskip=-8pt}
    \setlength{\belowcaptionskip}{-1pt}  
    \centering
  \begin{minipage}[b]{0.41\linewidth}
    \centering      
        \begin{subfigure}[t]{\linewidth}
        \centering
        \begin{tikzpicture}[scale=1, >=Latex, line cap=round]
    \path[use as bounding box] (-3.85, -1.75) rectangle (3.35, 4.1);
    \node[anchor=south west, xshift=0pt, yshift=-3pt, font=\small] 
                 at (current bounding box.south west) {\textbf{(a)}};

\colorlet{reviewedbycolor}{C0}
\colorlet{aedgecolor}{C1}     
\colorlet{authoredgecolor}{C2}
\colorlet{workedgecolor}{C3}
\colorlet{celledgecolor}{C4}
\tikzset{
  reviewedby/.style={draw=reviewedbycolor,fill=reviewedbycolor,thick,->},
  aedge/.style={draw=aedgecolor,fill=aedgecolor,thick,->},
  workedge/.style={draw=workedgecolor,fill=workedgecolor,thick,->},
  celledge/.style={draw=celledgecolor,fill=celledgecolor,thick,->},
  albowededge/.style={draw=aedgecolor,thick,->, shorten >=0.1pt},
  authoredge/.style={draw=authoredgecolor,fill=authoredgecolor,thick,->}
}

\begin{scope}[every node/.style={ellipse,draw,font=\ttfamily\small,minimum width=0.3cm, minimum height=0.5cm,inner sep=0.1pt},shift={(0.3,2.8)}]
    \node[fill=C1!20] (Pp) at (-0.8,1) {Paper};

    \node[] (A) at (-2.8,0) {Ada};
    \node[] (one) at (-2.8,-1) {``123''};
    \node[] (B) at (-1.25,0) {Bob};
    \node[] (four) at (-1.25,-1) {``456''};
    \node[] (D) at (0.2,0) {Dan};
    \node[] (seven) at (0.2,-1) {``789''};
    \node[] (C) at (1.4,0) {Cole};
    \node[] (E) at (2.5,0) {Eve};
    
    \node[](At) at (-2.8,1) {PaperA};
    \node[](Bt) at (1.4,1) {PaperABC};

\end{scope}

\path [reviewedby] (At) edge (A);
\path [reviewedby] (Bt) edge (A);
\path [reviewedby] (Bt) edge (B);
\path [reviewedby] (Bt) edge (C);

\path [aedge] (At) edge (Pp);
\path [aedge] (Bt) edge (Pp);

\path[celledge] (A) edge[bend left] (one);
\path[workedge] (A) edge[bend right] (one);
\path[celledge] (B) edge (four);
\path[workedge] (D) edge (seven);




\path [authoredge] (Bt) edge (E) ;

\begin{scope}[every node/.style={font=\ttfamily\small},shift={(0,-0.6)}]
  \node[font=\ttfamily] (lbl1) at (-2.5,1.95) {reviewedBy};
  \draw[reviewedby,->] (-3.3,1.8) -- (-1.4,1.8);
  \node[font=\ttfamily] (lbl2) at (-0.1,1.95) {author};
  \draw[authoredge,->] (-1,1.8) -- (0.9,1.8);
  \node[font=\ttfamily] (lbl3) at (2,1.95) {a};
  \draw[aedge,->] (1.3,1.8) -- (3,1.8);
  
  \node[font=\ttfamily] (lbl4) at (-1.4,1.55) {workPhone};
  \draw[workedge,->] (-2.2,1.4) -- (-0.3,1.4);

  \node[font=\ttfamily] (lbl4) at (0.9,1.55) {cellPhone};
  \draw[celledge,->] (0.1,1.4) -- (2,1.4);
\end{scope}

\node[anchor=north, align=left, draw, rounded corners, fill=gray!10, text width=6.5cm, font=\small]
      at (0,0.65) {%
      \textbf{Schema $\mathcal{S}$:}
      \begin{itemize}[leftmargin=*]
        \item  ($sel_1, \varphi_1$)=\raisebox{-0.7\baselineskip}{$\begin{aligned}
            (\in\texttt{Paper},&\geq_1\texttt{reviewedBy}.\texttt{\textbf{Reviewer}}\wedge\\
            &\leq_3\texttt{reviewedBy}.\texttt{\textbf{Reviewer}})
        \end{aligned}$}
        \item ($sel_2, \varphi_2$)=$(\neg\top, \geq_1\pi.\text{len}\geq3)$
      \end{itemize}
          \noindent
$def(\texttt{\textbf{Reviewer}}, \mathcal{S}) =\ (\geq_1\pi. \text{len}\geq 3)$,\\
\noindent where $\pi=\texttt{workPhone}\cup\texttt{cellPhone}$};
\end{tikzpicture}
        \phantomcaption
        \label{fig:runKG_a}
        \end{subfigure}%
    \vspace{1mm}
    \begin{subfigure}[t]{\linewidth}
    \centering
    \begin{tikzpicture}[scale=1, >=Latex, line cap=round]
     \path[use as bounding box] (-3.85, -2.9) rectangle (3.4, 0.5);

      \node[anchor=south west, xshift=0pt, yshift=-3pt, font=\small] 
                 at (current bounding box.south west) {\textbf{(b)}};
\colorlet{reviewedbycolor}{C0}
\colorlet{aedgecolor}{C1}     
\colorlet{authoredgecolor}{C2}
\tikzset{
  reviewedby/.style={draw=reviewedbycolor,fill=reviewedbycolor,thick,->},
  aedge/.style={draw=aedgecolor,fill=aedgecolor,thick,->},
  authoredge/.style={draw=authoredgecolor,fill=authoredgecolor,thick,->}
}

      \node[anchor=north, align=left, draw, rounded corners, fill=gray!10, text width=6.7cm, font=\small, inner sep=0pt]
      at (0,0.5) {
      \begin{itemize}[leftmargin=*,nosep, label={}]
        \item $T_0$=$\{\langle \texttt{{PaperABC}} \rangle, \langle \texttt{{PaperA}} \rangle\}$ 
        \item $T_1$=$\{\langle \texttt{{PaperABC,reviewedBy,Ada}} \rangle\}$ 
        \item $T_2$=$\{\langle \texttt{{PaperABC,reviewedBy,Bob}} \rangle\}$ 
        \item $T_3$=$\{\langle \texttt{{PaperABC,reviewedBy,Bob,cellPhone,``456''}} \rangle\}$ 
        \item $T_4$=\raisebox{-0.7\baselineskip}{$\begin{aligned}
            \{&\langle \texttt{{PaperABC,reviewedBy,Ada,workPhone,``123''}} \rangle,\\[-2pt]
            &\langle\texttt{{PaperABC,reviewedBy,Ada,cellPhone,``123''}} \rangle\}
        \end{aligned}$}
        \item $T_5$=$\{\langle \texttt{{PaperA,reviewedBy,Ada}} \rangle\}$ 
        \item $T_6$=\raisebox{-0.7\baselineskip}{$\begin{aligned}
            \{&\langle \texttt{{PaperA,reviewedBy,Ada,workPhone,``123''}} \rangle,\\[-2pt]
            &\langle\texttt{{PaperA,reviewedBy,Ada,cellPhone,``123''}} \rangle\}
        \end{aligned}$}
      \end{itemize}
      };

\end{tikzpicture}
    \phantomcaption
    \label{fig:runKG_b}
\end{subfigure}%
  \end{minipage}
  \hfill
  \begin{subfigure}[b]{0.58\linewidth}
    \centering
    \begin{tikzpicture}
    \path[use as bounding box] (-4.5, -4.5) rectangle (4.5, 4.6);

    \node[anchor=south west, xshift=-1pt, yshift=-3pt, font=\small] 
                 at (current bounding box.south west) {\textbf{(c)}};
    \node (mytree) at (0,0) {
    \begin{forest}
  dir tree switch=at 0,
  for tree={
    font=\ttfamily\small,
    rect,
    align=center,
    inner sep=0.3pt,            
    edge+={draw=darkgray},
    where level=0{%
      colour me out=blue!50!white, 
    }{%
      if level=1{%
        colour me out=blue!95!black, 
      }{%
        colour me out=magenta!50!orange!75!white, 
        edge+={-Triangle},
      },
    },
  }
  [{VIO$^{(T_0, \exists)}$($\in$Paper,$\geq_1$reviewedBy.\textbf{Reviewer}$\wedge\leq_3$reviewedBy.\textbf{Reviewer})}
    [{VIO$^{(T_0, \exists)}$($\in$Paper,$\geq_1$reviewedBy.\textbf{Reviewer})}, edge={draw=C6}, edge label={node[midway,left,font=\small,color=C6]{1}}
      [{VIO$^{(T_1, \forall)}$($=$Ada,\textbf{Reviewer})VIO$^{(T_2, \forall)}$($=$Bob,\textbf{Reviewer})}, fill=C0!20
            [{VIO$_\pi$(Ada,``123'')VIO$_\pi$(Bob,``456'')}]
            [{VIO$_\pi$(Ada,``123'')VIO$^{(T_3, \forall)}$($=$``456'',$\text{len}\geq3$)}]
            [{VIO$^{(T_4, \forall)}$($=$``123'',$\text{len}\geq3$)VIO$_\pi$(Bob,``456'')}]
            [{VIO$^{(T_4, \forall)}$($=$``123'',$\text{len}\geq3$)VIO$^{(T_3, \forall)}$($=$``456'',$\text{len}\geq3$)}]
      ]
      [{VIO$^{(T_1, \forall)}$($=$Ada,\textbf{Reviewer})VIO$_{\texttt{reviewedBy}}$(PaperABC,Bob)}, fill=C0!20, edge={draw=C6} , edge label={node[midway,left,font=\small,color=C6]{2}}
        [{VIO$^{(T_1, \forall)}$($=$Ada,$\geq_1\pi.\text{len}\geq3$)VIO$_{\texttt{reviewedBy}}$(PaperABC,Bob)}, edge={draw=C6},edge label={node[midway,left,font=\small,color=C6]{3}}
            [{VIO$_\pi$(Ada,``123'')VIO$_{\texttt{reviewedBy}}$(PaperABC,Bob)}]
            [{VIO$^{(T_4, \forall)}$($=$``123'',$\text{len}\geq3$)VIO$_{\texttt{reviewedBy}}$(PaperABC,Bob)}, edge={draw=C6}, edge label={node[midway,left,font=\small,color=C6]{4}}]
        ]
      ]
      [{VIO$_{\texttt{reviewedBy}}$(PaperABC,Ada)VIO$^{(T_2, \forall)}$($=$Bob,\textbf{Reviewer})}, fill=C0!20
        [{VIO$_{\texttt{reviewedBy}}$(PaperABC,Ada)VIO$_\pi$(Bob,``456'')}]
        [{VIO$_{\texttt{reviewedBy}}$(PaperABC,Ada)VIO$^{(T_3, \forall)}$($=$``456'',$\text{len}\geq3$)}]]
      [{VIO$_{\texttt{reviewedBy}}$(PaperABC,Ada)VIO$_{\texttt{reviewedBy}}$(PaperABC,Bob)}, fill=C0!20]
      [{VIO$^{(T_5, \forall)}$($=$Ada,\textbf{Reviewer})}, fill=C2!20
        [{VIO$_\pi$(Ada,``123'')}]
        [{VIO$^{(T_6, \forall)}$($=$``123'',$\text{len}\geq3$)}]]
      [{VIO$_\texttt{reviewedBy}$(PaperA,Ada)}, fill=C2!20]
    ]
    [{VIO$^{(T_0,\exists)}$($\in$Paper,$\leq_3$reviewedBy.\textbf{Reviewer})}, edge={draw=C6}, edge label={node[midway,left,font=\small,color=C6]{5}}
        [{SAT$_\texttt{reviewedBy}^4$($\in$Paper,$\geq_1\pi.\text{len}\geq 3$)}, edge={draw=C6}, edge label={node[midway,left,font=\small,color=C6]{6}}]
    ]
  ]
\end{forest}
};
  


    \end{tikzpicture}
    \phantomcaption
    \label{fig:tree}
\end{subfigure}

\mycaption{Running example. (a) A graph and schema modeling academic publications. (b) The trace sets used in (c). (c) Expansion graph and DFS from $(sel_1,\varphi_1)$.}{In (c), some intermediate steps such as shape unfolding \texttt{\textbf{Reviewer}$\to($$\geq_1\pi.\text{len}\geq 3$)} are omitted.}
\label{fig:sidebyside}
\end{figure*}

\subsection{Violation Expressions}
\label{sec:vio}
The repeated applications of rewrite rules induce a graph data structure for violation expressions, called expansion graph, which represents the space of possible strategies to violate a constraint. For our running example, the expansion graph violating $(sel_1,\varphi_1)$ in Fig.~\ref{fig:runKG_a} is shown in Fig.~\ref{fig:tree}. The nodes of the expansion graph are violation expressions. We now define the components of violation expressions.

A violation expression is constructed from a set of fundamental tasks combined with two \revwayne{task} operators:

\begin{itemize}[left=0pt, align=left]
    \item \textbf{Simultaneous} $(\cdot)$: A conjunctive operator requiring that all tasks are achieved.
    \item \textbf{Alternative} $(+)$: A disjunctive operator requiring that one of the tasks is achieved.
\end{itemize}
There are three fundamental tasks:
\begin{itemize}[left=0pt, align=left]
    \item \textbf{Shape violation task} $\texttt{VIO}(sel,\varphi)$: \revwayne{Find a minimal set of graph operations that, when applied to the graph, cause a selected node $u\models sel$ to no longer satisfy $\varphi$.} 
    \item \textbf{Path violation task} \(\texttt{VIO}_\pi(u,v)\): \revwayne{Find a minimal set of graph operations that, when applied to the graph, remove all $\pi$-matching traces between nodes $u$ and $v$.} 
    \item \textbf{Qualified satisfaction task} \(\texttt{SAT}_\pi^n(sel,\varphi)\): \revwayne{Find a minimal set of graph operations that, when applied to the graph, result in a selected node $u\models sel$ having \(n\) distinct $\pi$-targets, each satisfying the shape \(\varphi\).}
\end{itemize}

A challenge arises when rewriting shapes with path expressions. For example, consider violating $(sel_1,\varphi_1)$ on the \texttt{PaperA} node in our running example. A subshape of $\varphi_1$ requires that \texttt{PaperA} have at least one valid \texttt{\textbf{Reviewer}} ($\geq_1\texttt{reviewedBy}.\texttt{\textbf{Reviewer}}$). 
As we will see in Sec.~\ref{sec:rules}, our ARS will follow the trace matching \texttt{reviewedBy} from \texttt{PaperA} to reach \texttt{Ada} and generate a sub-task: make \texttt{Ada} no longer a valid \texttt{\textbf{Reviewer}}. However, our ARS needs to remember that this sub-task on \texttt{Ada} is generated because of the trace from \texttt{PaperA}. This context is managed by provenance $(T,Q)$, which augments the shape violation task to $\texttt{VIO}^{(T,Q)}(sel,\varphi)$, where:
\begin{itemize}[left=0pt, align=left]
    \item \textbf{Trace Set} $T\subseteq\Gamma(\mathcal{G})$ records the sequence of edges followed from the initial task's selected nodes to the current sub-task's selected nodes. In our example, our ARS navigates from \texttt{PaperA} to \texttt{Ada}, so the trace set is $\{\langle\texttt{PaperA},\texttt{reviewedBy},\texttt{Ada}\rangle\}$. \revwayne{Fig.~\ref{fig:runKG_b} details the evolution of trace set $T_0$ into $T_1 \dots T_6$, which correspond to the trace sets used in the nodes in the expansion graph in Fig.~\ref{fig:tree}.}
    \item \textbf{Mode} $Q\in\{\exists,\forall\}$ quantifies the traces that must be ``modified'' in $T$. The mechanism for modifying a trace is detailed in Sec.~\ref{sec:shape_vio}.
    \begin{itemize}[left=0pt, align=left]
        \item $Q=\forall$ (Modify All): Requires modifying \textit{all} the traces in $T$. This is necessary when generating a violation to the ``at least'' ($\geq_n\pi.\varphi$) constraints. For example, to make \texttt{Ada} an invalid \texttt{\textbf{Reviewer}}, we must violate the requirement that she has at least one phone string ($\geq_1\pi.\texttt{len}\geq3$, where $\pi=\texttt{workPhone}\cup\texttt{cellPhone}$). Since \texttt{Ada} is connected to \texttt{``123''} via both \texttt{workPhone} and \texttt{cellPhone}, our ARS needs to modify both $\langle\texttt{{Ada,workPhone,``123''}} \rangle$ and $\langle\texttt{{Ada,cellPhone,``123''}} \rangle$ to ensure \texttt{Ada} has no phone.
        \item $Q=\exists$ (Modify One): Requires modifying just \textit{one} trace in $T$. This occurs when generating a violation to the universal quantifier $\forall\pi.\varphi$ constraints, where one counterexample is sufficient.
     \end{itemize}
\end{itemize}
Intuitively, provenance tells us which traces led to the current sub-task and the scope for modifying those traces. The syntax of the violation expressions in our ARS is:
\begin{align*}
    \texttt{Term}  &::=\texttt{VIO}^{(T,Q)}(sel,\varphi)\mid \texttt{VIO}_\pi(u,v)\mid \texttt{SAT}_\pi^n(sel,\varphi) \mid \texttt{Term} \cdot \texttt{Term},\\
    \texttt{Expr} &::= \texttt{Term} \mid \texttt{Expr} + \texttt{Term}.
\end{align*}

Formally, the \textbf{expansion graph} for a constraint $(sel, \varphi)$ is a directed acyclic graph defined as follows:
\begin{itemize}[left=0pt, align=left]
    \item \textbf{Nodes:} Each \textbf{node} corresponds to a \texttt{Term} in our grammar, a conjunctive set of tasks ($\texttt{T}_1\cdot\texttt{T}_2 \cdots$) that must be achieved simultaneously. Nodes represent strategies for creating violations.
    \item \textbf{Root:} The \textbf{root} node of the graph is the initial violation task $\texttt{VIO}^{(T_0,Q_0)}(sel, \varphi)$. The initial provenance  $(T_0, Q_0)$ is set as:
    \begin{itemize}[left=0pt, align=left]
        \item $T_0=\{\langle u\rangle\mid u\models sel\}$: The initial trace set contains a zero-length trace for each node $u$ that matches $sel$. For our running example, since both \texttt{PaperA} and \texttt{PaperABC} satisfy the selector $sel_1$, the initial trace set is $T_0=\{\langle \texttt{{PaperABC}} \rangle, \langle \texttt{{PaperA}} \rangle\}$.
        \item $Q_0=\exists$: The initial mode is $\exists$ by default until rewritten by subsequent rewrite rules.
    \end{itemize}
    \item \textbf{Edges and Branching:} An edge represents the application of a rewrite rule. When the applied rewrite rule yields a disjunctive expression $(\texttt{T}_1+\texttt{T}_2+ \cdots)$, the parent node branches, generating one child for each term. In the running example, the task of violating $(sel_1,\varphi_1)$ on the \texttt{PaperA} node (root in Fig.~\ref{fig:tree}) branches into two children: one for violating ``at least one \texttt{\textbf{Reviewer}}'' (node~{\color{C6}$1$} in Fig.~\ref{fig:tree}) and one for violating ``at most three \texttt{\textbf{Reviewer}}'' (node~{\color{C6}$5$} in Fig.~\ref{fig:tree}) because one can perform either action to make \texttt{PaperA} violate $(sel_1,\varphi_1)$.
\end{itemize}
Consequently, the \textbf{internal nodes} of the graph are expressions that can be further rewritten, while the \textbf{leaf nodes} are expressions in normal form. Each unique path from the root to a leaf node represents a complete strategy for violating the constraint $(sel, \varphi)$.

\subsection{Rewrite Rules}
\label{sec:rules}


We now present the rewrite rules that induce the expansion graph. The expansion graph will allow the exhaustive collection of violation expressions that will be explained in Sec.~\ref{sec:dedup}.

\noindent\textbf{Rule 1: Shape Reference ($s$).}
When a shape reference $s$ is in a task, it is replaced by the shape definition $def(s, \mathcal{S})$. This ``unfolds'' the reference so that rewriting proceeds on the actual shape formula.
\begin{align*}
\texttt{VIO}^{(T,Q)}(sel, s) &\longrightarrow \texttt{VIO}^{(T,Q)}(sel, def(s, \mathcal{S}))\\
\texttt{SAT}^n_\pi(sel, s) &\longrightarrow \texttt{SAT}^n_\pi(sel, def(s, \mathcal{S}))
\end{align*}

\noindent\textbf{Rule 2: Conjunction ($\wedge$).}
A node satisfies $\varphi \wedge \varphi'$ if it satisfies both shapes. To violate the conjunction, it is enough to violate either $\varphi$ or $\varphi'$. This choice is represented by the alternative operator $+$.
\begin{equation*}
\texttt{VIO}^{(T,Q)}(sel, \varphi \wedge \varphi') \longrightarrow \texttt{VIO}^{(T,Q)}(sel, \varphi) + \texttt{VIO}^{(T,Q)}(sel, \varphi')
\end{equation*}

\noindent\textbf{Rule 3: Disjunction ($\vee$).}
A node satisfies $\varphi \vee \varphi'$ if it satisfies at least one shape. To violate the disjunction, a node must fail \textit{both} $\varphi$ and $\varphi'$. The rewrite rule generates this conjunctive task for each node $u$ that satisfies the selector $sel$. Formally,
\begin{multline*}\texttt{VIO}^{(T,Q)}(sel, \varphi \vee \varphi') \ \longrightarrow\ 
\\ \sum_{u: u \models sel}  
\big(\mathbf{if}\ u \models \varphi\ \mathbf{then}\ \texttt{VIO}^{(T|_u,Q)}(=u,\varphi)\big) \ \cdot\ \\
\big(\mathbf{if}\ u \models \varphi'\ \mathbf{then}\ \texttt{VIO}^{(T|_u,Q)}(=u,\varphi')\big)\end{multline*}
The summation $\sum_{u: u \models sel}$ states that a violation can be performed on any node $u$ satisfying the selector. For a chosen $u$, the provenance is updated using $T|_u$ (the subset of traces in $T$ ending at $u$), and the $(\cdot)$ operator creates the conjunctive task. The $\mathbf{if}\dots\mathbf{then}$ condition allows omitting the sub-task for a shape $u$ that is  already violating a shape.
%
If a shape (e.g., $\varphi'$) is not satisfied by any node that matches the selector, then no violation task will be generated for $\varphi'$ or any of its subshapes, making them \textit{unreachable} in our ARS.

\noindent\textbf{Rule 4: Negation ($\neg$).}
For a negated shape $\neg\varphi$, we push the negation inward using standard logical equivalences, as follows. 
\begin{enumerate}
    \item De Morgan's Law \begin{align*}
    \texttt{VIO}^{(T,Q)}(sel, \neg(\varphi \wedge \varphi')) &\longrightarrow \texttt{VIO}^{(T,Q)}(sel, \neg\varphi \vee \neg\varphi') \\
    \texttt{VIO}^{(T,Q)}(sel, \neg(\varphi \vee \varphi')) &\longrightarrow \texttt{VIO}^{(T,Q)}(sel, \neg\varphi \wedge \neg\varphi') \\
    \texttt{VIO}^{(T,Q)}(sel, \neg\neg\varphi) &\longrightarrow \texttt{VIO}^{(T,Q)}(sel, \varphi)
\end{align*}
    \item Quantifier Dualities \begin{align*}
    \texttt{VIO}^{(T,Q)}(sel, \neg(\forall\pi.\varphi)) &\longrightarrow \texttt{VIO}^{(T,Q)}(sel, \geq_1\pi.\neg\varphi) \\
    \texttt{VIO}^{(T,Q)}(sel, \leq_0\pi.\neg\varphi)) &\longrightarrow \texttt{VIO}^{(T,Q)}(sel, \forall\pi.\varphi) \\
    \texttt{VIO}^{(T,Q)}(sel, \neg(\geq_n\pi.\varphi)) &\longrightarrow \texttt{VIO}^{(T,Q)}(sel, \leq_{n-1}\pi.\varphi) \\
    \texttt{VIO}^{(T,Q)}(sel, \neg(\leq_n\pi.\varphi)) &\longrightarrow \texttt{VIO}^{(T,Q)}(sel, \geq_{n+1}\pi.\varphi)
\end{align*}
    \item Shape Unfolding \begin{align*}
    \texttt{VIO}^{(T,Q)}(sel, \neg s) &\longrightarrow \texttt{VIO}^{(T,Q)}(sel, \neg def(s, \mathcal{S}))
\end{align*}
\end{enumerate}
This produces an equivalent form that can then be rewritten by other rules.

\noindent\textbf{Rule 5: Qualified Cardinality ($\geq_n \pi.\varphi$).}
A node satisfies $\geq_n \pi.\varphi$ if it has at least $n$ $\pi$-targets that satisfy $\varphi$. To violate this constraint, we must reduce the number of qualified $\pi$-targets to less than $n$. First, we determine the number of $\pi$-targets to invalidate. Let $m$ be the current number of qualified $\pi$-targets. We must invalidate at least $k = m - n + 1$ of them. In our running example (Fig.~\ref{fig:runKG_a}), \texttt{PaperABC} has $m=2$ qualified reviewers (\texttt{Ada} and \texttt{Bob}). To violate $\geq_1\texttt{reviewedBy}.\texttt{\textbf{Reviewer}}$, we must invalidate $k=2-1+1=2$. For each of the $k$ targets selected, we have two invalidation strategies:
\begin{enumerate}[left=0pt, align=left]
    \item \textbf{Invalidate the target:} Make the target node violate $\varphi$, e.g., make \texttt{Ada} no longer a valid \texttt{Reviewer}.
    \item \textbf{Invalidate the trace:} Remove the $\pi$-matching trace connecting to the target node, e.g., remove the \texttt{reviewedBy} edge between \texttt{PaperABC} and \texttt{Ada}.
\end{enumerate}

When a rewrite rule follows a trace from a node $u$ to another node $v$ to generate a sub-task, it updates the provenance by \textbf{trace extensions}. This operation appends new trace segments to existing traces. For instance, when following the \texttt{reviewedBy} edge from \texttt{PaperABC} to \texttt{Ada}, the initial trace $\langle\texttt{PaperABC}\rangle$ is extended to $\langle\texttt{PaperABC},\texttt{reviewedBy},\texttt{Ada}\rangle$. Formally,
\[
T'_{u\to v}=\{\tau\circ\tau '\mid\tau\in T|_u\ \text{and}\ \tau'\in\Gamma(u,\pi,v)\},
\]
where $\circ$ denotes trace concatenation ($\langle w, p, u\rangle \circ \langle u, q, v\rangle = \langle w, p, u, q, v\rangle$). 
The full rewrite rule combines the two invalidation strategies: 
\begin{multline*}
\texttt{VIO}^{(T,Q)}(sel, \geq_n \pi.\varphi) \longrightarrow \sum_{u: u \models sel} \sum_{\substack{\Delta \subseteq \llbracket\pi,\varphi\rrbracket(u) \\ |\Delta|=k}} \\
\prod_{v \in \Delta} \left( \texttt{VIO}^{(T_{u\to v}',\forall)}(=v, \varphi) + \texttt{VIO}_\pi(u,v) \right).
\end{multline*}
The outer summation selects a node $u$ satisfying $sel$ (e.g., \texttt{PaperABC}). The inner summation chooses a subset $\Delta$ of $k$ nodes to invalidate (e.g., \texttt{Ada} and \texttt{Bob}). The product states that one of the invalidation strategies must be applied to each node in $\Delta$. The $+$ states the choice between invalidating the target or invalidating the trace. The mode for the ``invalidate target'' sub-task is $Q=\forall$ since all traces from $u$ to $v$ must be modified. In our running example, since we must invalidate $k=2$ targets and have two choices for each, we generate $2^2=4$ conjunctive sub-tasks, as shown by the blue nodes in Fig.~\ref{fig:tree}.

\noindent\textbf{Rule 6: Qualified Cardinality ($\leq_n\pi.\varphi$).}
A node $u$ satisfies $\leq_n \pi.\varphi$ if it has at most $n$ $\pi$-targets that satisfy $\varphi$.
To violate this constraint, we must increase the number of qualified $\pi$-targets to more than $n$. We directly apply the qualified satisfaction task $\texttt{SAT}^{n+1}_\pi$:
$$\texttt{VIO}^{(T,Q)}(sel, \leq_n\pi.\varphi) \longrightarrow \texttt{SAT}^{n+1}_\pi(sel, \varphi).$$
Violating a $\leq_n\pi.\varphi$ constraint always yields a terminal $\texttt{SAT}$ task. Since the objective of the \texttt{SAT} task is to satisfy $\varphi$ as an indivisible unit, our ARS does not decompose $\varphi$ further, making the subshapes of $\varphi$ \textit{unreachable}.

\noindent\textbf{Rule 7: Universal Quantifier ($\forall \pi.\varphi$).}
A node satisfies $\forall \pi.\varphi$ if all its $\pi$-targets satisfy $\varphi$. To violate this constraint, at least one $\pi$-target must fail $\varphi$. We have two invalidation strategies:
\begin{enumerate}[left=0pt, align=left]
    \item \textbf{Invalidate an existing target:} Select a $\pi$-target $v$ and generate a sub-task to make $v$ violate $\varphi$.
    \item \textbf{Introduce a new violating target:} Add a new $\pi$-target that satisfies $\neg\varphi$.
\end{enumerate}
The full rewrite rule combines the two invalidation strategies:
\begin{multline*}
\texttt{VIO}^{(T,Q)}(sel, \forall\pi.\varphi) \longrightarrow  \sum_{u:u\models sel}\sum_{v\in\llbracket\pi\rrbracket(u)}\texttt{VIO}^{(T_{u\to v}',\exists)}(=v, \varphi) +\\\texttt{SAT}_\pi^1(sel,\neg\varphi).
\end{multline*}
The rule consists of two parts connected by $+$, representing the two strategies. The first part invalidates an existing target. The summations select a node $u$ satisfying $sel$ and chooses a $\pi$-target $v$. The sub-task uses trace extension $T'_{u\to v}$ and mode $Q=\exists$, as violating a single path is sufficient. The second part introduces a new violating target by creating a new $\pi$-target that satisfies $\neg\varphi$.
%
If no node satisfying the selector has $\pi$-targets, the first part of the rule vanishes. The only option is the $\texttt{SAT}$ task, which is terminal. In this case, the subshapes of $\varphi$ are \textit{unreachable}.

In 
Rules 3, 6, and 7, we observed situations where subshapes can never appear in a \texttt{VIO} task. This occurs when a rewrite rule fails to produce \texttt{VIO} sub-tasks or yields a terminal \texttt{SAT} task. We formalize this with the notion of \textbf{shape reachability}, which captures when a subshape can appear in a \texttt{VIO} task during the rewriting process.

\begin{definition}[\textbf{Shape Reachability}]\label{def:reachability}   
A subshape $\psi$ of a shape $\varphi$ is \textbf{reachable} from an initial task $\texttt{VIO}^{(T_0,Q_0)}(sel, \varphi)$ if $\psi=\varphi$ or if there is a sequence of rewrite steps
$$\texttt{VIO}^{(T_0,Q_0)}(sel, \varphi) \longrightarrow \dots \longrightarrow \texttt{Expr},$$
where \texttt{Expr} contains a sub-task $\texttt{VIO}^{(T,Q)}(sel', \varphi')$ and $\psi$ is reachable from $\texttt{VIO}(sel', \varphi')$.

A subshape $\psi$ of a shape $\varphi$ is \textbf{unreachable} if it is not reachable. This occurs if, for every sub-task $\texttt{VIO}^{(T,Q)}(sel', \varphi')$ generated during rewriting where $\psi$ is a subshape of $\varphi'$, one of the following blocking conditions is met:
\begin{itemize}[left=0pt, align=left]
\item \textbf{Non-satisfied Disjunct:} $\varphi' = \varphi'' \vee \varphi'''$, $\psi$ is a subshape of $\varphi''$, and no node satisfies both $sel'$ and $\varphi''$.
\item \textbf{At-most Cardinality:} $\varphi' = \leq_n \pi.\varphi''$ and $\psi$ is a subshape of $\varphi''$.
\item \textbf{Universal Quantifier with No Targets:} $\varphi' = \forall\pi.\varphi''$, $\psi$ is a subshape of $\varphi''$, and no node satisfying $sel'$ has a $\pi$-target.
\end{itemize}
For an unreachable shape $\psi$, our ARS will never generate a \texttt{VIO} task with $\psi$ in it.
\end{definition}

Having established shape reachability, we now turn to a fundamental property of our rewriting system: \emph{termination}. The following theorem shows that, under mild assumptions, the rewriting process always terminates, i.e., $\mathcal{A}$ is strongly normalizing.
\begin{restatable}[\textbf{Strong Normalization of $\mathcal{A}$}]{theorem}{SN}
\label{thm:1}
If $\mathcal{G}$ and $\mathcal{S}$ are finite and $\mathcal{S}$ is non-recursive, then the rewriting system $\mathcal{A}$ is strongly normalizing.
\end{restatable}
\begin{pf}
See App.~\ref{app:strong} for the proof.
\end{pf}

\subsection{Subshape Collection and De-duplication}  
\label{sec:dedup}
To generate a coverage guaranteed test dataset, we must systematically navigate the expansion graph to exhaustively collect violation expressions. Furthermore, normalized violation expressions will be mapped to graph operations that create violations in Sec.~\ref{sec:normal}, which are applied to separate copies of the valid graph $\mathcal{G}$ to isolate the effect of violation expressions.

From Def.~\ref{def:reachability}, the reachable subshapes are those that can appear in a shape violation task during rewriting. Our goal is to collect all the reachable subshapes of each constraint by traversing its expansion graph and recording the subshapes within violation tasks.

Since \(\mathcal{S}\) is non-recursive, we process constraints in topological order, starting from shapes not referred to by others. For each \((sel, \varphi)\), we build its expansion graph rooted at \(\texttt{VIO}^{(T_0,Q_0)}(sel, \varphi)\) and run a depth-first search (DFS) to enumerate all the reachable subshapes \(\psi \in \mathrm{Sub}(\varphi)\). A subshape \(\psi\) is \textbf{covered} by a rewrite path if that path contains a \(\texttt{VIO}^{(T,Q)}(sel', \psi)\) term.

During DFS, when a choice \((+)\) is encountered, one branch is selected at random to form a path from root to leaf, yielding candidate violation expressions. Finally, we perform de-duplication: after all constraints are processed, rewrites within or across constraints may yield the same normal form expression. Since these represent identical graph operations, only one is retained.

\noindent\textbf{Running Example.} In Fig.~\ref{fig:runKG_a}, since $(sel_1,\varphi_1)$ is not referred to by other shapes and \texttt{\textbf{Reviewer}} is associated with the shape formula $\varphi_2$, DFS starts from \texttt{VIO}$^{(T_0,\exists)}$($sel_1,\varphi_1$). 
For the purpose of illustration, we first identify the set of five reachable subshapes $\mathrm{Sub}(\varphi_1)$: 
\begin{enumerate}[left=0pt, align=left]
    \item[(a)] $\geq_1\texttt{reviewedBy}.\texttt{\textbf{Reviewer}}\wedge\leq_3\texttt{reviewedBy}.\texttt{\textbf{Reviewer}}$
    \item[(b)] $\geq_1\texttt{reviewedBy}.\texttt{\textbf{Reviewer}}$
    \item[(c)] $\leq_3\texttt{reviewedBy}.\texttt{\textbf{Reviewer}}$
    \item[(d)] $\geq_1\pi.\texttt{len}\geq3$ (reachable via (b))
    \item[(e)] $\texttt{len}\geq3$ (reachable via (b))
\end{enumerate}
The DFS continues until all five reachable subshapes in $\mathrm{Sub}(\varphi_1)$ are covered. To illustrate how the termination condition works, we examine two potential outcomes of DFS.
In the following, the numbers {\color{C6}1}-{\color{C6}6} refer to the nodes in the expansion graph in Fig.~\ref{fig:tree}.
\begin{itemize}[left=0pt, align=left]
    \item $\mathcal{P}_1=\{({\color{C6}1},{\color{C6}2},{\color{C6}3},{\color{C6}4})\}$:  The DFS is incomplete because shape (c) $\leq_3\texttt{reviewedBy}.\texttt{\textbf{Reviewer}}$ is not covered.
    \item $\mathcal{P}_2=\{({\color{C6}1},{\color{C6}2},{\color{C6}3},{\color{C6}4}),({\color{C6}5},{\color{C6}6})\}$:  The DFS terminates because the addition of path (${\color{C6}5},{\color{C6}6}$) covers the missing shape (c).
\end{itemize}

We observe that the traversal visits each subshape at most once. After traversal, de‑duplication is a pass over the collected set. Therefore, the total run time scales linearly with the number of subshapes:

\begin{prop}[\textbf{Time Complexity}]
    For a finite and non-recursive schema $\mathcal{S}$ with $n_s$ subshapes, the subshape collection and de‑duplication algorithm runs in $O(n_s)$.
\end{prop}
\begin{pf} 
See App.~\ref{app:complexity} for the proof.
\end{pf}

\subsection{Apply Normalized Expressions}
\label{sec:normal}

Once normalized, de-duplicated violation expressions have been collected, we describe how each of the three fundamental tasks is mapped to graph operations, which are applied to separate copies of the valid graph $\mathcal{G}$, and our framework's coverage guarantee.

\subsubsection{Shape violation $\texttt{VIO}^{(T,Q)}(sel,\varphi)$}\label{sec:shape_vio}
The goal of a normalized $\texttt{VIO}^{(T,Q)}(sel,\varphi)$ task is to make a selected node $u\models sel$ violate the atomic shape $\varphi$. The required graph operations depend on whether $\varphi$ is negated:
\begin{itemize}[left=0pt, align=left]
    \item If $\varphi$ is a \textbf{non-negated} atomic shape, the goal is to invalidate $\varphi$.
    \item If $\varphi$ is a \textbf{negated} atomic shape of the form $\neg\beta$, the goal is to perform graph operations such that $\beta$ is true.
\end{itemize}
Table~\ref{tab:atomic_ops} lists the graph operations for a subset of atomic shapes; the full list is in App.~\ref{app:graph_ops}. 

Once the graph operations are determined, the provenance $(T, Q)$ guides where to apply them. 
Since all the traces in $T$ end at the selected node $u$ by construction, we modify the final hop of these traces according to the mode $Q$. If $Q=\exists$, we modify only one trace in $T$; if $Q=\forall$, we modify all traces in $T$.


\noindent \textbf{Running Example.} Assume after DFS and de-duplication on the expansion graph in Fig.~\ref{fig:tree} we have collected the shape violation task
\texttt{VIO$^{(T_4, \forall)}$($=$``123'',$\text{len}\geq3$)}, where 

\noindent$T_4=$\raisebox{-0.7\baselineskip}{$\begin{aligned}
            \{&\langle \texttt{{PaperABC,reviewedBy,Ada,workPhone,``123''}} \rangle,\\[-2pt]
            &\langle\texttt{{PaperABC,reviewedBy,Ada,cellPhone,``123''}} \rangle\}.
        \end{aligned}$}
        
The selected node $u$ is \texttt{``123''} according to the selector $=$\texttt{``123''}. The shape $\texttt{len}\geq3$ is non-negated, so we must make it false. As defined in Table~\ref{tab:atomic_ops}, we can then replace \texttt{``123''} with a new value that has a length less than three (e.g., $\texttt{``23''}$). The traces in $T_4$ identify two final hops, the edge traversal corresponding to \texttt{workPhone} and \texttt{cellPhone}. Since the mode is $Q=\forall$, a replace operation is constructed for each hop, resulting in two operations:
    \begin{itemize}[left=0pt, align=left]
        \item $rep$(\texttt{Ada,workPhone,``123'',``23''})
        \item $rep$(\texttt{Ada,cellPhone,``123'',``23''}).
    \end{itemize}
\begin{table}
    \captionsetup{aboveskip=2pt}
    \setlength{\belowcaptionskip}{-1pt}
    \centering
    \mycaption{Graph operations for $\texttt{VIO}^{(T,Q)}(sel,\varphi)$, given a node $u \models sel$ and the final hop $(w,p,u)$ of a trace.}{$u'$ is generated by editing $u$ (e.g., adding, removing, transposing characters, or combinations thereof) to violate $\varphi$.}
    \begin{tabular}{cc}
        \toprule
         $\varphi$ & Graph Operations\\\midrule
         $=y$& $rep(w,p,u,u')$, where $u'\neq u $ \\
         $\neg=y$& $rep(w,p,u,y)$\\
         $\text{len}\geq y$&  $rep(w,p,u,u')$, where $len(u') < y$  \\ 
         $\neg\text{len}\geq y$ & $rep(w,p,u,u')$, where $len(u') \geq y$ \\
         \bottomrule
    \end{tabular}
    \label{tab:atomic_ops}
\end{table}

\subsubsection{Path violation $\texttt{VIO}_\pi(u,v)$}\label{sec:path_vio}
The goal of a $\texttt{VIO}_\pi(u,v)$ task is to remove all $\pi$-matching traces between $u$ and $v$. The process is a recursion on the structure of $\pi$. For a trace $\tau$ that matches $\pi$, the algorithm finds and removes an essential edge in $\tau$ to break it. The recursion unfolds as follows:

\begin{itemize}[left=0pt, align=left]
    \item \textbf{Base case, predicate $p$:} a trace $\tau$ matches $p$ if $(u,p,v)$ is present. To violate this, we remove the edge: $-e(u,p,v)$ if $p \in \mathcal{E}$, or $-a(u,p)$ if $p \in \mathcal{K}$.
    \item \textbf{Alternative $(\pi_1 \cup \pi_2)$: } a trace $\tau$ matches $\pi_1 \cup \pi_2$ if it matches either $\pi_1$ or $\pi_2$. To ensure no match exists, we must eliminate both possibilities. Thus, the algorithm recursively calls itself on $(u,\pi_1,v)$ and $(u,\pi_2,v)$.
\end{itemize}

The remaining cases, $(\pi_1^-)$, $(\pi_1 \cdot \pi_2)$, $(\pi_1^*)$, and $(\pi_1?)$ follow the same recursive principle and are detailed in App.~\ref{app:pathbreak}.

\subsubsection{Qualified satisfaction $\texttt{SAT}_\pi^n(sel,\varphi)$}\label{sec:qMCount}
The goal of a $\texttt{SAT}_\pi^n(sel,\varphi)$ task is to ensure that a selected node $u \models sel$ has exactly $n$ distinct $\pi$-targets satisfying $\varphi$. Let $m$ be the current number of such targets, the goal is to add $k = n - m$ new targets that satisfy $\varphi$. This is achieved in two steps:
\begin{itemize}[left=0pt, align=left]
    \item \textbf{Find qualified targets:} Search the graph $\mathcal{G}$ for nodes that satisfy $\varphi$ but not already $\pi$-targets of $u$.
    \item \textbf{Construct $\pi$-matching trace:} Create a new $\pi$-matching trace to connect $u$ to each found target.
\end{itemize}
Further details are in App.~\ref{app:qMCount}.

After introducing how violation expressions are mapped to graph operations. We continue with our running example to illustrate the test case generation process.

\noindent\textbf{Running Example.} Recall $\mathcal{P}_2=\{({\color{C6}1},{\color{C6}2},{\color{C6}3},{\color{C6}4}),({\color{C6}5},{\color{C6}6})\}$ from our running example. Its leaf nodes (nodes {\color{C6}4} and {\color{C6}6}) are expressions:
\begin{enumerate}[ left=0pt, align=left]
    \item [(E1)] \parbox[t]{\linewidth}{\texttt {VIO$^{(T_4, \forall)}$($=$``123'',$\text{len}\geq3$)VIO$_{\texttt{reviewedBy}}$(PaperABC,Bob)}}
    \item [(E2)] \parbox[t]{\linewidth}{\texttt{SAT$_\texttt{reviewedBy}^4$($\in$\texttt{Paper},$\geq_1\pi.\text{len}\geq 3$)}}
\end{enumerate}
Take (E1) as an example. This expression applies two sets of operations to $\mathcal{G}$:
\begin{itemize}[left=0pt, align=left]
    \item \texttt{VIO$^{(T_4, \forall)}$($=$``123'',$\text{len}\geq3$)} replaces the final hop of all traces in $T_4$ with a value failing the shape, e.g., \texttt{``23''}. This yields the following graph operations, disqualifying \texttt{Ada} as a \texttt{\textbf{Reviewer}}:
    \begin{itemize}[label=\textbullet,left=0pt, align=left]
        \item $rep$(\texttt{Ada,workPhone,``123'',``23''})
        \item $rep$(\texttt{Ada,cellPhone,``123'',``23''})
    \end{itemize}
    \item \texttt{VIO$_{\texttt{reviewedBy}}$(PaperABC,Bob)} is the base case of the path violation algorithm, disconnecting \texttt{Bob} from \texttt{PaperABC}, yielding $-e$(\texttt{PaperABC,reviewedBy,Bob}).
\end{itemize}
Validating $\mathcal{G}^\prime$ against $\mathcal{S}$ shows that both \texttt{PaperA} and \texttt{PaperABC} violate $(sel_1,\varphi_1)$; their offending subgraphs are in Fig.~\ref{fig:off_sub}. After applying E2 to a separate copy of $\mathcal{G}$, we obtain one additional invalid graph, yielding two test cases for the running example.



\begin{figure}[t]
    \captionsetup{aboveskip=2pt}
    \setlength{\belowcaptionskip}{-1pt}
    \centering
    \begin{tikzpicture}[scale=1, >=Latex, line cap=round]
\colorlet{reviewedbycolor}{C0}
\colorlet{aedgecolor}{C1}     
\colorlet{authoredgecolor}{C2}
\colorlet{workedgecolor}{C3}
\colorlet{celledgecolor}{C4}

\tikzset{
  reviewedby/.style={draw=reviewedbycolor,fill=reviewedbycolor,thick,->},
  aedge/.style={draw=aedgecolor,fill=aedgecolor,thick,->},
  authoredge/.style={draw=authoredgecolor,fill=authoredgecolor,thick,->},
  albowededge/.style={draw=aedgecolor,thick,->, shorten >=0.1pt},
  workedge/.style={draw=workedgecolor,fill=workedgecolor,thick,->},
  celledge/.style={draw=celledgecolor,fill=celledgecolor,thick,->}
}

\begin{scope}[every node/.style={ellipse,draw,font=\ttfamily\small,minimum width=0.3cm, minimum height=0.5cm,inner sep=0.1pt}]

    \node[](At) at (-2.5,1) {PaperA};
    \node[fill=C1!20](papera) at (-4.3,1) {Paper};
    \node[] (A) at (-2.5,-0.2) {Ada};
    \node[] (ap)  at (-4.4,0.2) {``23''};

    \node[](ABC) at (1.3,1) {PaperABC};
    \node[fill=C1!20](paperabc) at (-0.7,1) {Paper};
    \node[] (A_) at (0.8,-0.2) {Ada};
    \node[] (ap_) at (-0.7,0) {``23''};
    \node[] (C) at (1.8,-0.2) {Cole};
\end{scope}

\node[draw, dashed, rounded corners, shape=rectangle, inner sep=2pt,
      fit=(At)(A)(ap)(papera),
      label={[font=\small,align=center]below:
      $\mathsf{OffSub}_{\mathcal{G}'}(\texttt{PaperA}, (sel_1,\varphi_1))$}] {};

\node[draw, dashed, rounded corners, shape=rectangle, inner sep=2pt,
      fit=(ABC)(A_)(C)(ap_)(paperabc),
      label={[font=\small,align=center]below:
      $\mathsf{OffSub}_{\mathcal{G}'}(\texttt{PaperABC}, (sel_1,\varphi_1))$}] {};

\path [reviewedby] (At) edge (A);
\path [reviewedby] (ABC) edge (A_);
\path [reviewedby] (ABC) edge (C);


\path [aedge] (At) edge (papera);
\path [aedge] (ABC) edge (paperabc);
\path [workedge] (A) edge[bend right] (ap);
\path [celledge] (A) edge[bend left] (ap);
\path [workedge] (A_) edge[bend right] (ap_);
\path [celledge] (A_) edge[bend left] (ap_);

\end{tikzpicture}
    \mycaption{Offending subgraphs of \texttt{PaperA} and \texttt{PaperABC} from applying E1 on the graph of the running example.}{}
    \label{fig:off_sub}
\end{figure}

\subsubsection{Eligibility Conditions and Coverage}\label{sec:coverage}
While we have outlined the procedure for each task, certain preconditions must be met for a task to be feasible. A task can be impossible for two reasons: 
\begin{itemize}[left=0pt, align=left]
    \item The task is \textbf{logically \revwayne{invalid}}. For example, the shape violation task $\texttt{VIO}^{(T,Q)}(sel,\top)$ is impossible because it requires to violate a logical tautology ($\top$).
    \item The task is \textbf{operationally impossible}. For example, the shape violation task $\texttt{VIO}^{(T,Q)}(sel,=u)$, $T$ being a zero-length trace $\{\langle u\rangle\}$, is impossible, because our graph operations (Def.~\ref{def:ops}) are defined to modify edges and attributes, but a zero-length trace has no hop to modify.
\end{itemize}
To handle these cases, we introduce the concept of \textit{eligibility} and use it as a filter: eligible tasks are mapped to graph operations, while ineligible tasks are discarded. A task is eligible only if it is both logically valid and operationally feasible. We now describe these two conditions.

\noindent\textbf{Logical Precondition.} A task must be logically valid.
\begin{itemize}[left=0pt, align=left]
    \item For a shape violation task $\texttt{VIO}^{(T,Q)}(sel,\varphi)$: $sel$ must not logically entail\footnote{Non-entailment includes cases such as: 
        (i) selector matches at least one node ($\exists u \ (u\models sel)$), 
        (ii) $\varphi$ is not trivial ($\varphi \neq \top$ or any tautology), and 
        (iii) $\varphi$ is not implied by $sel$ (e.g., if $sel = \varphi_1 \wedge \varphi_2$ and $\varphi = \varphi_1$, then $\varphi$ is implied by $sel$ and thus fails non-entailment).} $\varphi$, since we can not violate a constraint that is guaranteed to be true.
    \item For a qualified satisfaction task $\texttt{SAT}_\pi^n(sel,\varphi)$: $sel$ must select at least one node and $\varphi$ must not be a contradiction, since we can not satisfy an unsatisfiable logical formula.
\end{itemize}
\noindent\textbf{Operational Feasibility.} A task is operationally infeasible when it must be induced via a zero-length trace $\langle u \rangle$, as this requires modifying the node $u$ itself, which our edge-based graph operations can not do. To understand when this occurs for shape violation tasks, we first need to define intrinsic shapes.
\begin{definition}[\textbf{Intrinsic Shapes}]
    A shape $\varphi$ is intrinsic if its truth value for a node $u$ is independent of the graph’s edges and attributes. For example, $=y$ is intrinsic, as it only depends on the node $u$ itself.
\end{definition}
We then detail when a task is operationally infeasible:
\begin{itemize}[left=0pt, align=left]
    \item For a shape violation task $\texttt{VIO}^{(T,Q)}(sel,\varphi)$, the task is infeasible when $\varphi$ is intrinsic and either of the following holds:
    \begin{itemize}[left=0pt, align=left]
        \item $Q = \forall$ and $T$ contains at least one zero-length trace (as the ``modify all'' requirement can not be met), or
        \item $Q = \exists$ and all traces in $T$ are zero-length traces (as there is no hop to modify).
    \end{itemize}
    \item For a path violation task $\texttt{VIO}_\pi(u,v)$, the task is infeasible when the zero-length trace $\langle u\rangle$ matches $\pi$, which occurs when $u = v$ and $\pi$ permits zero-length traces (e.g., via $\pi_1^*$ or $\pi_1?$ operators).
\end{itemize}
We now state the formal definition of eligibility.

\begin{definition}[\textbf{Eligibility of Tasks}]\label{def:eligibility}
A task is \emph{eligible} if and only if it meets the following conditions:
\begin{enumerate}[left=0pt, align=left]
    \item \textbf{Logical precondition}: 
    \begin{itemize}[left=0pt, align=left]
        \item For $\texttt{VIO}^{(T,Q)}(sel,\varphi)$, $sel$ does not logically entail $\varphi$.
        \item For $\texttt{SAT}_\pi^n(sel,\varphi)$, $sel$ matches at least one node and $\varphi$ is not a contradiction.
    \end{itemize}
    \item \textbf{Operational feasibility:}
    \begin{itemize}[left=0pt, align=left]
        \item For $\texttt{VIO}^{(T,Q)}(sel,\varphi)$, the task is feasible unless $\varphi$ is intrinsic and $T$ contains only zero-length traces (or at least one if $Q=\forall$).
        \item For $\texttt{VIO}_\pi(u,v)$, the task is feasible unless $\langle u\rangle$ matches $\pi$.
    \end{itemize}
\end{enumerate}
After ineligible tasks are discarded, the following theorem states that every remaining task yields a set of graph operations.
\end{definition}
\begin{restatable}[\textbf{Coverage Guarantee}]{theorem}{CoverageGuarantee}\label{lem:coverage}
For any schema $\mathcal{S}$ and valid graph $\mathcal{G}$, every reachable and eligible normalized shape violation task, every eligible path violation task, and every eligible qualified satisfaction task will generate a non-empty set of graph operations.
\end{restatable}
\begin{pf}
See App.~\ref{app:coverage_proof} for the proof.
\end{pf}

\section{Repair Systems}
\label{sec:repairsys}
We introduce three representative automated graph repair systems: a symbolic method using answer set programming, a supervised-learning based method, and an LLM-based method. We focus on fully automated systems to ensure scalability. Consequently, we exclude Human-in-the-Loop methods (e.g., Pachera et al.~\cite{pachera2025user}) and other interactive approaches.

\noindent \textbf{Answer Set Programming-Based Repair.}
Ahmetaj et al.~\cite{ahmetaj2022repairing} repair violations using a logic-based framework called Answer Set Programming (ASP)~\cite{lifschitz2019answer}. ASP is a form of declarative programming where problems are modeled as logical programs to be solved by powerful  and general-purpose solvers, e.g., Clingo~\cite{gebser2019multi}.

Ahmetaj et al.~\cite{ahmetaj2022repairing} encode the schema, invalid graph, and graph operations into an answer set program. The program also includes an optimization objective to find repairs that minimize the number of graph operations, where each atomic operation (e.g., adding an edge) counts as one operation. This approach directly encodes the heuristic that the most desirable repair is one that alters the graph the least. The result is an \textit{answer set}, which represents the complete collection of equally minimal repairs. From this set of potential solutions, the authors' implementation~\cite{shacl-repairs}  selects the first answer returned by the solver to serve as the final repair.

A critical consideration for this approach is its computational cost. The solving of the encoded answer set program has a worst-case time complexity that is exponential in the combined size of the graph and the schema, which limits its scalability to large datasets.

\noindent \textbf{Neural Knowledge Base Repair.} 
Pellissier Tanon et al.~\cite{pellissier2021neural} propose a supervised learning approach to learn repairs directly from historical graph edits. They train a neural network model with custom architecture to predict the correct repair for a given schema violation. The training data of the model consists of schema violations and the corresponding human-provided repairs. Once trained, it can predict repairs by identifying patterns in the graph and context surrounding the error. The primary limitation of this approach is its dependency on a substantial volume of training data, which may not be available for the specific graph one intends to repair. The authors' implementation~\cite{neural} leveraged the extensive edit history of Wikidata~\cite{vrandevcic2014wikidata}, a crowd-sourced graph database, to train their model.

\noindent\textbf{LLM-Based Repair.}
LLMs have recently emerged as a tool to interact with graphs~\cite{terdalkar2025graph,arnaout2022utilizing}. \revwayne{Their primary limitation is their black-box nature, as performance is highly sensitive to the choice of model and the design of the prompt. This requires extensive empirical testing, a process known as prompt engineering~\cite{schulhoff2024prompt}. Our framework supports this need for systematic evaluation; its ability to generate novel violations on private and unpublished graphs also mitigates the risk of data contamination from LLM training corpora~\cite{kiela2021dynabench}.} For our evaluation, we selected four models and accessed them in January 2025: GPT4o (OpenAI), Claude 3.0 Opus (Anthropic), Gemini 1.5 Pro (Google), and Llama 3.1 405B (Meta, via AWS Bedrock). 

\revwayne{To demonstrate how our framework provides systematic evaluation of the impact of different prompts}, we construct a modular prompting template to provide necessary context, including Schema Context ($\mathcal{C}_m$) and Graph Context ($\mathcal{C}_g$). To interface with the LLM, all graph and schema information is serialized into the Turtle format~\cite{turtle}, and the LLM is instructed to generate repairs in SPARQL syntax~\cite{SPARQL}. See App.~\ref{app:llm_repair} for the full prompt template. 
We define three  configurations for both the schema and graph contexts, yielding nine variations, as detailed below.

\paragraph{Schema Context $\mathcal{C}_m$.} The amount of schema information provided to the model is varied across three levels:
\begin{enumerate}[left=0pt, align=left]
    \item \textbf{Entire Schema ($M$)}: Provides the full schema $\mathcal{S}$. This offers maximum context but may exceed the LLM’s context window for large schemas.
    \item \textbf{Minimal Schema ($S$)}: Provides only the offended constraint $(sel,\varphi)$. This improves scalability by focusing on the violated constraint.
    \item \textbf{Annotated Schema ($S_n$)}: Extends the minimal schema ($S$) with natural language descriptions (e.g., \texttt{rdfs:label}). This lets us assess the effect of natural language descriptions.
\end{enumerate}

\paragraph{Graph Context $\mathcal{C}_g$.} Similarly, the amount of graph information is varied across three levels:
\begin{enumerate}[left=0pt, align=left]
    \item \textbf{Entire Graph ($G$)}: Provides the full graph $\mathcal{G}'$. Like the entire schema, this maximizes context at the cost of scalability.
    \item \textbf{Focused Subgraph ($F$)}: Provides only the offending subgraph $\mathsf{OffSub}_{\mathcal{G}'}(u, (sel,\varphi))$ of offending focus $u$ and offended constraint $(sel,\varphi)$. This isolates the part of the graph where the error occurs. For cardinality constraints, $\varphi = \geq_n \pi.\varphi'$, we also include context around nodes that satisfy $\varphi'$ to provide sufficient context. Specifically, $F$ also includes the subgraph visited during the validation against $\varphi'$. For instance, in our running example where (E1) is applied so that \texttt{PaperABC} violates $(sel_1,\varphi_1)$, the subgraphs for \texttt{Bob} and \texttt{Dan} are included because they satisfy the \texttt{\textbf{Reviewer}} constraint. The full context is visualized in Fig.~\ref{fig:focused_subgraph}.
    \item \textbf{Focused Subgraph with Example ($F^+$)}: Augments the focused subgraph ($F$) with a positive example. In addition to the subgraph showing the error, it includes another subgraph where the same constraint is  satisfied, acting as a one-shot example to guide the LLM.
\end{enumerate}
\begin{figure}[t]
    \captionsetup{aboveskip=2pt}
    \setlength{\belowcaptionskip}{-1pt}
    \centering
    \begin{tikzpicture}[scale=1, >=Latex, line cap=round]
\colorlet{reviewedbycolor}{C0}
\colorlet{aedgecolor}{C1}     
\colorlet{authoredgecolor}{C2}
\colorlet{workedgecolor}{C3}
\colorlet{celledgecolor}{C4}

\tikzset{
  reviewedby/.style={draw=reviewedbycolor,fill=reviewedbycolor,thick,->},
  aedge/.style={draw=aedgecolor,fill=aedgecolor,thick,->},
  authoredge/.style={draw=authoredgecolor,fill=authoredgecolor,thick,->},
  albowededge/.style={draw=aedgecolor,thick,->, shorten >=0.1pt},
  workedge/.style={draw=workedgecolor,fill=workedgecolor,thick,->},
  celledge/.style={draw=celledgecolor,fill=celledgecolor,thick,->}
}

\begin{scope}[every node/.style={ellipse,draw,font=\ttfamily\small,minimum width=0.3cm, minimum height=0.5cm,inner sep=0.1pt}]

    \node[](ABC) at (1,1) {PaperABC};
    \node[fill=C1!20](paper) at (-1,1) {Paper};
    \node[] (A_) at (0.5,-0.2) {Ada};
    \node[] (ap_) at (-1.,0) {``23''};
    \node[] (C) at (1.5,-0.2) {Cole};

    \node[] (B) at (3.5,1) {Bob};
    \node[] (four) at (5,1) {\texttt{``456''}};
    \node[draw=none] () at (4.35,0.4) {$\mathsf{OffSub}_{\mathcal{G}'}(\texttt{Bob}, (=\texttt{Bob},\texttt{\textbf{Reviewer}}))$};

    \node[] (D) at (3.5,-0.2) {Dan};
    \node[] (seven) at (5,-0.2) {\texttt{``789''}};
        \node[draw=none] () at (4.35,-0.8) {$\mathsf{OffSub}_{\mathcal{G}'}(\texttt{Dan}, (=\texttt{Dan},\texttt{\textbf{Reviewer}}))$};


\end{scope}

\node[draw, dashed, rounded corners, shape=rectangle, inner sep=2pt,
      fit=(B)(four)] {};
\node[draw, dashed, rounded corners, shape=rectangle, inner sep=2pt,
      fit=(D)(seven)
      ] {};

\node[draw, dashed, rounded corners, shape=rectangle, inner sep=2pt,
      fit=(ABC)(A_)(C)(ap_)(paper),
      label={[font=\small,align=center]below:
      $\mathsf{OffSub}_{\mathcal{G}'}(\texttt{PaperABC}, (sel_1,\varphi_1))$}] {};


\path [reviewedby] (ABC) edge (A_);
\path [reviewedby] (ABC) edge (C);    
\path [celledge] (B) edge (four);
\path [workedge] (D) edge (seven);

\path [aedge] (ABC) edge (paper);

\path [workedge] (A_) edge[bend right] (ap_);
\path [celledge] (A_) edge[bend left] (ap_);

\end{tikzpicture}
    \mycaption{Focused subgraph context ($F$) of \texttt{PaperABC} violating $(sel_1,\varphi_1)$ after E1 is applied in the running example.}{}
    \label{fig:focused_subgraph}
\end{figure}

\section{Evaluation Metrics}
\label{sec:metrics}
We define three tiered metrics to evaluate the quality of a repair $\mathfrak{R}$. We formalize the metrics as binary indicator functions. Each metric is assessed only if all lower-tier metrics are satisfied. The metrics are listed below in order of increasing stringency. For the
$i$-th test case. 
\begin{enumerate}[left=0pt, align=left]
    \item \textbf{Syntactic Validity ($V_{syn}(i)$):} $\mathfrak{R}$ must be parsed without error. $V_{syn}(i) = 1$ if $\mathfrak{R}$ can be parsed by a standard library (e.g., RDFLib~\cite{rdflib}) without error. Note: an empty repair is considered syntactically valid; else, $V_{syn}(i) = 0$.
    \item \textbf{Semantic Validity ($V_{sem}(i)$):} After a syntactically valid $\mathfrak{R}$ is applied to $\mathcal{G}'$ to obtain $\mathcal{G}'_{\mathfrak{R}}$, the violations should be eliminated. Thus, $\mathfrak{R}$ is semantically valid ($V_{sem}(i) = 1$) if $\mathcal{G}^\prime_{\mathfrak{R}} \models \mathcal{S}$; else, $V_{sem}(i) = 0$.
    \item \textbf{Isomorphism ($V_{\mathrm{iso}}(i)$):} Semantic validity does not ensure that the original graph \(\mathcal{G}\) is recovered. A semantically-valid $\mathfrak{R}$ is further tested for isomorphism. If $\mathcal{G}^\prime_\mathfrak{R}$ and $\mathcal{G}$ are isomorphic~\cite{hogan2015skolemising}, then $\mathfrak{R}$ is isomorphic ($V_{\mathrm{iso}}(i)=1$).
\end{enumerate}
In Table 4, we report the success rate as the score for each metric, calculated as the average of these indicators over the number of test cases $M$ ($\frac{1}{M}\sum_{i=1}^M V(i)$). For example, ASP generates 140 semantically valid repairs out of 144 test cases. Thus, its semantic validity score is $\frac{140}{144}\approx 97.22\%$.

\noindent \textbf{Remark.} Semantic validity and isomorphism formalize different notions of a ``good'' repair. Semantic validity captures correctness with respect to the schema; isomorphism checks whether the repair reverses the VE. As Fig.~\ref{fig:mot_ex} illustrates, there can be multiple semantically valid repairs for a violation. To determine the correct repair, one often needs to inspect the surrounding context of the error~\cite{pellissier2021neural} or consult external knowledge, such as a human expert~\cite{pachera2025user}. We therefore report both metrics so that practitioners can prioritize schema conformance or isomorphism as appropriate. Using the original graph as a reference is a commonly adopted ``silver standard'' in knowledge-graph evaluation~\cite{pachera2025user,paulheim2016knowledge}.

\begin{figure*}[t]
    \captionsetup{aboveskip=2pt}
    \setlength{\belowcaptionskip}{-1pt}
    \centering
    \includegraphics[width=1.\linewidth]{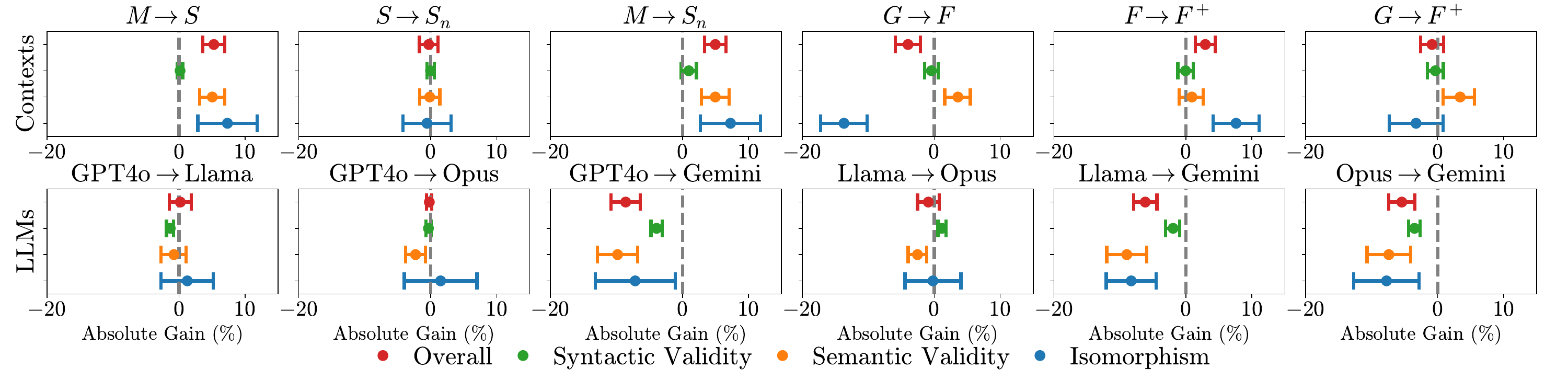}
    \caption{\textbf{Estimated effects of context selection and LLM selection.}}
    \label{fig:lmm}
\end{figure*}
\section{Evaluation Results}
\label{sec:results}
\revwayne{We evaluate the repair systems on three datasets of RDF graphs with corresponding SHACL schemas.  \textbf{Brick Ontology}~\cite{balaji2016brick} standardizes building assets and their relationships. Using the 1.3 release, we construct eight graphs and schemas based on hardware specifications in the ASHRAE Guideline 36~\cite{ashrae201836}. On average, the graphs contain 35 edges and the schemas contain 108 edges. In total, our framework yields 144 test cases.
\textbf{LUBM Ontology}~\cite{guo2005lubm} is a synthetic dataset modeling the relationships between entities such as departments and professors in a university. We adapt the graph and schema from Figuera et al.~\cite{figuera2021trav,tracedSPARQL}. The graph contains 229 edges and the schema contains 207 edges. Our framework yields 70 test cases.
\textbf{QUDT Ontology}~\cite{qudt} provides a unified conceptual model for physical quantities and units. In the 2.1 release, the graph contains 81,293 edges and the schema contains 2,642 edges.  Our framework yields 134 test cases. As this dataset is two to three orders of magnitude larger than the others, it serves to test the scalability of each repair method. }

Our evaluation proceeds in two main stages. First, \revwayne{since LLMs are sensitive to prompt design and model choices, we leverage our framework to analyze the behavior of LLMs as repair systems. 
Sec.~\ref{sec:impact} characterizes the performance variance of the LLM, establishing its performance upper bound.} Second, to demonstrate our framework's diagnostic ability across diverse approaches, we evaluate and compare the performance of all three repair methods in Sec.~\ref{sec:baselines}: answer set programming based-repair (ASP), neural knowledge base repair (Neural), and \revwayne{the performance bounds of} LLM-based repair (LLM).


\subsection{Impact of Context and Model Choice on LLM-Based Repair}\label{sec:impact}
We analyze the performance of four different LLMs across nine prompting strategies to measure the effects of both context and model choice on repair quality. Due to budget constraints, prompts that cost more than \$0.50 are omitted. These exclusions only occur on the large-scale QUDT dataset and affect the following prompting strategies:
\begin{itemize}[left=0pt, align=left]
    \item The $MG,SG$, and $S_nG$ prompts for all LLMs.
    \item The $MF$ and $MF^+$ prompts for Claude 3.0 Opus, whose granular tokenization results in higher costs.
\end{itemize}
The fact that all the affected prompting strategies involve either the entire schema ($M$) or the entire graph ($G$) shows the poor scalability of $M$ and $S$ contexts. A detailed cost analysis for all models and prompts is available in App.~\ref{app:llm_repair}.



\noindent\textbf{Schema Context.}
We investigate the following questions: (i) Does using $M$ produce a different result from $S$? (ii) Does using $S$ produce a different result from $S_n$? (iii) Does using $M$ produce a different result from $S_n$? 
\revwayne{Fig.~\ref{fig:lmm} reports the absolute gain in percentage points ($Score_{new} - Score_{old}$). For example, if the prompt strategy $M$ achieves $80\%$ semantic validity and the prompt strategy $S$ achieves $85\%$, we report a semantic validity absolute gain of $+5\%$ for $M\to S$, where the 95\% confidence interval is estimated using linear mixed-effects model (LMM)~\cite{pinheiro2000linear}. Further details on the statistical analysis are in App.~\ref{app:stats}.} 
Metrics whose intervals do not cross 0\% are statistically significant ($p<0.05$). 

Switching from $M$ to $S$ improves all metrics except syntactic validity, suggesting that full-schema prompts may overload or distract the LLM, while concise prompts focusing on the violated constraint yield better repairs. $S$ and $S_n$ perform similarly, and both outperform $M$; however, $S$ is more efficient as it requires fewer tokens.

\noindent\textbf{Graph Context.}  In Fig.~\ref{fig:lmm}, we also analyze the impact of different graph contexts. Switching from $G$ to $F$ reduces the overall performance, especially for the isomorphism metric, but improves semantic validity. This indicates that while removing extraneous information helps the LLM focus on eliminating violations (thus improving semantic validity), it also deprives the model of the context needed to generate repairs that closely match the original graph. Unlike schema context, where reducing information ($M \to S$) is beneficial, reducing graph context harms the stricter metrics, suggesting that information not directly used in validation still aids the LLM in inferring suitable repairs. We also evaluate the $F^+$ context, which adds an example. The results show a significant improvement from $F$ to $F^+$. Furthermore, $F^+$ achieves overall performance that is statistically insignificant from $G$, but without its prohibitive cost. Overall, $F^+$ offers a balance between performance and cost.

\noindent\textbf{Summary of Context Impact.}  
The most effective and scalable results are achieved with schema context $S$ (minimal schema) and graph context $F^+$ (focused subgraph with example), which together provide essential information while controlling cost. 

\noindent\textbf{Model Choice.}  In Fig.~\ref{fig:lmm}, we compare different model choices. Replacing GPT4o, Llama, or Opus with Gemini all significantly degrade performance, while the other three perform similarly.

\subsection{Comparison of Repair Systems}
\label{sec:baselines}
We now conduct a comparative evaluation of the three paradigms:  answer set programming-based repair (ASP), neural knowledge base repair (Neural), and LLM-based repair (LLM).

\subsubsection{Evaluation Setup}
To ensure a fair comparison, we tailor the evaluation setup for each repair system. This involves adapting the reference implementations and defining how to handle their output formats.

\noindent \textbf{Answer Set Programming-Based Repair.}
The reference implementation~\cite{shacl-repairs} arbitrarily selects the first solution from its answer set, which may not represent the optimal performance of ASP. To provide a more thorough assessment, we evaluate the entire answer set of repairs the solver generates. We introduce two metrics for this purpose:

\begin{itemize}[left=0pt, align=left]
    \item Isomorphic Hit (binary): Equals 1 if the answer set contains the correct (isomorphic) repair, and 0 otherwise.
    \item Isomorphic Precision: Equals $\frac{\mathrm{Isomorphic Hit}}{|\mathrm{AnswerSet}|}$. This measures how efficiently the solver returns the isomorphic repair.
\end{itemize}

\noindent \textbf{Neural Knowledge Base Repair.}
The neural knowledge base repair presents two challenges. First, it requires a training set of the graph being repaired, which is not available for our datasets. We address this by training the authors' official model using their provided Wikidata edit history data set~\cite{neural}. Second, since the model is trained with Wikidata, it expects Wikidata-specific violation types, not SHACL. We therefore implement an adapter to translate SHACL violations into the format required by the model.

\noindent\textbf{LLM-Based Repair.}
For the LLM-based approach, we exclude prompting strategies with excessive costs, leaving only four scalable strategies: $SF$, $S_nF$, $SF^+$, and $S_nF^+$. We report:
\begin{enumerate}[left=0pt, align=left]
    \item \textbf{Subset average:} The mean score across all four strategies and all four LLMs, representing the typical performance for the LLM-based approach.
    \item \textbf{Best configuration:} The performance of Llama3.1 with $SF^+$, which was identified as one of the top-performers in Sec.~\ref{sec:impact}. This represents the upper-bound for the LLM-based approach.
\end{enumerate}

\begin{table}[t]
\captionsetup{aboveskip=2pt}
    \setlength{\belowcaptionskip}{-1pt}
\centering
\mycaption{Average scores ($\%$) of the repair systems.}{``Brick'' averages over the Brick dataset; ``All'' averages over all datasets. Higher scores are better for all metrics: syntactic validity (Syn. Val.), semantic validity (Sem. Val.), isomorphism (Iso.), isomorphic hit (Iso. Hit), and precision (Iso. Prec.).}
\label{tab:merged_results}
\begin{tabularx}{\linewidth}{Y l Y Y Y Y Y}
\toprule
Scope & Method & Syn. Val. & Sem. Val.  & Iso. & Iso. Hit & Iso. Prec. \\
\midrule
\multirow{4}{*}{Brick} 
  & ASP~\cite{ahmetaj2022repairing} & \textbf{100.0} & \textbf{97.22} & 32.64 & 66.67 & 29.32 \\
  & Neural~\cite{pellissier2021neural} & \textbf{100.0} & 0  & 0 & -- & -- \\
  & Subset Avg &99.44 & 90.49 & 62.67 & --&--  \\

  & Llama+$SF^+$ & 99.31 & 91.67 & \textbf{68.06} & -- & -- \\
\midrule
\multirow{3}{*}{All} 
  & Neural~\cite{pellissier2021neural} & \textbf{100.0} & 10.37 & 8.36 & -- & -- \\
  & Subset Avg &96.96 & 87.95 & 52.41 &-- &--  \\
  & Llama+$SF^+$ & 96.84 & \textbf{91.38} & \textbf{58.62} & -- & -- \\
\bottomrule
\end{tabularx}
\end{table}

\subsubsection{Analysis}
\label{sec:analysis}
Preliminary experiments showed that ASP scaled poorly, failing to solve medium and large test cases (LUBM and QUDT) within several hours. We thus set a 30-minute time limit for all experiments. Under this constraint, ASP completed only the Brick ontology cases while other methods completed all test cases. Therefore, we only use Brick to assess the repair quality of ASP. 
On the Brick dataset (Table~\ref{tab:merged_results}), both ASP and Neural achieve perfect syntactic validity. Neural's score is a consequence of it consistently producing empty repairs. Most syntactic errors from the LLM are formatting issues such as missing quotation marks.

Regarding semantic validity, Neural scores $0\%$, as its empty repairs do not fix violations. This poor performance highlights its sensitivity to domain shift. In contrast, ASP is  highly effective, achieving 97.22\% by repairing 140 out of 144 test cases. 
The four remaining failures are due to a non-algorithmic issue: the implementation's handling of \textit{blank nodes}. In RDF graphs, entities are identified by identifiers and blank nodes are unnamed identifiers without a global name. The failures occurred because the reference implementation~\cite{shacl-repairs} struggled to correctly reference or create these anonymous entities in generating repairs. The discovery of this vulnerability showcases the utility of our framework in identifying practical robustness issues. Despite this implementation-related flaw, the high success rate of ASP demonstrates the power of logic-based solvers, which consistently outperform the LLMs.

The performance of ASP on isomorphism drops drastically. The isomorphic hit reveals that only $66.67\%$ of the answer sets generated contained the correct repair. Our investigation reveals that the ASP encoding excludes some possible repairs: the ASP does not encode a mechanism to connect nodes disconnected from the offending focus. Take Fig.~\ref{fig:focused_subgraph} for an example, to repair the violation of $(sel_1,\varphi_1)$ for \texttt{PaperABC}, ASP's generated answer set includes adding new \texttt{workPhone} or \texttt{cellPhone} attributes (e.g., \texttt{``new\_1475''} string) to \texttt{Ada} or \texttt{Cole}. However, connecting \texttt{PaperABC} to an existing node, \texttt{Bob} or \texttt{Dan} via \texttt{reviewedBy} is not included. Even though reusing \texttt{Bob} or \texttt{Dan} is also a minimal repair with only one graph operation. 

The second block of Table~\ref{tab:merged_results} aggregates the results across all datasets (Brick, LUBM, QUDT), excluding ASP due to timeouts. LLMs outperform Neural on all metrics except syntactic validity because most repairs produced by Neural is empty, which is reflected in the decrease in Neural's semantic validity. 


\subsubsection{Per-constraint Analysis}
\label{sec:perconstraint_analysis}
We then perform a per‑constraint analysis of the best LLM configuration (Llama3.1+$SF^+$) and ASP on the Brick dataset. Figure~\ref{fig:sem-iso} plots, for each constraint, the LLM’s performance (x‑axis) against its margin over ASP (y‑axis), for semantic validity and isomorphism. Point size indicates sample size. Since ASP completed only the Brick cases, the number of points is limited, and small samples exhibit high variance. As the metric becomes more stringent (semantic validity $\to$ isomorphism), nearly all points shift left, indicating lower LLM performance on the stricter metric.

A notable case is \texttt{path violation}, where \revwayne{both LLM and ASP are nearly perfect on semantic validity} but LLM outperforms ASP by nearly $90\%$ on isomorphism. As mentioned in Sec.~\ref{sec:analysis}, ASP can not perform repairs that reuse nodes disconnected from the offending focus, whereas for LLM, the graph context $F^+$ allows such reuse of nodes, as visualized in Fig.~\ref{fig:focused_subgraph}. \revwayne{This results in the drastic performance difference on different metrics for the same violations.} This demonstrates the diagnostic power of our framework:  \textit{pinpoint the specific constraints that a repair method fails,} \revwayne{which is achieved by controlled and isolated injection of violations to ensure attributable performance difference to specific constraint violations without the confounder of interactions between violations.}


\begin{figure}
    \captionsetup{aboveskip=2pt}
    \setlength{\belowcaptionskip}{-1pt}    
    \centering
    \includegraphics[width=\linewidth]{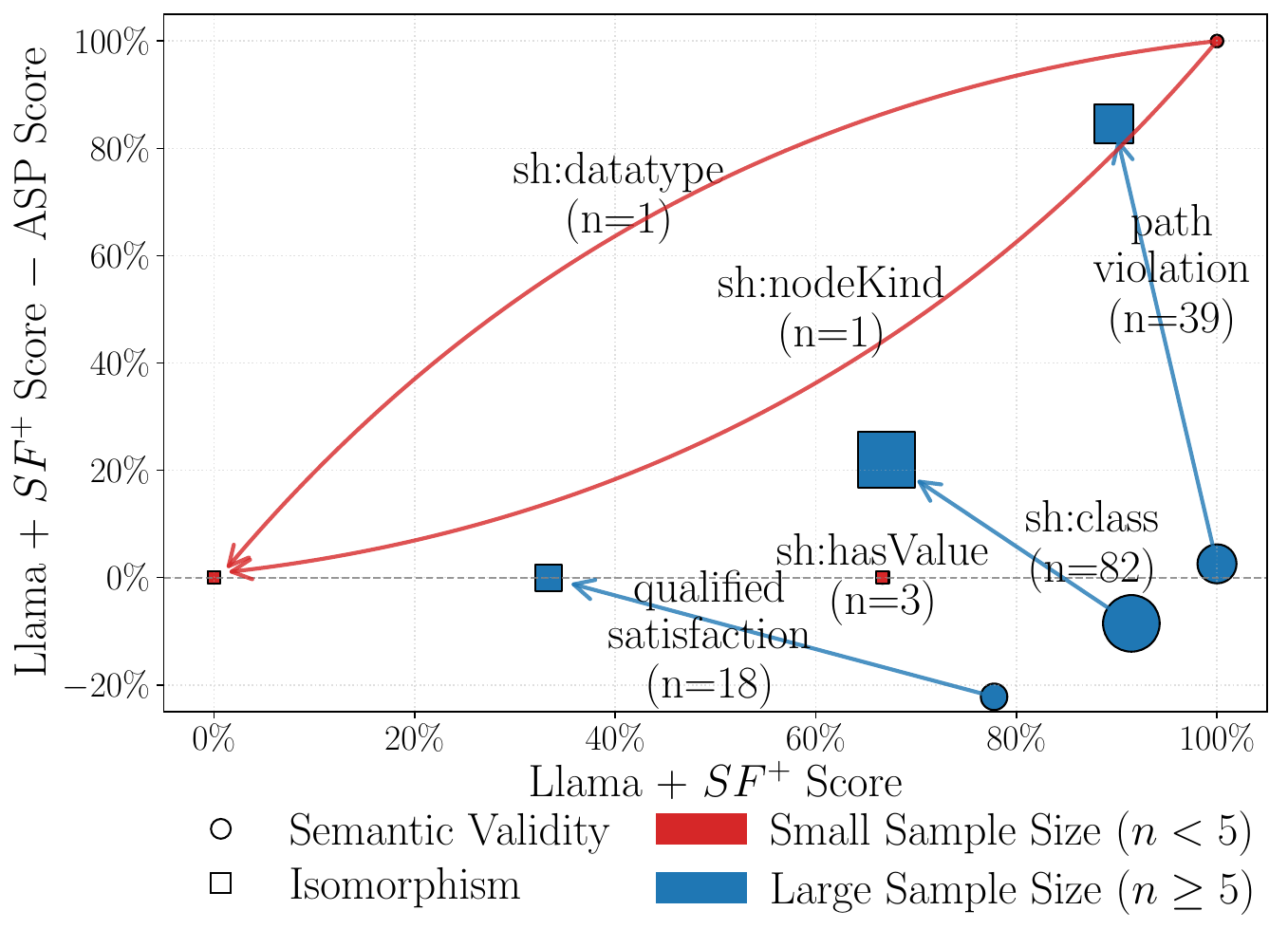}
    \mycaption{Per‑constraint performance shift from semantic validity to isomorphism, comparing LLM and ASP on the Brick dataset.} {Point size $\propto$ sample count. \texttt{sh:nodeKind} requires a node be of a given node kind, such as a blank node.}
    \label{fig:sem-iso}
\end{figure}
\section{Conclusion}
\label{sec:conclusion}

Our framework systematically generates schema violations to produce test datasets with coverage guarantees, enabling per‑constraint performance analysis. Built on an abstract rewriting system, it ensures multi‑metric evaluation across diverse repair methods. Applying the framework to symbolic, supervised learning, and LLM-based methods reveals the strengths and weaknesses of each approach. The symbolic method achieves high semantic validity on small test cases; the supervised-learning based method is sensitive to domain shift; and for LLM-based methods, they are generally effective without dominating across all metrics, with their performance and scalability being highly dependent on the contextual information in the prompt. 

A key contribution of our framework is the ability to characterize the specific types of repairs a system can and can not perform. For instance, our analysis revealed a limitation in the symbolic method: its inability to reuse disconnected nodes. 

One limitation of our framework is its inability to handle recursive schemas. However, note that while recursion is supported by the ShEx standard, it remains undefined in SHACL. 
Future work will focus on applying multiple violation expressions to a single graph to emulate realistic scenarios where multiple violations coexist. Another direction is to investigate the ``ripple effect'' (as seen in Fig.~\ref{fig:off_sub}), where a single violation expression causes multiple nodes to violate a constraint. 

\section{Acknowledgment}
This work was supported by the Laboratory Directed Research and Development (LDRD) Program of Lawrence Berkeley National Laboratory, provided by the Director, Office of Science, of the U.S. Department of Energy under Contract No. DE-AC02-05CH11231.

\bibliographystyle{ACM-Reference-Format}
\bibliography{sample-base}

@inproceedings{ahmetaj2025common,
  title={Common foundations for SHACL, ShEx, and PG-Schema},
  author={Ahmetaj, Shqiponja and Boneva, Iovka and Hidders, Jan and Hose, Katja and Jakubowski, Maxime and Labra Gayo, Jose Emilio and Martens, Wim and Mogavero, Fabio and Murlak, Filip and Okulmus, Cem and others},
  booktitle={Proceedings of the ACM on Web Conference 2025},
  pages={8--21},
  year={2025}
}

@phdthesis{jakubowski2024shapes,
  title={Shapes Constraint Language: Formalization, Expressiveness, and Provenance},
  author={Jakubowski, Maxime},
  year={2024},
  school={Ph. D. Dissertation. Universiteit Hasselt and Vrije Universiteit Brussel}
}

@inproceedings{fierro2022application,
  title={Application-driven creation of building metadata models with semantic sufficiency},
  author={Fierro, Gabe and Saha, Avijit and Shapinsky, Tobias and Steen, Matthew and Eslinger, Hannah},
  booktitle={Proceedings of the 9th ACM International Conference on Systems for Energy-Efficient Buildings, Cities, and Transportation},
  pages={228--237},
  year={2022}
}

@article{fierro2020mortar,
author = {Fierro, Gabe and Pritoni, Marco and Abdelbaky, Moustafa and Lengyel, Daniel and Leyden, John and Prakash, Anand and Gupta, Pranav and Raftery, Paul and Peffer, Therese and Thomson, Greg and Culler, David E.},
title = {Mortar: An Open Testbed for Portable Building Analytics},
year = {2019},
issue_date = {February 2020},
publisher = {Association for Computing Machinery},
address = {New York, NY, USA},
volume = {16},
number = {1},
issn = {1550-4859},
url = {https://doi.org/10.1145/3366375},
doi = {10.1145/3366375},
journal = {ACM Trans. Sen. Netw.},
month = dec,
articleno = {7},
numpages = {31}
}

@inproceedings{figuera2021trav,
  title={Trav-SHACL: Efficiently validating networks of SHACL constraints},
  author={Figuera, M{\'o}nica and Rohde, Philipp D and Vidal, Maria-Esther},
  booktitle={Proceedings of the Web Conference 2021},
  pages={3337--3348},
  year={2021}
}

@inproceedings{corman2018semantics,
  title={Semantics and validation of recursive SHACL},
  author={Corman, Julien and Reutter, Juan L and Savkovi{\'c}, Ognjen},
  booktitle={The Semantic Web--ISWC 2018: 17th International Semantic Web Conference, Monterey, CA, USA, October 8--12, 2018, Proceedings, Part I 17},
  pages={318--336},
  year={2018},
  organization={Springer}
}

@article{bi2024codekgc,
  title={Codekgc: Code language model for generative knowledge graph construction},
  author={Bi, Zhen and Chen, Jing and Jiang, Yinuo and Xiong, Feiyu and Guo, Wei and Chen, Huajun and Zhang, Ningyu},
  journal={ACM Transactions on Asian and Low-Resource Language Information Processing},
  volume={23},
  number={3},
  pages={1--16},
  year={2024},
  publisher={ACM New York, NY}
}

@inproceedings{malaviya2020commonsense,
  title={Commonsense knowledge base completion with structural and semantic context},
  author={Malaviya, Chaitanya and Bhagavatula, Chandra and Bosselut, Antoine and Choi, Yejin},
  booktitle={Proceedings of the AAAI conference on artificial intelligence},
  volume={34},
  number={03},
  pages={2925--2933},
  year={2020}
}

@article{wang2022simkgc,
  title={SimKGC: Simple contrastive knowledge graph completion with pre-trained language models},
  author={Wang, Liang and Zhao, Wei and Wei, Zhuoyu and Liu, Jingming},
  journal={arXiv preprint arXiv:2203.02167},
  year={2022}
}

@article{pareti2021review,
  title={A review of SHACL: from data validation to schema reasoning for RDF graphs},
  author={Pareti, Paolo and Konstantinidis, George},
  journal={Reasoning Web International Summer School},
  pages={115--144},
  year={2021},
  publisher={Springer}
}

@article{shen2022comprehensive,
  title={A comprehensive overview of knowledge graph completion},
  author={Shen, Tong and Zhang, Fu and Cheng, Jingwei},
  journal={Knowledge-Based Systems},
  volume={255},
  pages={109597},
  year={2022},
  publisher={Elsevier}
}

@inproceedings{mihindukulasooriya2023text2kgbench,
  title={Text2kgbench: A benchmark for ontology-driven knowledge graph generation from text},
  author={Mihindukulasooriya, Nandana and Tiwari, Sanju and Enguix, Carlos F and Lata, Kusum},
  booktitle={International Semantic Web Conference},
  pages={247--265},
  year={2023},
  organization={Springer}
}

@inproceedings{duan2024scoop,
  title={SCOOP all the Constraints’ Flavours for your Knowledge Graph},
  author={Duan, Xuemin and Chaves-Fraga, David and Derom, Olivier and Dimou, Anastasia},
  booktitle={European Semantic Web Conference},
  pages={217--234},
  year={2024},
  organization={Springer}
}

@inproceedings{wright2020schimatos,
  title={Schimatos: a SHACL-based web-form generator for knowledge graph editing},
  author={Wright, Jesse and Rodr{\'\i}guez M{\'e}ndez, Sergio Jos{\'e} and Haller, Armin and Taylor, Kerry and Omran, Pouya G},
  booktitle={International Semantic Web Conference},
  pages={65--80},
  year={2020},
  organization={Springer}
}

@inproceedings{cimmino2020astrea,
  title={Astrea: automatic generation of SHACL shapes from ontologies},
  author={Cimmino, Andrea and Fern{\'a}ndez-Izquierdo, Alba and Garc{\'\i}a-Castro, Ra{\'u}l},
  booktitle={European Semantic Web Conference},
  pages={497--513},
  year={2020},
  organization={Springer}
}

@article{rabbani2023extraction,
  title={Extraction of validating shapes from very large knowledge graphs},
  author={Rabbani, Kashif and Lissandrini, Matteo and Hose, Katja},
  journal={Proceedings of the VLDB Endowment},
  volume={16},
  number={5},
  pages={1023--1032},
  year={2023},
  publisher={ASSOC COMPUTING MACHINERY}
}

@misc{rdflib,
  author = {RDFLib},
  title = {RDFLib},
  year = {2025},
  howpublished = {\url{https://rdflib.readthedocs.io/en/stable/}},
}

@misc{rstr,
  author = {},
  title = {rstr},
  howpublished = {\url{https://github.com/leapfrogonline/rstr}},
  year = {2025}
}

@misc{shacl-repairs,
  author = {Robert David},
  title = {SHACL-repairs},
  year = {2025},
  howpublished = {\url{https://github.com/robert-david/shacl-repairs}},
  note = {Accessed: August 2025}

}

@misc{neural,
  author = {Thomas Pellisier Tanon},
  title = {Neural Knowledge Base Repairs},
  year = {2025},
  howpublished = {\url{https://github.com/Tpt/bass-materials}},
  note = {Accessed: August 2025}

}

@inproceedings{ahmetaj2022repairing,
  title={Repairing SHACL constraint violations using answer set programming},
  author={Ahmetaj, Shqiponja and David, Robert and Polleres, Axel and {\v{S}}imkus, Mantas},
  booktitle={International Semantic Web Conference},
  pages={375--391},
  year={2022},
  organization={Springer}
}

@article{fan2019deducing,
  title={Deducing certain fixes to graphs},
  author={Fan, Wenfei and Lu, Ping and Tian, Chao and Zhou, Jingren},
  journal={Proceedings of the VLDB Endowment},
  volume={12},
  number={7},
  pages={752--765},
  year={2019},
  publisher={VLDB Endowment}
}

@inproceedings{pellissier2021neural,
  title={Neural knowledge base repairs},
  author={Pellissier Tanon, Thomas and Suchanek, Fabian},
  booktitle={The Semantic Web: 18th International Conference, ESWC 2021, Virtual Event, June 6--10, 2021, Proceedings 18},
  pages={287--303},
  year={2021},
  organization={Springer}
}

@inproceedings{pellissier2019learning,
  title={Learning how to correct a knowledge base from the edit history},
  author={Pellissier Tanon, Thomas and Bourgaux, Camille and Suchanek, Fabian},
  booktitle={The World Wide Web Conference},
  pages={1465--1475},
  year={2019}
}

@inproceedings{arnaout2022utilizing,
  title={Utilizing language model probes for knowledge graph repair},
  author={Arnaout, Hiba and Tran, Trung-Kien and Stepanova, Daria and Gad-Elrab, Mohamed Hassan and Razniewski, Simon and Weikum, Gerhard},
  booktitle={Wiki Workshop 2022},
  year={2022}
}

@inproceedings{hogan2015skolemising,
  title={Skolemising blank nodes while preserving isomorphism},
  author={Hogan, Aidan},
  booktitle={Proceedings of the 24th International Conference on World Wide Web},
  pages={430--440},
  year={2015}
}

@online{qudt,
  author       = {{FAIRsharing.org}},
  title        = {QUDT; Quantities, Units, Dimensions and Types},
  year         = 2022,
  url          = {https://doi.org/10.25504/FAIRsharing.d3pqw7},
  note         = {Last Edited: Friday, May 6th 2022, 2:03, Last Accessed: Friday, December 13th 2024, 11:47, Last Reviewed: Monday, May 2nd 2022, 3:22}
}

@online{tracedSPARQL,
  author       = {{SDM-TIB}},
  title        = {TracedSPARQL},
  year         = 2023,
  url          = {https://github.com/SDM-TIB/TracedSPARQL},
  note         = {Accessed: 2024-12-13}
}

@inproceedings{balaji2016brick,
  title={Brick: Towards a unified metadata schema for buildings},
  author={Balaji, Bharathan and Bhattacharya, Arka and Fierro, Gabriel and Gao, Jingkun and Gluck, Joshua and Hong, Dezhi and Johansen, Aslak and Koh, Jason and Ploennigs, Joern and Agarwal, Yuvraj and others},
  booktitle={Proceedings of the 3rd ACM International Conference on Systems for Energy-Efficient Built Environments},
  pages={41--50},
  year={2016}
}

@inproceedings{hammarRealEstateCoreOntology,
  title = {The {{RealEstateCore Ontology}}},
  author = {Hammar, Karl and Wallin, Erik Oskar and Karlberg, Per and Halleberg, David},
  pages = {16},
  langid = {english}
}

@article{guo2005lubm,
  title={LUBM: A benchmark for OWL knowledge base systems},
  author={Guo, Yuanbo and Pan, Zhengxiang and Heflin, Jeff},
  journal={Journal of Web Semantics},
  volume={3},
  number={2-3},
  pages={158--182},
  year={2005},
  publisher={Elsevier}
}

@article{ashrae201836,
  title={36: High performance sequences of operation for HVAC systems},
  author={ASHRAE, Guideline},
  journal={American Society of Heating, Refrigerating and Air-Conditioning Engineers, Atlanta},
  year={2018}
}

@article{yuDelineationMetabolicGene2016,
  title = {Delineation of Metabolic Gene Clusters in Plant Genomes by Chromatin Signatures},
  author = {Yu, Nan and N{\"u}tzmann, Hans-Wilhelm and MacDonald, James T. and Moore, Ben and Field, Ben and Berriri, Souha and Trick, Martin and Rosser, Susan J. and Kumar, S. Vinod and Freemont, Paul S. and Osbourn, Anne},
  year = {2016},
  month = mar,
  journal = {Nucleic Acids Research},
  volume = {44},
  number = {5},
  pages = {2255--2265},
  issn = {0305-1048, 1362-4962},
  doi = {10.1093/nar/gkw100},
  urldate = {2025-04-22},
  copyright = {http://creativecommons.org/licenses/by/4.0/},
  langid = {english}
}

@incollection{rodriguez-muroOntologyBasedDataAccess2013,
  title = {Ontology-{{Based Data Access}}: {{Ontop}} of {{Databases}}},
  shorttitle = {Ontology-{{Based Data Access}}},
  booktitle = {Advanced {{Information Systems Engineering}}},
  author = {{Rodr{\'i}guez-Muro}, Mariano and Kontchakov, Roman and Zakharyaschev, Michael},
  editor = {Hutchison, David and Kanade, Takeo and Kittler, Josef and Kleinberg, Jon M. and Mattern, Friedemann and Mitchell, John C. and Naor, Moni and Nierstrasz, Oscar and Pandu Rangan, C. and Steffen, Bernhard and Sudan, Madhu and Terzopoulos, Demetri and Tygar, Doug and Vardi, Moshe Y. and Weikum, Gerhard and Salinesi, Camille and Norrie, Moira C. and Pastor, {\'O}scar},
  year = {2013},
  volume = {7908},
  pages = {558--573},
  publisher = {Springer Berlin Heidelberg},
  address = {Berlin, Heidelberg},
  doi = {10.1007/978-3-642-41335-3_35},
  urldate = {2025-04-22},
  isbn = {978-3-642-38708-1 978-3-642-38709-8}
}

@inproceedings{suchanekYAGO45Large2024,
  title = {{{YAGO}} 4.5: {{A Large}} and {{Clean Knowledge Base}} with a {{Rich Taxonomy}}},
  shorttitle = {{{YAGO}} 4.5},
  booktitle = {Proceedings of the 47th {{International ACM SIGIR Conference}} on {{Research}} and {{Development}} in {{Information Retrieval}}},
  author = {Suchanek, Fabian M. and Alam, Mehwish and Bonald, Thomas and Chen, Lihu and Paris, Pierre-Henri and Soria, Jules},
  year = {2024},
  month = jul,
  pages = {131--140},
  publisher = {ACM},
  address = {Washington DC USA},
  doi = {10.1145/3626772.3657876},
  urldate = {2024-11-14},
  isbn = {9798400704314},
  langid = {english}
}

@article{vrandecicWikidataFreeCollaborative2014,
  title = {Wikidata: A Free Collaborative Knowledgebase},
  shorttitle = {Wikidata},
  author = {Vrande{\v c}i{\'c}, Denny and Kr{\"o}tzsch, Markus},
  year = {2014},
  month = sep,
  journal = {Communications of the ACM},
  volume = {57},
  number = {10},
  pages = {78--85},
  issn = {0001-0782, 1557-7317},
  doi = {10.1145/2629489},
  urldate = {2024-11-09},
  langid = {english}
}

@misc{w3c_data_shapes_test_suite,
  author       = {Gayo, Jose Emilio Labra and Knublauch, Holger and Kontokostas, Dimitris},
  title        = "{Data Shapes Test Suite}",
  howpublished = "\url{https://w3c.github.io/data-shapes/data-shapes-test-suite/}",
  note         = "Accessed: 2025-04-23",
  year         = "2025"
}

@inproceedings{bollacker2008freebase,
  title={Freebase: a collaboratively created graph database for structuring human knowledge},
  author={Bollacker, Kurt and Evans, Colin and Paritosh, Praveen and Sturge, Tim and Taylor, Jamie},
  booktitle={Proceedings of the 2008 ACM SIGMOD international conference on Management of data},
  pages={1247--1250},
  year={2008}
}

@inproceedings{farzana2023knowledge,
  title={Knowledge graph-enhanced neural query rewriting},
  author={Farzana, Shahla and Zhou, Qunzhi and Ristoski, Petar},
  booktitle={Companion Proceedings of the ACM Web Conference 2023},
  pages={911--919},
  year={2023}
}

@inproceedings{xu2024retrieval,
  title={Retrieval-augmented generation with knowledge graphs for customer service question answering},
  author={Xu, Zhentao and Cruz, Mark Jerome and Guevara, Matthew and Wang, Tie and Deshpande, Manasi and Wang, Xiaofeng and Li, Zheng},
  booktitle={Proceedings of the 47th International ACM SIGIR Conference on Research and Development in Information Retrieval},
  pages={2905--2909},
  year={2024}
}

@article{pinheiro2000linear,
  title={Linear mixed-effects models: basic concepts and examples},
  author={Pinheiro, Jos{\'e} C and Bates, Douglas M},
  journal={Mixed-effects models in S and S-Plus},
  pages={3--56},
  year={2000},
  publisher={Springer}
}

@misc{RDF,
  title = {Resource Description Framework (RDF)},
  howpublished = {\url{https://www.w3.org/RDF/}},
  note = {Accessed: 2025-05-05},
  year = {2025}
}

@misc{SPARQL,
  title = {Simple Protocol and RDF Query Language, (SPARQL)},
  howpublished = {\url{https://www.w3.org/TR/sparql11-query/}},
  note = {Accessed: 2025-05-05},
  year = {2025}
}

@misc{shacl_spec,
  title = {Shapes Constraint Language (SHACL)},
  howpublished = {\url{https://www.w3.org/TR/shacl/}},
  note = {Accessed: 2025-05-05}, 
  year = {2025}
}

@misc{shex_spec,
  title = {ShEx - Shape Expressions},
  howpublished = {\url{https://shex.io/}},
  note = {Accessed: 2025-05-05}, 
  year = {2025}
}

@misc{turtle,
  title = {Turtle Syntax},
  howpublished = {\url{https://www.w3.org/TR/turtle/}},
  note = {Accessed: 2025-05-05}, 
  year = {2025}
}

@incollection{zdaniuk2024ordinary,
  title={Ordinary least-squares (OLS) model},
  author={Zdaniuk, Bozena},
  booktitle={Encyclopedia of quality of life and well-being research},
  pages={4867--4869},
  year={2024},
  publisher={Springer}
}

@article{angles2023pg,
  title={Pg-schema: Schemas for property graphs},
  author={Angles, Renzo and Bonifati, Angela and Dumbrava, Stefania and Fletcher, George and Green, Alastair and Hidders, Jan and Li, Bei and Libkin, Leonid and Marsault, Victor and Martens, Wim and others},
  journal={Proceedings of the ACM on Management of Data},
  volume={1},
  number={2},
  pages={1--25},
  year={2023},
  publisher={ACM New York, NY, USA}
}

@article{pachera2025user,
  title={User-Centric Property Graph Repairs},
  author={Pachera, Amedeo and Bonifati, Angela and Mauri, Andrea},
  journal={Proceedings of the ACM on Management of Data},
  volume={3},
  number={1},
  pages={1--27},
  year={2025},
  publisher={ACM New York, NY, USA}
}

@inproceedings{terdalkar2025graph,
  title={Graph Repairs with Large Language Models: An Empirical Study},
  author={Terdalkar, Hrishikesh and Bonifati, Angela and Mauri, Andrea},
  booktitle={Proceedings of the 8th Joint Workshop on Graph Data Management Experiences \& Systems (GRADES) and Network Data Analytics (NDA)},
  pages={1--10},
  year={2025}
}

@article{paulheim2016knowledge,
  title={Knowledge graph refinement: A survey of approaches and evaluation methods},
  author={Paulheim, Heiko},
  journal={Semantic web},
  volume={8},
  number={3},
  pages={489--508},
  year={2016},
  publisher={SAGE Publications Sage UK: London, England}
}

@article{vrandevcic2014wikidata,
  title={Wikidata: a free collaborative knowledgebase},
  author={Vrande{\v{c}}i{\'c}, Denny and Kr{\"o}tzsch, Markus},
  journal={Communications of the ACM},
  volume={57},
  number={10},
  pages={78--85},
  year={2014},
  publisher={ACM New York, NY, USA}
}

@inproceedings{barret2024finding,
  title={Finding the PG schema of any (semi) structured dataset: a tale of graphs and abstraction},
  author={Barret, Nelly and Enache, Tudor and Manolescu, Ioana and Mohanty, Madhulika},
  booktitle={2024 IEEE 40th International Conference on Data Engineering Workshops (ICDEW)},
  pages={365--369},
  year={2024},
  organization={IEEE}
}

@article{dershowitz1979proving,
  title={Proving termination with multiset orderings},
  author={Dershowitz, Nachum and Manna, Zohar},
  journal={Communications of the ACM},
  volume={22},
  number={8},
  pages={465--476},
  year={1979},
  publisher={ACM New York, NY, USA}
}

@article{hobbs2004ontology,
  title={An ontology of time for the semantic web},
  author={Hobbs, Jerry R and Pan, Feng},
  journal={ACM Transactions on Asian Language Information Processing (TALIP)},
  volume={3},
  number={1},
  pages={66--85},
  year={2004},
  publisher={ACM New York, NY, USA}
}

@book{lifschitz2019answer,
  title={Answer set programming},
  author={Lifschitz, Vladimir},
  volume={3},
  year={2019},
  publisher={Springer Cham}
}

@article{gebser2019multi,
  title={Multi-shot ASP solving with clingo},
  author={Gebser, Martin and Kaminski, Roland and Kaufmann, Benjamin and Schaub, Torsten},
  journal={Theory and Practice of Logic Programming},
  volume={19},
  number={1},
  pages={27--82},
  year={2019},
  publisher={Cambridge University Press}
}

@article{schulhoff2024prompt,
  title={The prompt report: a systematic survey of prompt engineering techniques},
  author={Schulhoff, Sander and Ilie, Michael and Balepur, Nishant and Kahadze, Konstantine and Liu, Amanda and Si, Chenglei and Li, Yinheng and Gupta, Aayush and Han, HyoJung and Schulhoff, Sevien and others},
  journal={arXiv preprint arXiv:2406.06608},
  year={2024}
}

@article{kiela2021dynabench,
  title={Dynabench: Rethinking benchmarking in NLP},
  author={Kiela, Douwe and Bartolo, Max and Nie, Yixin and Kaushik, Divyansh and Geiger, Atticus and Wu, Zhengxuan and Vidgen, Bertie and Prasad, Grusha and Singh, Amanpreet and Ringshia, Pratik and others},
  journal={arXiv preprint arXiv:2104.14337},
  year={2021}
}
\clearpage
\appendix
\section{Proof of Strong Normalization}\label{app:strong}
We provide the proof for Theorem~\ref{thm:1}, restated below for convenience.
\SN*

To prove that $\mathcal{A}$ is strongly normalizing, we show that every sequence of rewrites must be finite. We define a measure function $\mu$ that maps each violation expression \texttt{Expr} to an element of a well-founded set and prove that for any rewrite step $\texttt{E}_1\to\texttt{E}_2$, the measure of the result is strictly smaller than the measure of the original, i.e., $\mu(\texttt{E}_1)>\mu(\texttt{E}_2)$.

\subsection{Well-Founded Set}
The well-founded set for our proof is the set of finite multisets of pairs of natural numbers, denoted $M(\mathbb{N} \times\mathbb{N})$. A multiset is a collection where elements may appear more than once. This set is ordered using the Dershowitz-Manna (DM) ordering~\cite{dershowitz1979proving}, commonly used to prove program termination. The DM ordering is denoted as $\gg$, which is built upon a base order $>$ on the pairs $(d,c)\in\mathbb{N}\times\mathbb{N}$. The base order on pairs is defined as:
\[(d_1,c_1)>(d_2,c_2)\iff (d_1>d_2)\ \text{or}\ (d_1=d_2\wedge c_1>c_2)\]
This ordering is well-founded, meaning an infinite descending chain of pairs $(d_1,c_1)>(d_2,c_2)>\cdots$ is impossible. This is because
\begin{itemize}[left=0pt, align=left]
    \item The $d$ component is a natural number and can only decrease a finite number of times before it ends at $0$.
    \item Once the $d$ component stops decreasing, the $c$ component must start strictly decreasing until it ends at $0$.
\end{itemize}

The DM ordering extends this to multisets. Formally, given two finite multisets $M_1$ and $M_2$, we say $M_1\gg M_2$ if $M_2$ can be obtained from $M_1$ by removing a single element $x \in M_1$ and replacing it with a \textit{finite} (possibly empty) multiset of new elements $Y$, such that every new element $y \in Y$ is strictly smaller than $x$, i.e., $x > y$. For example, the multiset $\{(3, 3), (3, 3), (4, 1)\}$ is greater than both the multisets $\{(3, 3), (4, 1)\}$ and $\{(2, 1), (4, 1)\}$.

The DM ordering is well-founded because the underlying order on the $(d,c)$ pairs is itself well-founded. An infinite descending chain of multisets $M_1 \gg M_2 \gg M_3 \gg \cdots$ would imply the existence of an infinite descending chain of the $(d,c)$ pairs, which is impossible.

\subsection{Measure Function}
We define the measure function $\mu$ that maps any violation expression to a multiset of $(d,c)$ pairs. This measure is constructed from a primary measure $d$ for quantifier depth and a secondary measure $c$ for structural complexity. 

\subsubsection{Primary Measure: Quantifier Depth $d(\varphi)$}
First, we define the function $d(\varphi)$ to be the maximum nesting depth of quantifiers in $\varphi$. 
\[
d(\varphi) =
\begin{cases}
0 & \varphi \text{ atomic}, \\
d(def(s,\mathcal{S})) & \varphi = s \in \Sigma, \\
\max\{d(\varphi_1), d(\varphi_2)\} & \varphi = \varphi_1 \wedge \varphi_2 \ \text{or} \ \varphi_1 \vee \varphi_2, \\
d(\psi) & \varphi = \neg \psi \ \text{and $\varphi$ non-atomic}, \\
1+d(\psi) & \varphi \in \{\forall \pi.\psi, \ \geq_n\pi.\psi, \ \leq_n\pi.\psi\}.
\end{cases}
\]
The condition that the schema $\mathcal{S}$ is non-recursive is crucial here, as it guarantees that $d(s)$ is well-defined.
\subsubsection{Secondary Measure: Structural Complexity $c(\varphi)$}
Next, we define a function $c(\varphi)$ that maps $\varphi$ to a natural number representing its complexity.
\[
c(\varphi) =
\begin{cases}
1 & \varphi \text{ atomic}, \\
1+c(def(s,\mathcal{S})) & \varphi = s \in \Sigma, \\
1+c(\varphi_1)+c(\varphi_2) & \varphi = \varphi_1 \wedge \varphi_2 \ \text{or} \ \varphi_1 \vee \varphi_2, \\
2\cdot c(\psi) & \varphi = \neg \psi \ \text{and $\varphi$ non-atomic}, \\
1+c(\psi) & \varphi \in \{\forall \pi.\psi, \ \geq_n\pi.\psi, \ \leq_n\pi.\psi\}.
\end{cases}
\]
\subsubsection{The Full Measure $\mu(\texttt{E})$}
Finally, the measure $\mu(\texttt{E})$ for any violation expression \texttt{E} is the multiset union of the measures of all the fundamental tasks it contains. For a single fundamental task \texttt{T}, the measure is a singleton multiset:
\[
\mu(\texttt{{T}}) =
\begin{cases}
\{(d(\varphi), c(\varphi))\} & \texttt{T}=\texttt{VIO}^{(T,Q)}(sel,\varphi), \\
\{(0,c(\varphi))\} & \texttt{T}=\texttt{SAT}_\pi(sel,\varphi),\\
\{(0,1)\}&\texttt{T}=\texttt{VIO}_\pi(u,v),\\
\end{cases}
\]
For a compound expression \texttt{E} composed of fundamental tasks $\texttt{T}_1,\texttt{T}_2,\cdots,\texttt{T}_k$ (combined with $+$ or $\cdot$), the full measure is the multiset union of their individual measures.
\[
\mu(\texttt{E})=\bigcup_{i=1}^k\mu(\texttt{T}_i)
\]
\subsection{Strict Measure Decrease}
We now show that for each rewrite rule $LHS\to RHS$, it holds that $\mu(LHS)\gg\mu(RHS)$.

\noindent\textbf{Rule 1: Shape Reference ($s$)}
\begin{align*}
\texttt{VIO}^{(T,Q)}(sel, s) &\longrightarrow \texttt{VIO}^{(T,Q)}(sel, def(s, \mathcal{S}))\\
\mu(LHS)&=\{(d(s),c(s))\}=\{(d(s), 1+c(def(s, \mathcal{S})))\}\\
\mu(RHS)&=\{(d(s), c(def(s, \mathcal{S})))\}\\
\texttt{SAT}^n_\pi(sel, s) &\longrightarrow \texttt{SAT}^n_\pi(sel, def(s, \mathcal{S}))\\
\mu(LHS)&=\{(0, 1+c(def(s, \mathcal{S})))\}\\
\mu(RHS)&=\{(0, c(def(s, \mathcal{S})))\}
\end{align*}
Both rewrites replace one element with another. The $d$ component is unchanged, while the $c$ component strictly decreases. Thus, $\mu(LHS) \gg \mu(RHS)$. 

\noindent\textbf{Rule 2: Conjunction ($\wedge$)}
\begin{align*}
\texttt{VIO}^{(T,Q)}(sel, \varphi \wedge \varphi') &\longrightarrow \texttt{VIO}^{(T,Q)}(sel, \varphi) + \texttt{VIO}^{(T,Q)}(sel, \varphi')\\
\mu(LHS)&=\{(\max\{d(\varphi), d(\varphi')\}, 1 + c(\varphi) + c(\varphi'))\}\\
\mu(RHS)&=\{(d(\varphi), c(\varphi)), (d(\varphi'), c(\varphi'))\}
\end{align*}
The rewrite replaces one element with two new elements. Both pairs in $\mu(RHS)$ are strictly smaller than the single pair in $\mu(LHS)$. Thus, $\mu(LHS) \gg \mu(RHS)$.

\noindent\textbf{Rule 3: Disjunction ($\vee$)}
\begin{multline*}\texttt{VIO}^{(T,Q)}(sel, \varphi \vee \varphi') \ \longrightarrow\ 
\\ \sum_{u: u \models sel}  
\big(\mathbf{if}\ u \models \varphi\ \mathbf{then}\ \texttt{VIO}^{(T|_u,Q)}(=u,\varphi)\big) \ \cdot\ \\
\big(\mathbf{if}\ u \models \varphi'\ \mathbf{then}\ \texttt{VIO}^{(T|_u,Q)}(=u,\varphi')\big)\end{multline*}
\begin{align*}
    \mu(LHS)&=\{(\max\{d(\varphi), d(\varphi')\}, 1 + c(\varphi) + c(\varphi'))\}\\
    \mu(RHS)&=\bigcup_{u:u\models sel}\{(d(\varphi), c(\varphi)),  (d(\varphi') , c(\varphi'))\}
\end{align*}
The RHS potentially omits pairs, depending on the $\mathbf{if}$ conditions. 

The logic that $ \mu(LHS) \gg \mu(RHS)$ is identical to Rule 2. One element is replaced by a finite multiset of new elements, all of which are strictly smaller. The finiteness of the graph $\mathcal{G}$ is thus an important condition here.

\noindent\textbf{Rule 4: Negation ($\neg$)}
\begin{enumerate}
    \item De Morgan's Law \begin{align*}
    \texttt{VIO}^{(T,Q)}(sel, \neg(\varphi \wedge \varphi')) &\longrightarrow \texttt{VIO}^{(T,Q)}(sel, \neg\varphi \vee \neg\varphi') \\
    \mu(LHS)=\{(\max& \{d(\varphi), d(\varphi')\}, 2(1+c(\varphi)+c(\varphi')))\}\\
    \mu(RHS)=\{(\max &\{d(\varphi), d(\varphi')\}, 1+2c(\varphi)+2c(\varphi'))\}\\
    \texttt{VIO}^{(T,Q)}(sel, \neg(\varphi \vee \varphi')) &\longrightarrow \texttt{VIO}^{(T,Q)}(sel, \neg\varphi \wedge \neg\varphi') \\
    \mu(LHS)=\{(\max& \{d(\varphi), d(\varphi')\}, 2(1+c(\varphi)+c(\varphi')))\}\\
    \mu(RHS)=\{(\max& \{d(\varphi), d(\varphi')\}, 1+2c(\varphi)+2c(\varphi'))\}\\
    \texttt{VIO}^{(T,Q)}(sel, \neg\neg\varphi) &\longrightarrow \texttt{VIO}^{(T,Q)}(sel, \varphi)\\
    \mu(LHS)&=\{(d(\varphi), 4c(\varphi))\}\\
    \mu(LHS)&=\{(d(\varphi), c(\varphi))\}\\
\end{align*}
    \item Quantifier Dualities \begin{align*}
    \texttt{VIO}^{(T,Q)}(sel, \neg(\forall\pi.\varphi)) &\longrightarrow \texttt{VIO}^{(T,Q)}(sel, \geq_1\pi.\neg\varphi) \\
    \mu(LHS)&=\{(1+d(\varphi), 2(1+c(\varphi)))\}\\
    \mu(RHS)&=\{(1+d(\varphi), 1+2c(\varphi))\}\\
    \texttt{VIO}^{(T,Q)}(sel, \leq_0\pi.\neg\varphi)) &\longrightarrow \texttt{VIO}^{(T,Q)}(sel, \forall\pi.\varphi) \\
    \mu(LHS)&=\{(1+d(\varphi), 1+2c(\varphi))\}\\
    \mu(RHS)&=\{(1+d(\varphi), 1+c(\varphi))\}\\
    \texttt{VIO}^{(T,Q)}(sel, \neg(\geq_n\pi.\varphi)) &\longrightarrow \texttt{VIO}^{(T,Q)}(sel, \leq_{n-1}\pi.\varphi) \\
    \mu(LHS)&=\{(1+d(\varphi), 2(1+c(\varphi)))\}\\
    \mu(RHS)&=\{(1+d(\varphi), 1+c(\varphi))\}\\
    \texttt{VIO}^{(T,Q)}(sel, \neg(\leq_n\pi.\varphi)) &\longrightarrow \texttt{VIO}^{(T,Q)}(sel, \geq_{n+1}\pi.\varphi)\\
    \mu(LHS)&=\{(1+d(\varphi), 2(1+c(\varphi)))\}\\
    \mu(RHS)&=\{(1+d(\varphi), 1+c(\varphi))\}\\
\end{align*}
    \item Shape Unfolding \begin{align*}
    \texttt{VIO}^{(T,Q)}(sel, \neg s) &\longrightarrow \texttt{VIO}^{(T,Q)}(sel, \neg def(s, \mathcal{S}))\\
    \mu(LHS)&=\{(d(def(s,\mathcal{S})), 2(1+c(def(s,\mathcal{S}))))\}\\
    \mu(RHS)&=\{(d(def(s,\mathcal{S})), 2c(def(s,\mathcal{S})))\}\\    
\end{align*}
\end{enumerate}
In each case, the $d$ component is unchanged, while the $c$ component strictly decreases. Thus, $\mu(LHS) \gg \mu(RHS)$.

\noindent\textbf{Rule 5: Universal Quantifier ($\forall \pi.\varphi$)}
\begin{multline*}
\texttt{VIO}^{(T,Q)}(sel, \forall\pi.\varphi) \longrightarrow  \sum_{u:u\models sel}\sum_{v\in\llbracket\pi\rrbracket(u)}\texttt{VIO}^{(T_{u\to v}',\exists)}(=v, \varphi) +\\\texttt{SAT}_\pi^1(sel,\neg\varphi)
\end{multline*}
\begin{align*}
    \mu(LHS)&=\{(1+d(\varphi), 1+c(\varphi))\}\\
    \mu(RHS)&=\bigcup_{\substack{u:u\models sel\\v\in\llbracket\pi\rrbracket(u)}}\{(d(\varphi), c(\varphi))\}\cup \{0, 2c(\varphi)\}\\
\end{align*}
The $d$ component of every task on the $RHS$ is at most $d(\varphi)$, which is strictly smaller than the $LHS$ $d$ component of $1+d(\varphi)$. Therefore, every element in $\mu(RHS)$ is strictly smaller than the element in $\mu(LHS)$. Thus, $\mu(LHS) \gg \mu(RHS)$.

\noindent\textbf{Rule 6: Qualified Cardinality ($\geq_n \pi.\varphi$)}
\begin{multline*}
\texttt{VIO}^{(T,Q)}(sel, \geq_n \pi.\varphi) \longrightarrow \sum_{u: u \models sel} \sum_{\substack{\Delta \subseteq \llbracket\pi,\varphi\rrbracket(u) \\ |\Delta|=k}} \\
\prod_{v \in \Delta} \left( \texttt{VIO}^{(T_{u\to v}',\forall)}(=v, \varphi) + \texttt{VIO}_\pi(u,v) \right)
\end{multline*}
\begin{align*}
    \mu(LHS)&=\{(1+d(\varphi), 1+c(\varphi))\}\\
    \mu(RHS)&=\bigcup_{\substack{u:u\models sel\\\Delta \subseteq \llbracket\pi,\varphi\rrbracket(u)\\v \in \Delta}}\{(d(\varphi), c(\varphi))\}\cup \{0,1\}\\
\end{align*}
The $d$ component of every task on the $RHS$ is at most $d(\varphi)$, which is strictly smaller than the $LHS$ $d$ component of $1+d(\varphi)$. Therefore, every element in $\mu(RHS)$ is strictly smaller than the element in $\mu(LHS)$. Thus, $\mu(LHS) \gg \mu(RHS)$.

\noindent\textbf{Rule 7: Qualified Cardinality ($\leq_n\pi.\varphi$)}
$\texttt{SAT}^{n+1}_\pi$.
$$\texttt{VIO}^{(T,Q)}(sel, \leq_n\pi.\varphi) \longrightarrow \texttt{SAT}^{n+1}_\pi(sel, \varphi)$$
\begin{align*}
    \mu(LHS)&=\{(1+d(\varphi), c(\varphi))\}\\
    \mu(RHS)&=\{(0, c(\varphi))\}
\end{align*}
The $d$ component strictly decreases from $1+d(\varphi)$ to $0$. Thus, $\mu(LHS) \gg \mu(RHS)$.

Since every rewrite rule strictly decreases the measure under the well-founded DM ordering, no infinite rewrite sequence is possible. \qed

\section{Proof of Time Complexity}
\label{app:complexity}

We show that the subshape collection and de‑duplication algorithm runs in $O(n_s)$ time, 
where $n_s$ is the number of subshapes in the finite, non‑recursive schema $\mathcal{S}$.

\noindent\textbf{DFS traversal.} Since shape definitions can refer to other shapes in $\mathcal{S}$, the constraints in $\mathcal{S}$ form a dependency graph. As we assume $\mathcal{S}$ to be non-recursive, the dependency graph is a directed acyclic graph, allowing us to process the constraints in topological order. 

For each constraint $(sel,\varphi)\in\mathcal{S}$ processed in topological order, we build an expansion graph rooted at $\texttt{VIO}^{(T_0,Q_0)}(sel,\varphi)$ and perform a depth‑first search (DFS) on the expansion graph to collect all reachable subshapes. To prevent redundant work across the entire schema, a \textit{global visited set} (e.g., a hash set) is maintained. This ensures that each reachable subshape is explored exactly once, even if it is referred to by multiple constraints.

\noindent\textbf{Cost of traversal.}
Exploring a subshape consists of:
(i) adding it to the global visited set,
(ii) applying the corresponding rewrite rule(s), and 
(iii) pushing its immediate subshapes onto the DFS stack.
The number of immediate subshapes of any shape is bounded by the grammar (at most two for binary connectives, one for unary operators), so the total number of edge traversals is $O(n_s)$.
Thus, the DFS visits $n_s$ nodes and $O(n_s)$ edges, giving $O(n_s)$ time.

\noindent\textbf{De‑duplication.}
After traversal, the set of normal‑form violation expressions is de-duplicated using a hash set. 
Insertion into the hash set is $O(1)$ amortized per expression, and the number of expressions is $O(n_s)$ in the worst case, giving $O(n_s)$ time for this step.

\noindent\textbf{Total time.}
Topological sorting of the $C$ constraints takes $O(C + n_s)$ time, with $C \le n_s$. 
Adding the traversal and de‑duplication costs, the total running time is
\[
O(C + n_s) + O(n_s) + O(n_s) = O(n_s).
\]
Hence, the algorithm runs in time linear in the number of subshapes in $\mathcal{S}$.
\qed

\section{More Graph Operations}\label{app:graph_ops}
This appendix provides the complete mapping from atomic shapes in normalized violation expressions to graph operations for inducing violations. The atomic shapes can be categorized into node test based (App.~\ref{app:simple_ops}) and path expression based (App.~\ref{app:path_ops}).

\begin{table*}
    \centering
    \mycaption{Graph operations for $\texttt{VIO}^{(T,Q)}(sel,\varphi)$, given a node $u \models sel$ and a zero-length trace $\langle u\rangle$ or the final hop $(w,p,u)$ of a non-zero-length trace in $T$.}{“Ineligible” indicates that, by the coverage rules in Sec.~\ref{sec:coverage}, zero-length traces with intrinsic shapes are discarded.}
    \begin{tabular}{ccc}\toprule
         $\varphi$ & Graph Operations & Graph Operations\\
         & Given zero-length trace $\langle u\rangle$& Given final hop $(w,p,u)$\\\midrule
         $\text{test}(\omega_{\texttt{sh:hasValue},y})$& ineligible & $rep(w,p,u,u')$, where $u'\neq u $ \\
         $\neg\text{test}(\omega_{\texttt{sh:hasValue},y})$&  ineligible & $rep(w,p,u,y)$\\
         $\text{test}(\omega_{\texttt{sh:class},y})$&   \multicolumn{2}{c}{$-e(u,\texttt{a}, c)$, for each subclass $c$ of $y$}\\
         $\neg\text{test}(\omega_{\texttt{sh:class},y})$& \multicolumn{2}{c}{$+e(u,\texttt{a}, c)$, for some subclass $c$ of $y$}\\
         $\text{test}(\omega_{\texttt{sh:datatype},y})$& ineligible& $rep(w,p,u,z')$, where $z'$ is not of datatype $y$ \\
         $\neg\text{test}(\omega_{\texttt{sh:datatype},y})$&  ineligible& $rep(w,p,u,z')$, where $z'$ is of datatype $y$\\
          $\text{test}(\omega_{\texttt{sh:nodeKind},y})$& ineligible& $rep(w,p,u,z')$, where $z'$ is not of node kind $y$\\
         $\neg\text{test}(\omega_{\texttt{sh:nodeKind},y})$& ineligible&  $rep(w,p,u,z')$, where $z'$ is of node kind $y$\\
         $\text{test}(\omega_{\texttt{sh:minExclusive},y})$& ineligible &$rep(w,p,u,v')$, where $v'\leq y$  \\
         $\neg\text{test}(\omega_{\texttt{sh:minExclusive},y})$& ineligible &$rep(w,p,u,v')$, where $v'> y$ \\
         $\text{test}(\omega_{\texttt{sh:maxExclusive},y})$& ineligible &$rep(w,p,u,v')$, where $v'\geq y$  \\
         $\neg\text{test}(\omega_{\texttt{sh:maxExclusive},y})$&  ineligible &$rep(w,p,u,v')$, where $v'< y$  \\
         $\text{test}(\omega_{\texttt{sh:minInclusive},y})$& ineligible &$rep(w,p,u,v')$, where $v'<y$  \\
         $\neg\text{test}(\omega_{\texttt{sh:minInclusive},y})$& ineligible &$rep(w,p,u,v')$, where $v'\geq y$ \\
         $\text{test}(\omega_{\texttt{sh:maxInclusive},y})$& ineligible &$rep(w,p,u,v')$, where $v'>y$  \\
         $\neg\text{test}(\omega_{\texttt{sh:maxInclusive},y})$&  ineligible &$rep(w,p,u,v')$, where $v'\leq y$  \\
         $\text{test}(\omega_{\texttt{sh:minLength},y})$& ineligible &   $rep(w,p,u,u')$, where $len(u') < y$  \\ 
         $\neg\text{test}(\omega_{\texttt{sh:minLength},y})$ & ineligible & $rep(w,p,u,u')$, where $len(u') \geq y$ \\
         $\text{test}(\omega_{\texttt{sh:maxLength},y})$& ineligible & $rep(w,p,u,u')$, where $len(u')> y$   \\ 
         $\neg\text{test}(\omega_{\texttt{sh:maxLength},y})$& ineligible & $rep(w,p,u,u')$, where $len(u')\leq y$   \\ 
         $\text{test}(\omega_{\texttt{sh:pattern},y})$& ineligible& $rep(w,p,u,u')$, where $u'$ does not match $y$\\
         $\neg\text{test}(\omega_{\texttt{sh:pattern},y})$&  ineligible& $rep(w,p,u,x')$, where $x'$ matches $y$\\
         $\text{test}(\omega_{\texttt{sh:languageIn},y})$&  \multicolumn{2}{c}{$-e(u,@,y)$}\\
         $\neg\text{test}(\omega_{\texttt{sh:languageIn},y})$& \multicolumn{2}{c}{$+e(u,@,y)$} \\
         \midrule
         \multicolumn{3}{l}{$u'$: value generated by editing $u$ (adding/removing/transposing characters, etc.) to meet the violation condition.}\\
         \multicolumn{3}{l}{$v'$: numeric value randomly sampled to meet the violation condition.}\\
         \multicolumn{3}{l}{$z'$: value obtained by casting $u$ to a different datatype or node kind.}\\
         \multicolumn{3}{l}{$x'$: string generated by a regex-based generator~\cite{rstr} to match a given pattern.}\\
         \multicolumn{3}{l}{``@'': special predicate for language tags.}\\
         \bottomrule
    \end{tabular}
    
    \label{tab:full_atomic_ops}
\end{table*}
\subsection{Node Test Based Atomic Shapes}
\label{app:simple_ops}

We expand on Table~\ref{tab:atomic_ops} in Sec.~\ref{sec:shape_vio} by providing the complete mapping from \emph{normalized violation expressions} to graph operations.
As described in Sec.~\ref{sec:shape_vio}, a normalized violation expression $\texttt{VIO}^{(T,Q)}(sel,\varphi)$ targets \emph{final hops} of traces in $T$ for nodes $u$ satisfying the selector $sel$, and applies modifications so that $u$ no longer satisfies the atomic shape $\varphi$. 

Table~\ref{tab:full_atomic_ops} lists all atomic shapes in SHACL supported by our framework, together with the graph operations to induce violations. Recall from Sec.~\ref{sec:coverage} that if the targeted trace to modify is zero-length and $\varphi$ is an \emph{intrinsic} shape, the shape violation task is ineligible.

\subsection{Path Expression-Based Atomic Shapes}
\label{app:path_ops}

Path expression-based atomic shapes require more elaborate violation strategies, as multiple sets of graph operations may induce a violation.  

For each case, we recall the SHACL semantics, assume a normalized violation expression $\texttt{VIO}^{(T,Q)}(sel,\varphi)$ with a node $u \models sel$, and list possible strategies. Recall the notation $\texttt{VIO}_\pi(u,v)$ denotes the path violation task applied to the path expression $\pi$ from source node $u$ to target node $v$ and the notation \(\llbracket \pi \rrbracket(u) = \{\, v \mid (u,v) \in \llbracket \pi \rrbracket \,\}\) denotes the $\pi$-targets of $u$, i.e., all the nodes connected from $u$ via $\pi$.

\subsubsection{$\varphi=\text{eq}(\pi,\pi')$.}
A node $u$ satisfies $\varphi=\text{eq}(\pi,\pi')$ iff $\llbracket \pi \rrbracket(u) = \llbracket \pi' \rrbracket(u)$.  
Possible violation strategies:
\begin{itemize}[leftmargin=*]
    \item Construct a trace $\tau \models \pi$ such that $src(\tau) = u$ and $tgt(\tau) \not\in \llbracket \pi' \rrbracket(u)$.
    \item Construct a trace $\tau \models \pi'$ such that $src(\tau) = u$ and $tgt(\tau) \not\in \llbracket \pi \rrbracket(u)$.
    \item Select $v \in \llbracket \pi \rrbracket(u)$ and apply $\texttt{VIO}_\pi(u,v)$.
    \item Select $v \in \llbracket \pi' \rrbracket(u)$ and apply $\texttt{VIO}_{\pi'}(u,v)$.
\end{itemize}

\subsubsection{$\varphi=\neg\text{eq}(\pi,\pi')$.}
A node $u$ satisfies $\neg\text{eq}(\pi,\pi')$ iff $\llbracket \pi \rrbracket(u) \neq \llbracket \pi' \rrbracket(u)$.  
Possible violation strategies:
\begin{itemize}[leftmargin=*]
    \item For each $v \in \llbracket \pi \rrbracket(u) \setminus \llbracket \pi' \rrbracket(u)$, construct $\tau \models \pi'$ with $src(\tau) = u$ and $tgt(\tau) = v$; and for each $v \in \llbracket \pi' \rrbracket(u) \setminus \llbracket \pi \rrbracket(u)$, construct $\tau \models \pi$ with $src(\tau) = u$ and $tgt(\tau) = v$.
    \item For each $v \in \llbracket \pi \rrbracket(u) \setminus \llbracket \pi' \rrbracket(u)$, apply $\texttt{VIO}_\pi(u,v)$; and for each $v \in \llbracket \pi' \rrbracket(u) \setminus \llbracket \pi \rrbracket(u)$, apply $\texttt{VIO}_{\pi'}(u,v)$.
\end{itemize}

\subsubsection{$\varphi=\text{disj}(\pi,\pi')$.}
A node $u$ satisfies $\text{disj}(\pi,\pi')$ iff $\llbracket \pi \rrbracket(u) \cap \llbracket \pi' \rrbracket(u) = \emptyset$.  
Possible violation strategies for when $\llbracket \pi \rrbracket(u)=\llbracket \pi' \rrbracket(u) = \emptyset$:
\begin{itemize}[leftmargin=*]
    \item Reuse any node $v$ in the graph that is not $u$ and construct traces $\tau_1,\tau_2$ such that $\tau_1\models \pi$ and $\tau_2\models \pi'$ with $src(\tau_1)=src(\tau_2)=u$ and $tgt(\tau_1)=tgt(\tau_2)=v$.
\end{itemize}
Possible violation strategies for when $\llbracket \pi \rrbracket(u)$ and $\llbracket \pi' \rrbracket(u)$ are not empty:
\begin{itemize}[leftmargin=*]
    \item Select $v \in \llbracket \pi \rrbracket(u)$ and construct $\tau \models \pi'$ with $src(\tau) = u$ and $tgt(\tau) = v$.
    \item Select $v \in \llbracket \pi' \rrbracket(u)$ and construct $\tau \models \pi$ with $src(\tau) = u$ and $tgt(\tau) = v$.
\end{itemize}

\subsubsection{$\varphi=\neg\text{disj}(\pi,\pi')$.}
A node $u$ satisfies $\neg\text{disj}(\pi,\pi')$ iff $\llbracket \pi \rrbracket(u) \cap \llbracket \pi' \rrbracket(u) \neq \emptyset$.  
Possible violation strategies:
\begin{itemize}[leftmargin=*]
    \item For each $v \in \llbracket \pi \rrbracket(u) \cap \llbracket \pi' \rrbracket(u)$, apply $\texttt{VIO}_\pi(u,v)$ or apply $\texttt{VIO}_{\pi'}(u,v)$.
\end{itemize}

\subsubsection{$\varphi=\text{lessThan}(\pi,\pi')$.}
A node $u$ satisfies $\text{lessThan}(\pi,\pi')$ iff for all $v \in \llbracket \pi \rrbracket(u)$ and $w \in \llbracket \pi' \rrbracket(u)$, $v < w$.  
Possible violation strategies:
\begin{itemize}[leftmargin=*]
    \item Select $w \in \llbracket \pi' \rrbracket(u)$ and construct $\tau \models \pi$ with $src(\tau) = u$ and $tgt(\tau) \geq w$.
    \item Select $v \in \llbracket \pi \rrbracket(u)$ and construct $\tau \models \pi'$ with $src(\tau) = u$ and $tgt(\tau) \leq v$.
\end{itemize}

\subsubsection{$\varphi=\neg\text{lessThan}(\pi,\pi')$.}
A node $u$ satisfies $\neg\text{lessThan}(\pi,\pi')$ iff for all $v \in \llbracket \pi \rrbracket(u)$ and $w \in \llbracket \pi' \rrbracket(u)$, $v \geq w$. Possible violation strategies for when $\llbracket \pi \rrbracket(u)=\llbracket \pi' \rrbracket(u) = \emptyset$:
\begin{itemize}[leftmargin=*]
    \item Sample two nodes $v, w$ such that $v<w$ and construct traces $\tau_1,\tau_2$ such that $\tau_1\models \pi$ and $\tau_2\models \pi'$ with $src(\tau_1)=src(\tau_2)=u$ and $tgt(\tau_1)=v, tgt(\tau_2)=w$.
\end{itemize}
Possible violation strategies for when there exist $v \in \llbracket \pi \rrbracket(u)$, $w \in \llbracket \pi' \rrbracket(u)$, and $v \geq w$.
\begin{itemize}[leftmargin=*]
    \item To all such $v$ and $w$, apply $\texttt{VIO}_\pi(u,v)$ or apply $\texttt{VIO}_{\pi'}(u,w)$.
\end{itemize}

\subsubsection{$\varphi=\text{lessThanEq}(\pi,\pi')$ and $\neg\text{lessThanEq}(\pi,\pi')$.}
These are analogous to $\text{lessThan}$ and $\neg\text{lessThan}$, except that equality is permitted in the constraint.  
Violation strategies are obtained by replacing the strict inequality with a non-strict one in the above cases.

\subsubsection{$\varphi=\text{uniqueLang}(\pi)$.}
A node $u$ satisfies $\text{uniqueLang}(\pi)$ iff all distinct $v, w \in \llbracket \pi \rrbracket(u)$ have different language tags.  
To violate this constraint (i.e., to introduce at least one pair of literals with the same language tag), possible strategies for when $\llbracket \pi \rrbracket(u)= \emptyset$:
\begin{itemize}[leftmargin=*]
    \item Reuse nodes $v, w$ with the same language tag in the graph and construct traces $\tau_1,\tau_2 \models \pi$ with $src(\tau_1)=src(\tau_2) = u$ and $tgt(\tau_1) = v,tgt(\tau_2)=w$.
\end{itemize}
Possible strategies for when $|\llbracket \pi \rrbracket(u)|=1$:
\begin{itemize}[leftmargin=*]
    \item Select $v \in \llbracket \pi \rrbracket(u)$ and construct a trace $\tau \models \pi$ with $src(\tau) = u$ and $tgt(\tau) = v'$, where $v'$ has the same language tag as $v$.
\end{itemize}
Possible strategies for when $|\llbracket \pi 
\rrbracket(u)|>1$:
\begin{itemize}[leftmargin=*]
    \item Select distinct $v,w \in \llbracket \pi \rrbracket(u)$ and modify the language tag of $v$ so that it matches the tag of $w$.
\end{itemize}

\subsubsection{$\varphi=\neg\text{uniqueLang}(\pi)$.}
A node $u$ satisfies $\neg\text{uniqueLang}(\pi)$ iff there exist distinct $v, w \in \llbracket \pi \rrbracket(u)$ with the same language tag.  
To violate this constraint (i.e., to make all language tags distinct), possible strategies include:
\begin{itemize}[leftmargin=*]
    \item For each pair $(v,w)$ with the same language tag, apply  $\texttt{VIO}_\pi(u,v)$ or $\texttt{VIO}_\pi(u,w)$.
    \item For each pair $(v,w)$ with the same language tag, modify the language tag of $v$ so that it differs from $w$.
\end{itemize}
App.~\ref{app:pathbuild} details how to create a $\pi$-matching trace connecting two nodes.

\section{Proof of Coverage Guarantee}\label{app:coverage_proof}
In this appendix, we provide the proof of our coverage guarantee. Recall from Sec.~\ref{sec:ars} that:
\begin{itemize}[left=0pt, align=left]
    \item A shape $\varphi$ is \textit{reachable} if during the rewriting process of the ARS $\mathcal{A} = (\texttt{Expr}, \to)$, a shape violation task $\texttt{VIO}^{(T,Q)}(sel, \varphi)$ is generated.
    \item The ARS is \textit{strongly normalizing}, meaning every rewriting sequence terminates in a finite number of steps in a normal form.
    \item A shape violation task $\texttt{VIO}^{(T,Q)}(sel,\varphi)$ is \textit{eligible} if it satisfies both the \textit{logical precondition} and \textit{operational feasibility}.
    \item A path violation task $\texttt{VIO}_\pi(u,v)$ is \textit{eligible} if it satisfies \textit{operational feasibility}.
    \item A qualified satisfaction task $\texttt{SAT}_\pi^n(sel,\varphi)$ is \textit{eligible} if it satisfies the \textit{logical precondition}.
\end{itemize}
We now restate the theorem.
\CoverageGuarantee*

\begin{pf}
We prove the statement for 
(i) normalized shape violation tasks $\texttt{VIO}^{(T,Q)}(sel,\varphi)$, 
(ii) path violation tasks $\texttt{VIO}_\pi(u,v)$, and
(iii) qualified satisfaction task $\texttt{SAT}^n_\pi(sel, \varphi)$.

\paragraph{Shape violation tasks.}
Let $\varphi$ be a reachable shape in $\mathcal{S}$ with an eligible normalized violation task $\texttt{VIO}^{(T,Q)}(sel,\varphi)$.
We proceed by showing reachability, termination, and non‑emptiness of the operation set:
\begin{enumerate}[left=0pt, align=left]
\item \textbf{Reachability.}
By the construction of the ARS, every reachable shape appears as an argument of some violation task during the rewriting of the initial expression $\texttt{VIO}^{(T_0,Q_0)}(sel_0,\varphi_0)$ for the original constraint $(sel_0, \varphi_0)$.
\item \textbf{Termination.}  
By the definition of normalized task, there is no applicable rewrite rule for $\texttt{VIO}^{(T,Q)}(sel,\varphi)$. The ARS terminates and $\varphi$ is atomic.
\item \textbf{Non‑emptiness of operations.}  
By the definition of eligibility (Def.~\ref{def:eligibility}), the mapping from $\texttt{VIO}^{(T,Q)}(sel,\varphi)$ to graph operations is non-empty.

\end{enumerate}
\paragraph{Path violations tasks.}
Let $\texttt{VIO}_\pi(u,v)$ be an eligible path violation task. By the definition of eligibility (Def.~\ref{def:eligibility}), there is no zero‑length trace between $u$ and $v$ matching $\pi$. 
The path‑violation algorithm in Sec.~\ref{sec:path_vio} recursively decomposes $\pi$ until reaching the base case $p$, where the match exists because $(u,p,v)$ is present.
Removing this hop via $-e(u,p,v)$ (for $p \in \mathcal{E}$) or $-a(u,p)$ (for $p \in \mathcal{K}$) invalidates the trace.
Since eligibility rules out degenerate cases, at least one such essential hop exists for every trace in $\Gamma(u,\pi,v)$, and removing it yields a non‑empty set of graph operations.

\paragraph{Qualified satisfaction tasks.}
Let $\texttt{SAT}_\pi^n(sel,\varphi)$ be an eligible qualified satisfaction task. By the definition of eligibility (Def.~\ref{def:eligibility}), there is at least one node $u\models sel$ and $\varphi$ is not a contradiction. The path construction procedure outlined in App.~\ref{app:pathbuild} recursively constructs a trace $\tau$ matching $\pi$. Since eligibility rules out cases where there does not exist a node $v$ that can satisfy $\varphi$. Connecting $u$ with a node $v$ that satisfies $\varphi$ ($v\models\varphi$) yields a non-empty set of graph operations.
\end{pf}
\section{Details on Path Violation Tasks}
\label{app:pathbreak}
This appendix complements Sec.~\ref{sec:path_vio} by detailing how to remove all $\pi$-matching traces for the remaining path expression forms: $(\pi_1^-)$, $(\pi_1 \cdot \pi_2)$, $(\pi_1^*)$, and $(\pi_1?)$. The procedure follows the same principle as outlined in Sec.~\ref{sec:path_vio}: for each trace $\tau$ that matches $\pi$ with $src(\tau)=u$ and $tgt(v)$, identify an essential hop matching $\pi$ and remove it. Repeat until no $\pi$-matching trace remains.

As noted in Sec.~\ref{sec:coverage}, some cases are \emph{operationally infeasible} (ineligible). This occurs when 
there exists a zero-length trace $\langle u\rangle$ matching $\pi$. In the following, we assume eligible tasks.

\begin{itemize}[leftmargin=*]
    \item \textbf{Base case, label $q$:}
    A trace from $u$ to $v$ matches $q$ if
    $(u,v)\in\llbracket q\rrbracket^{\mathcal G}$.
    If $q\in\mathcal P$, this means that
    $(u,q,v)\in\mathcal E$, and the trace is broken by applying
    $-e(u,q,v)$. If $q\in\mathcal K$, this means that
    $\rho(u,q)=v$, and the trace is broken by applying $-a(u,q)$.

    \item \textbf{Alternative $(\pi_1 \cup \pi_2)$: } a trace $\tau$ matches $\pi_1 \cup \pi_2$ if it matches either $\pi_1$ or $\pi_2$. To ensure no match exists, we must eliminate both possibilities. Thus, the algorithm recursively calls itself on $(u,\pi_1,v)$ and $(u,\pi_2,v)$.
    \item \textbf{Inverse $(\pi_1^-)$}: a trace $\tau$ matches $\pi_1^-$ from nodes $u$ to $v$ if a trace matching $\pi_1$ exists from $v$ to $u$. To violate this, the algorithm recursively calls itself on the inverted case $(v,\pi_1,u)$.
    \item \textbf{Sequence $(\pi_1 \cdot \pi_2)$}: a trace $\tau$ matches $(\pi_1 \cdot \pi_2)$ from nodes $u$ to $v$ if there exists some node $w$ such that a $\pi_1$-matching sub-trace exists from $u$ to $w$ and a $\pi_2$-matching sub-trace exists from $w$ to $v$. To violates this, we must invalidate at least one of these two sub-traces. Therefore, the algorithm recursively calls itself on either $(u,\pi_1,w)$ or $(w,\pi_2,v)$.
    \item \textbf{Zero-or-more repetitions $(\pi_1^*)$}: a non-zero-length trace $\tau$ matches $\pi_1^*$ from nodes $u$ to $v$ if there exists nodes $w_1, w_2, \cdots$, such that  $(u,\pi_1,w_1,\pi_1,\cdots,\pi_1,v)$. It is sufficient to invalidate the first segment $(u,\pi_1,w_1)$ to invalidate the entire trace. Therefore, the algorithm recursively calls itself on $(u,\pi_1,w_1)$.
    
    \item \textbf{Zero-or-one $(\pi_1?)$}: this is a special case of $\pi_1^*$. A non-zero-length trace $\tau$ matches $(\pi_1?)$ from nodes $u$ to $v$ if a trace matching $\pi_1$ exists from $u$ to $v$. To violate this, the algorithm recursively calls itself on $(u,\pi_1,v)$.
\end{itemize}

\paragraph{Termination.}  
In all eligible cases, the process eventually reaches the \emph{base case} $q$, where a trace $\tau$ matches $q$ because $(u,q,v)$ is present. Removing this hop via $-e(u,q,v)$ (for $q \in \mathcal{E}$) or $-a(u,q)$ (for $q \in \mathcal{K}$) invalidates the trace.

\section{Details on Qualified Satisfaction Tasks}\label{app:qMCount}

This appendix elaborates on the procedure for $\texttt{SAT}_\pi^n(sel,\varphi)$ tasks described in Sec.~\ref{sec:qMCount}.  
The objective is to ensure that a node $u \models sel$ has exactly $n$ distinct qualified $\pi$-targets that satisfy $\varphi$. Let  \(m \)  denote the current number of such targets. The goal is to add $k = n - m$ new targets that satisfy $\varphi$. This is achieved in two steps: First, find qualified targets. Second, construct $\pi$-matching traces. The details are as follows.

Candidate qualified targets are located using the \textit{class-aware subgraph monomorphism search} proposed by Fierro et al.~\cite{fierro2022application}.  
The shape $\varphi$ is first converted into its graph representation $\mathcal{G}_\varphi$.  
A search is then performed over $\mathcal{G}$ to retrieve the largest subgraphs that are monomorphic to subgraphs of $\mathcal{G}_\varphi$.  
Each retrieved subgraph corresponds to a portion of $\mathcal{G}$ that satisfies either $\varphi$ in full or a proper subshape of $\varphi$.

When a retrieved subgraph $\bar{\mathcal{G}}(\bar{\mathcal{E}}, \bar{\rho}) \subseteq \mathcal{G}(\mathcal{E}, \rho)$ is fully monomorphic to $\mathcal{G}_\varphi$, there exists a node $v \in \bar{\mathcal{G}}$ such that $v \models \varphi$.  
In this case, a trace $\tau$ matching $\pi$ ($\tau\models\pi$) is constructed, with $src(\tau) = u$ and $tgt(\tau) = v$, and added to $\mathcal{G}$.

If the retrieved subgraph is monomorphic only to a proper subgraph of $\mathcal{G}_\varphi$, the corresponding node $v$ satisfies a proper subshape of $\varphi$.  
To make $v$ satisfy $\varphi$, the missing entities and relationships are generated using an LLM. This is because symbolic random string generators will create entities that are too dissimilar from the existing entities (e.g., \texttt{``new\_0585''}). Once $v$ satisfies $\varphi$, a $\pi$-matching trace from $u$ to $v$ is created as in the full-match case.

This process is repeated until $k$ distinct $\pi$-targets satisfying $\varphi$ have been obtained. 

In this way, the subgraph monomorphism search maximizes reuse of existing graph elements, resorting to synthetic entity creation only when necessary, thereby satisfying the $\texttt{SAT}_\pi^n(sel,\varphi)$ requirement while minimizing changes to the original graph.

\noindent\textbf{Running Example.} Consider the task
$$\texttt{SAT$_\texttt{reviewedBy}^4$($\in$\texttt{Paper},$\geq_1\pi.\text{len}\geq 3$)}$$ from our running example in Fig.~\ref{fig:tree}. Let $ \varphi=(\geq_1\pi.\text{len}\geq 3)$, a node of type \texttt{Paper} is required to have exactly four distinct \texttt{reviewedBy} targets that satisfy $\varphi$.

From the set $\{\texttt{PaperABC},\ \texttt{PaperA}\}$, both of which satisfy $\in\texttt{Paper}$, we randomly select \texttt{PaperABC} as the subject node $u$. The current set of qualified $\pi$-targets for $u$ is $\{\texttt{Ada},\ \texttt{Bob}\}$, giving $m = 2$.
Since $n = 4$, we require $k = 2$ additional targets.
Applying the subgraph monomorphism search to $\mathcal{G}$, we retrieve a node \texttt{Dan} that satisfies $\varphi$ and is not already in the target set. After adding a \texttt{reviewedBy} edge from \texttt{PaperABC} to \texttt{Dan} is added, the remaining requirement is reduced to $k = 1$.
The search then identifies \texttt{Cole}, which matches only a proper subshape of $\varphi$.
To make \texttt{Cole} satisfy $\varphi$, we prompt an LLM as illustrated in Fig.~\ref{fig:qual_max_prompt} to generate the missing entities.
Once \texttt{Cole} satisfies $\varphi$, a \texttt{reviewedBy} edge from \texttt{PaperABC} to \texttt{Cole} is added, completing the set of four qualified targets. App.~\ref{app:pathbuild} details how to create a $\pi$-matching trace connecting two nodes.

\begin{figure}[t!]
    \captionsetup{aboveskip=2pt}
    \setlength{\belowcaptionskip}{-1pt}    
    \centering{\small
    \scalebox{1}{\begin{myboxNoTitle}
    
    Generate an entity name to replace \texttt{urn:\_\_\_param\_\_\_\#p1} in the following graph. 
    
    \texttt{ex:Cole workPhone urn:\_\_\_param\_\_\_\#p1.}
        
    Your answer must be semantically similar to the corresponding entity in the following example but not exactly identical.
    
    \texttt{ex:Dan workPhone ``789''.}
    
    Compare with the example, observe if \texttt{urn:\_\_\_param\_\_\_\#p1} should be a URIRef. If so, make it a valid URIRef. Otherwise, make it a Literal.
    Return your answer in json without explanations. \{``answer'':your answer, ``is URIRef'':true/false\}.
    \end{myboxNoTitle}}}
    \mycaption{Example prompt for qualified satisfaction task.}{}
    \label{fig:qual_max_prompt}
\end{figure}

\section{Path Construction Tasks}\label{app:pathbuild}

This appendix describes the procedure for constructing a $\pi$-matching trace between two nodes $u$ and $v$. This procedure is used in path expression-based atomic shape violation tasks (App.~\ref{app:path_ops}) and qualified satisfaction tasks (App.~\ref{app:qMCount}).

The construction process is a recursion on the structure of $\pi$. Given two endpoints, nodes $u$ and $v$, and a path expression $\pi$, the algorithm determines the structural requirements for a match and adds the necessary edges or attributes. The recursion terminates when a \emph{base case} is reached:

\begin{itemize}[leftmargin=*]
    \item \textbf{Base case, label $q$}:
    A trace from node $u$ to $v$ matches $q$ if $(u,q,v)$ is present.  
    If $q \in \mathcal{E}$, add the edge $+e(u,q,v)$; if $q \in \mathcal{K}$, add the attribute $+a(u,q)$.
    \item \textbf{Alternative $(\pi_1 \cup \pi_2)$}: A trace from node $u$ to $v$ matches $(\pi_1 \cup \pi_2)$ if it matches either $\pi_1$ or $\pi_2$. To construct a trace matching this path expression, we can choose either possibility. Thus, the algorithm recursively calls itself on $(u,\pi_1,v)$ or $(u,\pi_1,v)$.
    \item \textbf{Sequence $(\pi_1 \cdot \pi_2)$}: A trace from node 
    $u$ to $v$ matches $(\pi_1 \cdot \pi_2)$ if there exists some node $w$ such that a $\pi_1$-matching trace exists from $u$ to $w$ and a $\pi_2$-matching trace exists from $w$ to $v$. To construct a trace matching this path expression, we must select an intermediate node $w$, then the algorithm recursively calls itself on $(u,\pi_1,w)$ and $(w,\pi_2,v)$.
    \item \textbf{Inverse $(\pi_1^-)$}:  A trace from node $u$ to $v$ matches $(\pi_1^-)$ if the trace from $v$ to $u$ matches $\pi$. To construct a trace matching this path expression, the algorithm recursively calls itself on $(v,\pi_1,u)$.
    \item \textbf{Zero-or-more repetitions $(\pi_1^*)$}: A trace from node $u$ to $v$ matches $(\pi_1^*)$ trivially via the zero-length trace $\langle u\rangle$ if $u=v$. If $u\neq v$, a trace from node $u$ to $v$ matches $(\pi_1^*)$ if there exists nodes  $w_1, w_2, \cdots$, such that  $(u,\pi_1,w_1,\pi_1,\cdots,\pi_1,v)$. To construct a trace matching this path expression, we select a set of intermediate nodes $w_1, w_2, \cdots, w_m$, and the algorithm recursively calls itself on each consecutive pair $(u,\pi_1,w_1),(w_m,\pi_1, v)$, and $(w_i,\pi_1,w_{i+1}), \forall i\in\{1,\cdots,m-1\}$.
    \item \textbf{Zero-or-one $(\pi_1?)$}: This is a special case of $\pi_1^*$. A trace from node $u$ to $v$ matches $(\pi_1?)$ trivially via the zero-length trace $\langle u\rangle$ if $u=v$. If $u\neq v$, The algorithm recursively calls itself on $(u,\pi_1,v)$. 
\end{itemize}

\paragraph{Termination.}  
The process eventually reaches the base case $q$, where the match can be realized by adding $(u,q,v)$ via $+e(u,q,v)$ (for $q \in \mathcal{E}$) or $+a(u,q)$ (for $q \in \mathcal{K}$).

\section{LLM-based Repair System Details}\label{app:llm_repair}
\begin{figure*}[t]
    \captionsetup{aboveskip=2pt}
    \setlength{\belowcaptionskip}{-1pt}    
    \centering
    \includegraphics[width=1\linewidth]{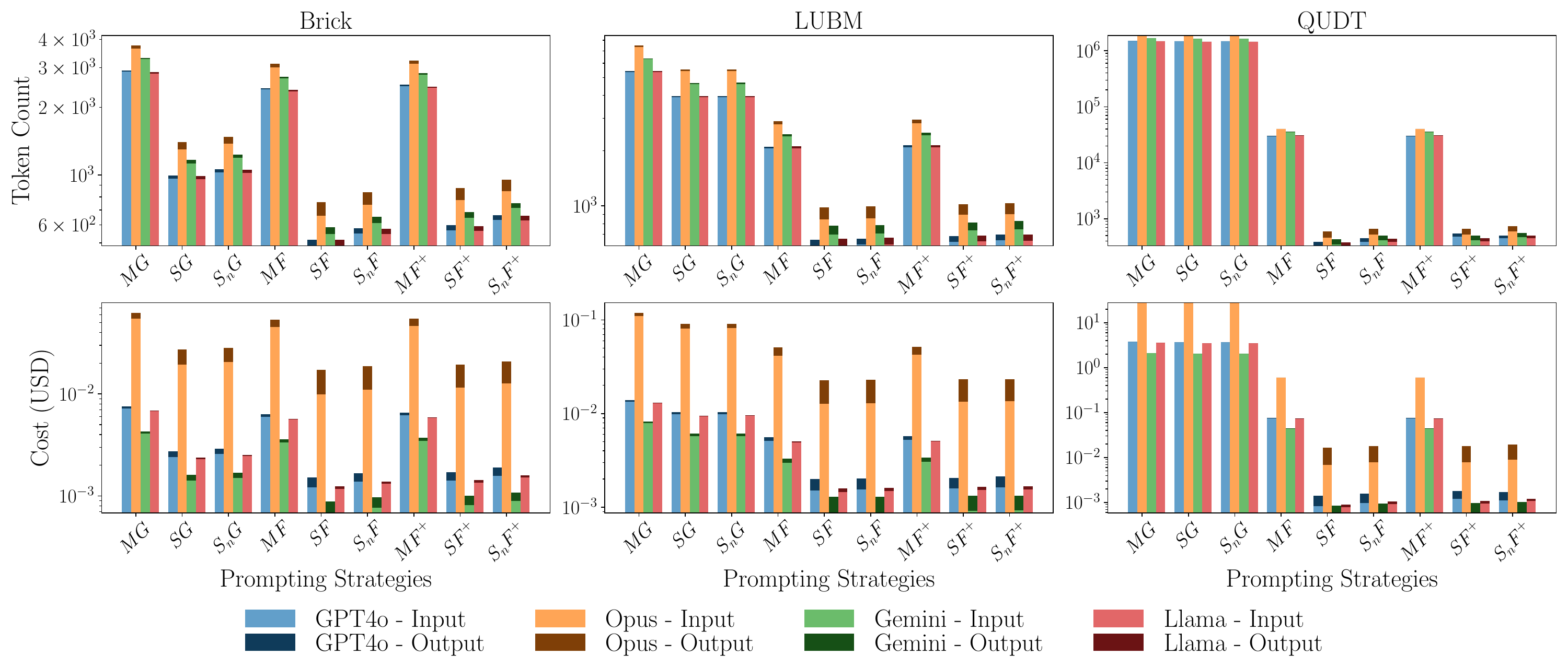}
    \caption{\textbf{Average token count and cost for different prompting strategies and LLMs.} }
    \label{fig:price}
\end{figure*}

This appendix provides more details on the implementation and results of LLM-based repair. The LLMs evaluated are GPT4o (OpenAI), Claude 3.0 Opus (Anthropic), Gemini 1.5 Pro (Google), and Llama 3.1 405B (Meta, via AWS Bedrock). All models were accessed in January 2025. The pricing structures are listed in Table \ref{tab:llm_price}. 

As stated in Sec.~\ref{sec:repairsys}, we use a modular prompt template
(Fig.~\ref{fig:prompt}) to perform an ablation study. There are three configurations for both the schema context $\mathcal{C}_m$ and the graph context $\mathcal{C}_g$. They are listed here for convenience.

\noindent \textit{Schema Context $\mathcal{C}_m$.} 
\begin{enumerate}[left=0pt, align=left]
    \item \textbf{Entire Schema ($M$)}: Provides the full schema $\mathcal{S}$. 
    \item \textbf{Minimal Schema ($S$)}: Provides only the offended constraint $(sel,\varphi)$. 
    \item \textbf{Annotated Schema ($S_n$)}: Extends $S$ with natural language descriptions (e.g., \texttt{rdfs:label}). 
\end{enumerate}

\noindent \textit{Graph Context $\mathcal{C}_g$.} 
\begin{enumerate}[left=0pt, align=left]
    \item \textbf{Entire Graph ($G$)}: Provides the full graph $\mathcal{G}'$. 
    \item \textbf{Focused Subgraph ($F$)}: Provides only the offending subgraph $\mathsf{OffSub}_{\mathcal{G}'}(u, (sel,\varphi))$ of offending focus $u$ and offended constraint $(sel,\varphi)$. For cardinality constraints, $\varphi = \geq_n \pi.\varphi'$, we also include context around nodes that satisfy $\varphi'$.
    \item \textbf{Focused Subgraph with Example ($F^+$)}: Augments $F$ with a positive example. 
\end{enumerate}

Fig.~\ref{fig:price} reports average input/output token counts and associated costs. Due to budget constraints, prompts that cost more than \$0.50 are omitted. These exclusions only occur on the large-scale QUDT dataset and affect the following prompting strategies:
\begin{itemize}[left=0pt, align=left]
    \item The $MG,SG$, and $S_nG$ prompts for all LLMs.
    \item The $MF$ and $MF^+$ prompts for Claude 3.0 Opus, whose granular tokenization results in higher costs.
\end{itemize} 
Therefore, the output token counts and costs are omitted in such cases in Fig.~\ref{fig:price}. As expected, using the entire schema and graph ($MG$) yields the highest input token counts and costs across all datasets. Adding natural language descriptions ($S_n$ vs.\ $S$) increases input size only slightly for LUBM, due to the fact that such annotations are sparse in the schema. With the exact same prompt, Claude 3.0 Opus shows the highest counts and costs among LLMs, reflecting its more granular tokenization.
\begin{figure}[t!]
    \captionsetup{aboveskip=2pt}
    \setlength{\belowcaptionskip}{-1pt}    
    \centering{
    \scalebox{1}{\begin{mybox}{Primer}
    \small You are an expert in repairing RDF graphs that violate SHACL shapes. Output a SPARQL operation that fixes the violation.
    \end{mybox}}
    \scalebox{1}{\begin{mybox}{Schema Context}
    \small Focus Node: \texttt{<urn:bldg/ahu> }\\
    Violated source SHACL shape:\\
    \texttt{[] sh:path brick:feeds; \\ 
    \phantom{[]} sh:qualifiedMinCount 1;\\
    \phantom{[]} sh:qualifiedValueShape [ sh:node <urn:co2\_zone> ].}\\
    Reason: qualifiedMinCount is violated\\
    \texttt{<urn:co2\_zone> a sh:NodeShape;\\
    \phantom{[]}    sh:property [ sh:path brick:hasPoint; ... ]}
    \end{mybox}}
    
    \scalebox{1}{\begin{mybox}{Graph Context}
    \small RDF knowledge graph:\\
    \texttt{<urn:bldg/ahu> a brick:Air\_Handling\_Unit;\\
    \hspace*{1em}    brick:feeds <urn:bldg/zone1> ;\\
    \hspace*{1em}    brick:hasPoint <urn:bldg/outside\_co2\_sensor>. ... }
    \end{mybox}}
    
    \scalebox{1}{\begin{mybox}{Instructions}
    \small Use ``INSERT DATA \{ \}'', ``DELETE DATA \{ \}'', or ``DELETE \{ \} INSERT \{ \} WHERE \{ \}'' to fix the violation with minimal change that is contextually appropriate. Invent placeholder names or remove existing instances only if necessary. Do not use nested curly brackets.
    
    Respond only with valid json format using key ``answer'' without explanations. For example, \{``answer'':``INSERT DATA {...}''\} or \{``answer'':``DELETE DATA \{...\}''\} or \{``answer'':``DELETE \{...\} INSERT \{...\} WHERE \{...\}''\}.
    
    \end{mybox}}}
    \caption{\textbf{Prompt template for our LLM-based repair systems}}
    \label{fig:prompt}
\end{figure}

\begin{table}[t]
    \captionsetup{aboveskip=2pt}
    \setlength{\belowcaptionskip}{-1pt}    
\centering
\caption{\textbf{Pricing Model.} Costs are in USD per 1 million tokens.}
\begin{tabular}{ccc}
\midrule
\textbf{Model} & \textbf{Input Cost} & \textbf{Output Cost} \\\hline
GPT4o & $\$2.5$ & $\$10$ \\
Claude 3.0 Opus & $\$15$ & $\$75$ \\
Gemini 1.5 Pro & $\$1.25$ & $\$5$ \\
Llama 3.1 405B & $\$2.4$ & $\$2.4$ \\\midrule
\end{tabular}
\label{tab:llm_price}
\end{table}

Subsequently, in Fig.~\ref{fig:heatmap}, we present the repair quality of each combination of prompting strategy and LLM, evaluated using the evaluation metrics outlined in Sec.~\ref{sec:metrics}. The gray entries are prompts that exceed our \$0.5 budget. Finally, we present the full results for all repair methods (answer set programming based, supervised learning based, and LLM based) in Table~\ref{tab:full_no_exclusion}.

\begin{figure*}[h]
    \captionsetup{aboveskip=2pt}
    \setlength{\belowcaptionskip}{-1pt}    
    \centering
    \includegraphics[width=0.8\linewidth]{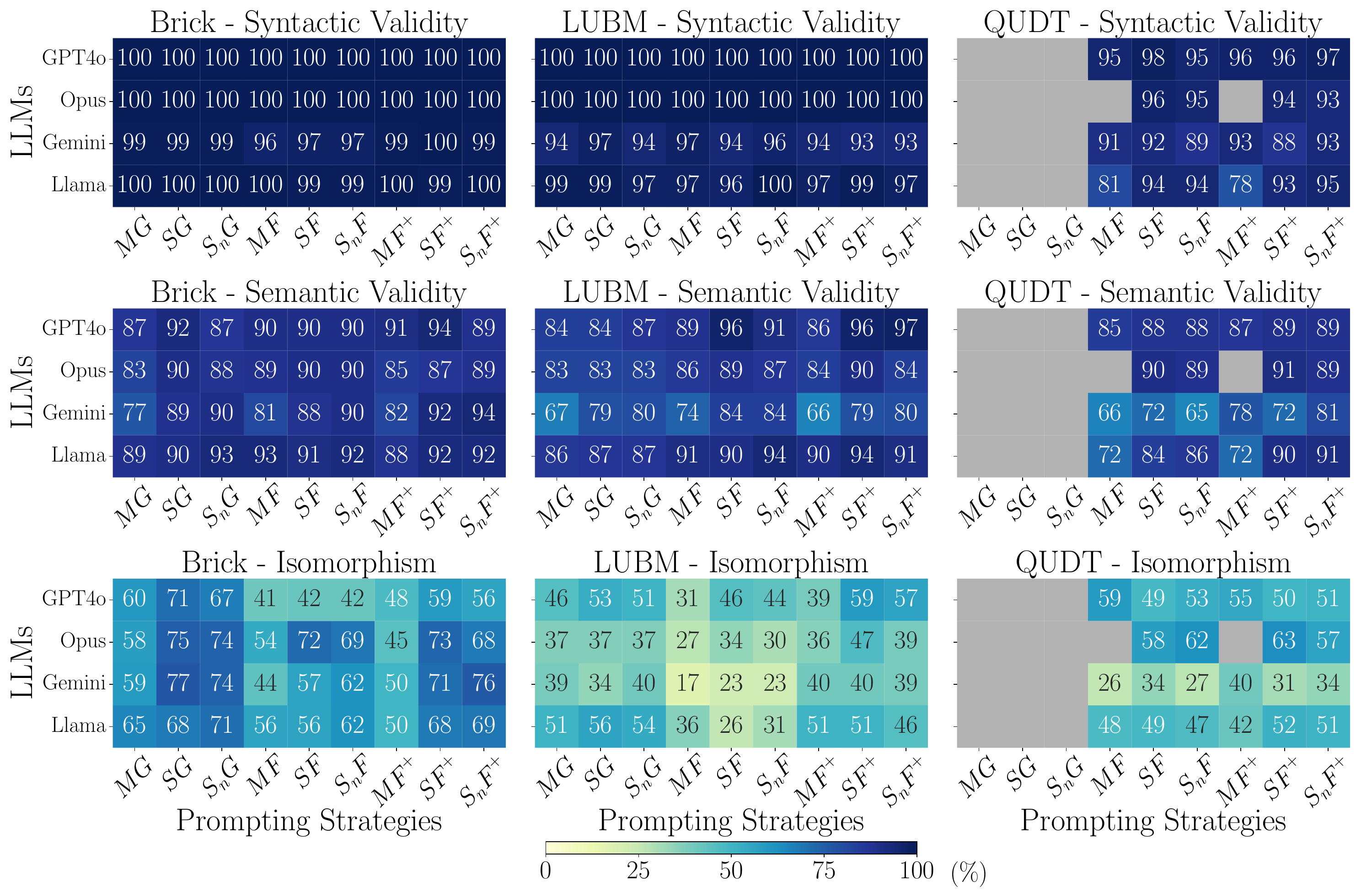}
    \caption{\textbf{Performance of each evaluation metric for every combination of LLM, strategy, and dataset.}}
    \label{fig:heatmap}
\end{figure*}

\begin{table}[t]
    \captionsetup{aboveskip=2pt}
    \setlength{\belowcaptionskip}{-1pt}    
\centering
{
\caption{\textbf{Average scores ($\%$) of all repair systems on all datasets.}}
\label{tab:full_no_exclusion}
\begin{tabularx}{\linewidth}{cYYYYY}
\toprule
    Method & Syn. Val. & Sem. Val. &  Iso. & Iso. Hit & Iso. Prec. \\\midrule    ASP~\cite{ahmetaj2022repairing}  &\textbf{100.0 } & 40.23 &  13.51 & 27.59 & 12.13\\
    Neural~\cite{pellissier2021neural} & \textbf{100.0}& 10.37 & 8.36 &- &-\\
    Subset Avg &96.96 & 87.95& 52.41 &- &- \\
    Llama$+SF^+$ &96.84 & \textbf{91.38}&  \textbf{58.62} &- &- \\\bottomrule
\end{tabularx}}
\end{table}

\section{Details on the Statistical Analysis}\label{app:stats}
In the statistical analysis in Sec.~\ref{sec:impact}, we seek to determine whether using different contexts or LLMs produces different results. In the following, we take the example of investigating whether using the schema context $M$ differs from $S$ to illustrate.

Our data are structured such that multiple measurements are taken across various datasets, LLMs, and evaluation metrics. This leads to repeated measures and correlated observations within each grouping factor (e.g., measurements from the same LLM are not independent). Standard linear models assume that all observations are independent~\cite{zdaniuk2024ordinary}, which is violated in this setting. Linear mixed-effects models (LMMs) are designed to address this issue: they allow us to model both the systematic effects of our main variables of interest (fixed effects) and the random variability due to repeated measurements within groups (random effects). This provides more accurate estimates and valid statistical inference in the presence of hierarchical data.

A linear mixed-effects model combines two types of effects:
\begin{enumerate}[left=0pt, align=left]
    \item Fixed effects capture the average influence of variables of primary interest that are consistent across all observations. In our example, the fixed effect is the schema context $M$ vs. $S$, where $M$ is the reference and $S$ is the comparison.
    \item Random effects account for random variation attributable to grouping factors, such as datasets, LLMs, and metrics. These effects capture the correlation and heterogeneity within groups.
\end{enumerate}

Mathematically, for an outcome $y_{ijkl}$ measured for schema context $i$, dataset $j$, LLM $k$, and metric $l$, the model can be written as:

\[
y_{ijkl} = \beta_0 + \beta_1 \cdot \text{Context}_{i} + b_j^{(\text{dataset})} + b_k^{(\text{LLM})} + b_l^{(\text{metric})} + \epsilon_{ijkl}, 
\]
where
\begin{itemize}[left=0pt, align=left]
    \item $\beta_0$ is the intercept (mean outcome for the reference. In our example, $M$ is considered the reference that $S$ will compare with),
    \item $\beta_1$ is the fixed-effect coefficient for the context (difference between $M$ and $S$),
    \item $b_j^{(\text{dataset})} \sim \mathcal{N}(0, \sigma^2_{\text{dataset}})$ is the random intercept for dataset $j$,
    \item $b_k^{(\text{LLM})} \sim \mathcal{N}(0, \sigma^2_{\text{LLM}})$ is the random intercept for LLM $k$,
    \item $b_l^{(\text{metric})} \sim \mathcal{N}(0, \sigma^2_{\text{metric}})$ is the random intercept for metric $l$,
    \item $\epsilon_{ijkl} \sim \mathcal{N}(0, \sigma^2)$ is the residual error.
\end{itemize}
This formulation enables us to estimate the effect of changing the schema context from $M$ to $S$ while controlling for the repeated measures and correlation introduced by datasets, LLMs, and metrics.

The effect of changing the schema context from $M$ to $S$ is given by the estimated fixed-effect coefficient, $\hat{\beta}_1$. This value represents the average difference in the outcome variable when using $S$ instead of $M$, after accounting for variability due to datasets, LLMs, and metrics. The estimation of the confidence interval for this effect quantifies the uncertainty of the estimate. Its calculation can be found in Pinheiro et al.~\cite{pinheiro2000linear}.

A key assumption of linear mixed-effects models is that the residual errors ($\epsilon_{ijkl}$) are normally distributed with mean zero and constant variance. This normality assumption underpins the validity of statistical inference (such as confidence intervals and $p$-values). To test for normality, we employ the Shapiro-Wilk (SW) test. The null hypothesis of the SW test assumes that the data follows a normal distribution. We set the significance level for the SW test at $0.05$. Therefore, if the $p$-value is greater than $0.05$, the residuals are considered to be normally distributed, satisfying the normality assumption of the LMM. If the residuals do not pass the normality test, we apply various transformations to the raw data and refit the LMM. The transformations considered are:

\begin{enumerate}[left=0pt, align=left]
    \item Arcsine Square Root Transformation: $\arcsin{\sqrt{x}}$.
    \item Logit Transformation: $\ln{\frac{x+\epsilon}{1-p+\epsilon}}$, where $\epsilon=10^{-6}$ is used to avoid numerical error.
    \item Box-Cox Transformation: $\begin{cases}
            \frac{y^\lambda -1}{y}, \text{if }\lambda\neq0\\
            \ln{y},\text{  else}
        \end{cases}$. Here, $\lambda$ is optimized to best approximate a normal distribution.
\end{enumerate}
If none of the transformations results in normal residuals after refitting the model, we proceed without any transformation and report the SW $p$-values directly. For any transformation applied, we back-transform the estimated effects to their original scale.

In Tables \ref{tab:prompt_selection} and \ref{tab:llm_selection}, we present the raw averages of the reference group and the comparison group. The critical value for both the estimated effect and the Shapiro-Wilk test is set to $0.05$. If the $p$-value of the estimated effect is less than $0.05$, it indicates a significant difference between the reference group and the comparison group; otherwise, the cell in the table is shaded with \textcolor{C6!20}{\rule{3ex}{1.8ex}}. Similarly, if the $p$-value of the SW test is greater than $0.05$, indicating that the residuals are normally distributed, the cell is not shaded; otherwise, it is shaded with \textcolor{C3!20}{\rule{3ex}{1.8ex}}.

\clearpage

\hvFloat[
  floatPos=p,
  rotAngle=90,
  objectPos=center,
  capPos=top,
  capWidth=2.2,
  capVPos=center
]{table}{%
  \centering
  \begin{tabular}{lcccccccccccc}\toprule
&Treatment&\shortstack{Reference\\Group} &\shortstack{Comparison\\Group} &\shortstack{Reference\\Group Avg} &\shortstack{Comparison\\Group Avg} &\shortstack{Estimated\\Effect}&$p$-value  &\shortstack{Transf-\\ormation} &\shortstack{SW test\\$p$-value} \\\midrule
\multirow{6}{*}{Overall} 
&$M\to S$&$MG$, $MF$, $MF^+$ &$SG$, $SF$, $SF^+$ &68.31\% &73.71\% &5.28\% &$4.4\times10^{-10}$  &- &$1.8\times10^{-1}$  \\
&$S\to S_n$&$SG$, $SF$, $SF^+$ &$S_nG$, $S_nF$, $S_nF^+$ &73.71\% &73.41\% &-0.30\% &\cellcolor{C6!20}$6.7\times10^{-1}$  &- &\cellcolor{C3!20}$3.8\times10^{-7}$  \\
&$M\to S_n$&$MG$, $MF$, $MF^+$ &$S_nG$, $S_nF$, $S_nF^+$ &68.31\% &73.41\% &4.98\% &$4.8\times 10^{-9}$  &arcsine &$5.0\times10^{-2}$  \\
&$G\to F$&$MG$, $SG$, $S_nG$ &$MF$, $SF$, $S_nF$ &74.93\% &69.37\% &-3.97\% &$2.4\times 10^{-5}$  &arcsine &$5.6\times10^{-2}$  \\
&$F\to F^+$&$MF$, $SF$, $S_nF$ &$MF^+$, $SF^+$, $S_nF^+$ &69.37\% &72.31\% &2.94\% &$1.5\times 10^{-4}$  &- &$8.4\times10^{-2}$  \\
&$G\to F^+$&$MG$, $SG$, $S_nG$ &$MF^+$, $SF^+$, $S_nF^+$ &74.93\% &72.31\% &-0.86\% &\cellcolor{C6!20}$3.3\times 10^{-1}$  &- &\cellcolor{C3!20}$4.4\times10^{-6}$  \\\hline
\multirow{6}{*}{\shortstack{Syntactic\\Validity}} 
&$M\to S$&$MG$, $MF$, $MF^+$ &$SG$, $SF$, $SF^+$ &96.87\% &97.59\% &0.16\% &$4.5\times 10^{-1}$  &box-cox &$6.3\times10^{-2}$  \\
&$S\to S_n$&$SG$, $SF$, $SF^+$ &$S_nG$, $S_nF$, $S_nF^+$ &97.59\% &97.56\% &-0.03\% &\cellcolor{C6!20}$9.1\times 10^{-1}$  &- &\cellcolor{C3!20}$1.5\times10^{-2}$  \\
&$M\to S_n$&$MG$, $MF$, $MF^+$ &$S_nG$, $S_nF$, $S_nF^+$ &96.87\% &97.56\% &0.96\% &\cellcolor{C6!20}$1.1\times 10^{-1}$  &- &\cellcolor{C3!20}$3.5\times10^{-8}$  \\
&$G\to F$&$MG$, $SG$, $S_nG$ &$MF$, $SF$, $S_nF$ &99.04\% &96.80\% &-0.42\% &\cellcolor{C6!20}$4.3\times 10^{-1}$  &- &\cellcolor{C3!20}$2.2\times10^{-9}$  \\
&$F\to F^+$&$MF$, $SF$, $S_nF$ &$MF^+$, $SF^+$, $S_nF^+$ &96.80\% &96.74\% &-0.06\% &\cellcolor{C6!20}$9.2\times 10^{-1}$  &- &\cellcolor{C3!20}$1.0\times10^{-9}$  \\
&$G\to F^+$&$MG$, $SG$, $S_nG$ &$MF^+$, $SF^+$, $S_nF^+$ &99.04\% &96.74\% &-0.35\% &\cellcolor{C6!20}$5.7\times 10^{-1}$  &- &\cellcolor{C3!20}$2.0\times10^{-11}$  \\\hline
\multirow{6}{*}{\shortstack{Semantic\\Validity}} 
&$M\to S$&$MG$, $MF$, $MF^+$ &$SG$, $SF$, $SF^+$ &82.70\% &87.87\% &5.04\% &$2.4\times 10^{-7}$  &- &$5.8\times10^{-1}$  \\
&$S\to S_n$&$SG$, $SF$, $SF^+$ &$S_nG$, $S_nF$, $S_nF^+$ &87.87\% &87.72\% &-0.16\% &\cellcolor{C6!20}$8.4\times 10^{-1}$  &- &$2.8\times10^{-1}$  \\
&$M\to S_n$&$MG$, $MF$, $MF^+$ &$S_nG$, $S_nF$, $S_nF^+$ &82.70\% &87.72\% &4.97\% &$4.0\times10^{-6}$  &- &$3.6\times10^{-1}$  \\
&$G\to F$&$MG$, $SG$, $S_nG$ &$MF$, $SF$, $S_nF$ &85.21\% &86.11\% &3.59\% &$5.8\times10^{-4}$  &arcsine &$9.8\times10^{-2}$  \\
&$F\to F^+$&$MF$, $SF$, $S_nF$ &$MF^+$, $SF^+$, $S_nF^+$ &86.11\% &86.89\% &0.90\% &\cellcolor{C6!20}$3.4\times10^{-1}$  &logit &$5.7\times10^{-1}$  \\
&$G\to F^+$&$MG$, $SG$, $S_nG$ &$MF^+$, $SF^+$, $S_nF^+$ &85.21\% &86.89\% &3.41\% &$1.3\times10^{-2}$  &logit &$6.2\times10^{-1}$  \\\hline
\multirow{6}{*}{Isomorphism} 
&$M\to S$&$MG$, $MF$, $MF^+$ &$SG$, $SF$, $SF^+$ &45.00\% &52.53\% &7.34\% &$1.4\times 10^{-3}$  &- &$1.4\times10^{-1}$  \\
&$S\to S_n$&$SG$, $SF$, $SF^+$ &$S_nG$, $S_nF$, $S_nF^+$ &52.53\% &51.97\% &-0.56\% &\cellcolor{C6!20}$7.6\times 10^{-1}$  &- &$6.4\times10^{-2}$  \\
&$M\to S_n$&$MG$, $MF$, $MF^+$ &$S_nG$, $S_nF$, $S_nF^+$ &45.00\% &51.97\% &7.28\% &$1.7\times 10^{-3}$  &logit &$6.9\times10^{-2}$  \\
&$G\to F$&$MG$, $SG$, $S_nG$ &$MF$, $SF$, $S_nF$ &56.42\% &43.91\% &-13.67\% &$3.4\times 10^{-14}$  &- &$7.8\times10^{-1}$  \\
&$F\to F^+$&$MF$, $SF$, $S_nF$ &$MF^+$, $SF^+$, $S_nF^+$ &43.91\% &51.51\% &7.60\% &$2.2\times10^{-5}$  &- &$1.7\times10^{-1}$  \\
&$G\to F^+$&$MG$, $SG$, $S_nG$ &$MF^+$, $SF^+$, $S_nF^+$ &56.42\% &51.51\% &-3.27\% &\cellcolor{C6!20}$1.2\times 10^{-1}$  &- &$1.3\times10^{-1}$  \\
\bottomrule
\end{tabular}
}{%
  \textbf{Context Selection.} \textcolor{C6!20}{\rule{3ex}{1.8ex}} indicates the averages of the reference and comparison group are not considered different. \textcolor{C3!20}{\rule{3ex}{1.8ex}} indicates the residuals are not considered normal.\label{tab:prompt_selection}
}{}

\clearpage

\hvFloat[
  floatPos=p,
  rotAngle=90,
  objectPos=center,
  capPos=top,
  capWidth=2,
  capVPos=center
]{table}{%
  \centering
  \begin{tabular}{lcccccccccc}\toprule
&Treatment &\shortstack{Reference\\Group Avg} &\shortstack{Comparison\\Group Avg} &\shortstack{Estimated\\Effect}&$p$-value  &\shortstack{Transf-\\ormation} &\shortstack{SW test\\$p$-value} \\\midrule
\multirow{6}{*}{Overall} 
&GPT4o $\to$Llama 3.1 405B &73.32\% &73.52\% &0.20\% &\cellcolor{C6!20}$8.1\times 10^{-1}$ &- &\cellcolor{C3!20}$2.9\times 10^{-6}$  \\
&GPT4o $\to$Claude 3.0 Opus &73.32\% &73.47\% &-0.21\% &\cellcolor{C6!20}$2.9\times 10^{-1}$ &logit &$5.7\times 10^{-2}$  \\
&GPT4o $\to$Gemini 1.5 Pro &73.32\% &67.35\% &-8.60\% &$1.0\times 10^{-15}$  &arcsine &$1.8\times 10^{-1}$  \\
&Llama 3.1 405B $\to$Claude 3.0 Opus &73.52\% &73.47\% &-0.89\% &\cellcolor{C6!20}$2.9\times 10^{-1}$ &- &$7.3\times 10^{-1}$  \\
&Llama 3.1 405B $\to$Gemini 1.5 Pro &73.52\% &67.35\% &-6.17\% &$7.0\times 10^{-12}$  &- &$5.5\times 10^{-1}$  \\
&Claude 3.0 Opus $\to$Gemini 1.5 Pro &73.47\% &67.35\% &-5.43\% &$8.6\times 10^{-8}$ &- &$5.7\times 10^{-1}$  \\\hline
\multirow{6}{*}{\shortstack{Syntactic\\Validity}} 
&GPT4o $\to$Llama 3.1 405B &99.04\% &96.38\% &-1.33\% &$3.2\times 10^{-9}$ &box-cox &$7.2\times 10^{-2}$  \\
&GPT4o $\to$Claude 3.0 Opus &99.04\% &99.00\% &-0.35\% &\cellcolor{C6!20}$8.0\times 10^{-2}$ &- &\cellcolor{C3!20}$9.2\times 10^{-8}$  \\
&GPT4o $\to$Gemini 1.5 Pro &99.04\% &95.12\% &-3.92\% &$2.1\times 10^{-19}$  &- &$8.9\times 10^{-2}$  \\
&Llama 3.1 405B $\to$Claude 3.0 Opus &96.38\% &99.00\% &1.16\% &$1.5\times 10^{-4}$ &- &$4.1\times 10^{-1}$  \\
&Llama 3.1 405B $\to$Gemini 1.5 Pro &96.38\% &95.12\% &-1.95\% &$2.6\times 10^{-5}$  &box-cox &$4.0\times 10^{-1}$  \\
&Claude 3.0 Opus $\to$Gemini 1.5 Pro &99.00\% &95.12\% &-3.51\% &$3.2\times 10^{-15}$  &- &\cellcolor{C3!20}$2.8\times 10^{-2}$  \\\hline
\multirow{6}{*}{\shortstack{Semantic\\Validity}} &GPT4o $\to$Llama 3.1 405B &89.42\% &88.54\% &-0.77\% &\cellcolor{C6!20}$4.3\times 10^{-1}$ &arcsine &$1.7\times 10^{-1}$  \\
&GPT4o $\to$Claude 3.0 Opus &89.42\% &87.23\% &-2.30\% &$2.9\times 10^{-3}$ &- &$1.4\times 10^{-1}$  \\
&GPT4o $\to$Gemini 1.5 Pro &89.42\% &79.58\% &-9.83\% &$2.4\times 10^{-10}$  &- &$7.5\times 10^{-1}$  \\
&Llama 3.1 405B $\to$Claude 3.0 Opus &88.54\% &87.23\% &-2.53\% &$4.8\times 10^{-4}$ &- &$9.6\times 10^{-1}$  \\
&Llama 3.1 405B $\to$Gemini 1.5 Pro &88.54\% &79.58\% &-8.96\% &$7.6\times 10^{-9}$ &- &$1.9\times 10^{-1}$  \\
&Claude 3.0 Opus $\to$Gemini 1.5 Pro &87.23\% &79.58\% &-7.40\% &$1.2\times 10^{-5}$ &- &$2.1\times 10^{-1}$  \\\hline
\multirow{6}{*}{Isomorphism} 
&GPT4o $\to$Llama 3.1 405B &51.21\% &52.33\% &1.26\% &\cellcolor{C6!20}$5.4\times 10^{-1}$ &box-cox &$5.4\times 10^{-2}$  \\
&GPT4o $\to$Claude 3.0 Opus &51.21\% &52.36\% &1.50\% &\cellcolor{C6!20}$5.9\times 10^{-1}$ &- &$2.7\times 10^{-1}$  \\
&GPT4o $\to$Gemini 1.5 Pro &51.21\% &44.04\% &-7.17\% &$2.1\times 10^{-2}$ &- &$4.8\times 10^{-1}$  \\
&Llama 3.1 405B $\to$Claude 3.0 Opus &52.33\% &52.36\% &-0.19\% &\cellcolor{C6!20}$9.3\times 10^{-1}$ &- &$7.4\times 10^{-2}$  \\
&Llama 3.1 405B $\to$Gemini 1.5 Pro &52.33\% &44.04\% &-8.29\% &$1.9\times 10^{-5}$ &- &$4.2\times 10^{-1}$  \\
&Claude 3.0 Opus $\to$Gemini 1.5 Pro &52.36\% &44.04\% &-7.75\% &$2.2\times 10^{-3}$ &- &$6.7\times 10^{-1}$  \\
\bottomrule
\end{tabular}
}{%
  \textbf{LLM Selection} \textcolor{C6!20}{\rule{3ex}{1.8ex}} indicates the averages of the reference and comparison group are not considered different. \textcolor{C3!20}{\rule{3ex}{1.8ex}} indicates the residuals are not considered normal.\label{tab:llm_selection}
}{}

\clearpage

\end{document}